\documentclass[12pt,reqno]{amsart}
\usepackage{graphicx}
\usepackage{amsthm,amsmath,amsfonts,amssymb,amsxtra, color}%,appendix,bookmark}
\usepackage{cite}
\usepackage{amsaddr}
\newtheorem{theorem}{Theorem}[section]
\newtheorem{lemma}[theorem]{Lemma}

\theoremstyle{definition}

\newtheorem{remark}[theorem]{Remark}

\newtheorem{proposition}[theorem]{Proposition}

\numberwithin{equation}{section}  

\begin{document}

\title[NG modes for lattice NJL models]{Nambu--Goldstone modes in a lattice Nambu--Jona-Lasinio model with multi flavor  symmetries}

\author[Y. Goto and T. Koma]{Yukimi Goto\textsuperscript{1}${}^\dagger$ \and Tohru Koma\textsuperscript{2}}
\thanks{${}^\dagger$Email:~{\tt yukimi.goto@gakushuin.ac.jp}}
\address{$1.$ Gakushuin University,~Department of Mathematics,\\ Mejiro, Toshima-ku, Tokyo 171-8588, Japan\\
	$2.$ Gakushuin University (retired), Department of Physics,\\ Mejiro, Toshima-ku, Tokyo 171-8588, Japan}

\maketitle

\medskip

\noindent
{\bf Abstract:} 
We study a lattice Nambu--Jona-Lasinio model with $\mathrm{SU}(2)$ and $\mathrm{SU}(3)$ flavor symmetries of staggered fermions  
in the Kogut--Susskind Hamiltonian formalism. 
This type of four-fermion interactions has been widely used for describing low-energy behaviors of strongly interacting quarks 
as an effective model. In particular, we focus on the Nambu--Goldstone modes associated with the spontaneous breakdown of the flavor symmetries.
In the strong coupling regime for the interactions, we prove the following:
(i) For the spatial dimension $\nu \ge 5$,  the $\mathrm{SU}(3)$ model shows a long-range order at sufficiently low temperatures.
(ii) In the case of the $\mathrm{SU}(2)$ symmetry, there appears a long-range order in the spatial dimension $\nu \ge 3$
at sufficiently low temperatures. (iii) These results hold in the ground states as well.
(iv) In general, if a long-range order emerges in this type of models, then there appear gapless excitations 
above the sector of the infinite-volume ground states. These are nothing but the Nambu--Goldstone modes
associated with the spontaneous breakdown of the global rotational symmetry of flavors.
(v) In particular, we establish that the number of the linearly independent Nambu--Goldstone modes is equal to 
the number of the broken symmetry generators on the Hilbert space constructed from a certain symmetry-breaking
infinite-volume ground state.
\bigskip

\tableofcontents

%%%%%%%%%%%%%%%%%%%%%%%%%%%%%%%%%%%%%%%%%%%%%
\section{Introduction}
We continue here our study of the phase transition for a lattice model that is related to quantum chromodynamics (QCD).
In our previous work~\cite{GK}, by using the method of reflection positivity, 
we proved the existence of the long-range order at sufficiently low and zero temperatures in a lattice Nambu--Jona-Lasinio (NJL) model, which has an effective four-fermion interaction~\cite{Nambu, NJL, NJL2}.  
More precisely, the long-range order is associated with a charge-density wave in condensed matter language. 
Actually it yields the spontaneous breakdown of the particle-hole symmetry.  
In the staggered fermion formulation~\cite{Susskind, GK}, this implies the spontaneous mass generation for the fermions, 
which yields the breaking of the chiral symmetry. 
Prior to our previous work~\cite{GK}, the chiral symmetry breaking for staggered fermions was already proved by Salmhofer and Seiler~\cite{SS1, SS2}, however, the Nambu--Goldstone mode was not pursued.
As already mentioned in~\cite{GK}, Salmhofer and Seiler~\cite{SS1}
reduced a model of lattice gauge theory with $\mathrm{U}(N)$ symmetry 
to a NJL-type model by integrating over the gauge fields in the strong coupling limit. Their model so obtained is slightly different from our NJL model in the present paper. See also~\cite[Rem.~2.3]{GK}.

In this paper, we focus on Nambu--Goldstone modes in a lattice NJL model~\cite{ABG,Hatsuda} 
with $\mathrm{SU}(2)$ and $\mathrm{SU}(3)$ flavor symmetries of staggered fermions 
in the Kogut--Susskind Hamiltonian formalism~\cite{Susskind, KS}.
As to the details of the staggered fermions and other related models in lattice field theories, the references~\cite{Rothe, Nakamura, Seiler} may be useful to readers.
The purpose of our study is threefold: 
First, we will prove that in a strong coupling regime for the four-fermion interaction at sufficiently low and zero temperatures, 
a phase transition occurs in the $\mathrm{SU}(3)$ (resp.~$\mathrm{SU}(2)$) NJL model, where the $\mathrm{SU}(3)$ (resp.~$\mathrm{SU}(2)$) symmetry is spontaneously broken.
Second, we will establish that, in the infinite-volume limit, the NJL models exhibit an excitation mode  
which has an infinitesimally small excitation energy above the ground state.
This is nothing but the Nambu--Goldstone mode associated with the breakdown of the continuous flavor symmetry. 
As far as we know, these are the first mathematically rigorous results 
which prove the flavor symmetry breaking for lattice fermions.
Besides, we will prove that the number of the linearly independent Nambu--Goldstone modes is equal to 
the number of the broken symmetry generators.
To the best of our knowledge, it is also the first time that this equality between the two numbers has been proved to hold. 
In fact, there has been no proof about this type equality so far even for the anti-ferromagnetic Heisenberg model, 
which is the typical example~\cite{DLS}.

Taking benefit from the celebrated work by Dyson, Lieb and Simon~\cite{DLS}, 
the main idea to prove the existence of long-range order is to check 
that our model has fermion reflection positivity~\cite{JP, Koma4}, 
and we examine whether many of the techniques used in~\cite{GK, Koma2, Koma3, Koma4, Koma5} are useful here.
Indeed, since the method of the infrared bound and reflection positivity can also be applied to lattice systems 
with continuous symmetry~\cite{DLS, FILS, FSS}, it is expected that the method used in~\cite{Koma4, GK} will work 
for proving the existence of a phase transition in the $\mathrm{SU}(N)$ NJL model.
Here we should remark that, as is well-known, there is an essential difference between the basic structure of the reflection positivity in the Hamiltonian formalism and that in the Euclidean formalism~\cite{SS1,SS2, Seiler, Luscher, BS, OS}.
Of course, although the reflection positivity for staggered fermions in the Euclidean formalism\footnote{As is well known, in the Euclidean formalism, the reflection positivity is needed to make the corresponding Hamiltonian self-adjoint.} is well-known, it does not imply the reflection positivity in the corresponding Hamiltonian formalism.
Actually the two formalisms are totally different from each other on the two different lattices.
See~\cite[Rem.~2.3]{GK} for more technical and conceptual differences between the lattice Hamiltonian and Lagrangian formalisms.

Although the approach to prove the long-range order in this paper is similar to that of~\cite{GK, Koma4, JP} 
and relies on the techniques in~\cite{KLS1, KLS2},  
it is worth emphasizing that the main obstacle we will encounter comes from the $\mathrm{SU}(N)$ symmetry, particularly, in the large $N$ case.
In fact, we have not been able to prove the existence of the phase transition for the SU(3) model in the lower spatial dimensions $\nu = 3,4$, 
whereas we already proved the existence of long-range order at sufficiently low (resp.~zero) temperatures in dimension $\nu =3$ (resp.~$\nu=2$) for the NJL model without the present $\mathrm{SU}(N)$ flavor symmetry that was studied in~\cite{GK}, see Remark~\ref{rem.SU(2)} and~\ref{rem.dim} below. Thus, in the case with the higher $N$ of $\mathrm{SU}(N)$, 
there appears a difficulty. This is contrast to the antiferromagnetic Heisenberg models with higher spin $S$, 
in which case, the higher spins show an advantage for using the sum rule \cite{DLS} about the magnitude of the spin. 
Unfortunately, the present $\mathrm{SU}(N)$ models seem to lack such a sum rule because of the fermion system. 
Actually, the corresponding fermion operators to the spin operators in the Heisenberg models 
do not necessarily satisfy a sum rule \cite{Koma4}.  
Instead of the sum rule, we have to rely on the technique developed in~\cite{KLS1, KLS2}.  
More precisely, we need to estimate the two-point correlation function in the left-hand side of the inequality (\ref{eq.S1expv}) below. 
The function is related to the energy density. 
Once the lower bound  in (\ref{eq.S1expv}) is improved, the assumption on the dimensionality $\nu$ might be relaxed.

Once a continuous symmetry is broken, a massless particle called the Nambu--Goldstone boson 
is also expected to appear~\cite{Nambu2, Goldstone, GSW}.
With regard to the proof of the existence of Nambu--Goldstone modes, we should remark that Momoi studied 
the spin-wave spectrum above the ground state in the Heisenberg antiferromagnets in \cite{Momoi, Momoi2}. 
In particular, in~\cite{Momoi2}, he gave a bound on the two-point correlation that shows the Nambu--Goldstone type slow decay.
His works have stimulated one of the present authors \cite{Koma2, Koma3, Koma4}, 
who showed the existence of a gapless excited state above a symmetry breaking ground state.  
Their extensions are very useful for the present paper.

Although our model treated here does not include gauge fields, we believe that similar results can be obtained in a gauge theory under certain conditions, such as the strong coupling limit in~\cite{SS1, SS2}. Furthermore,
we should also note that the emergence of spontaneous flavor symmetry breaking is expected to be restricted to a region of the model parameters in continuum gauge theories that include a four-fermion interacting NJL model~\cite{VW}.
Actually, for the usual continuum NJL model in the Lagrangian formalism~\cite{Hatsuda}, the flavor symmetry $\mathrm{SU}(N)$ is often assumed to be preserved in the vacuum.
On the other hand, in lattice QCD, there is a region known as the `Aoki phase' where the flavor symmetry of Wilson fermions is broken~\cite{Aoki1, Aoki2}.
In particular, it was predicted in~\cite{ABG} that the parity-flavor symmetry of the NJL model with Wilson fermions is spontaneously broken, using the large $N_c$ (number of colors) expansion.

Since these observations are of importance for the physical interpretation of our results, we should make the following remarks:
In the vector-like gauge theories such as QCD, Vafa and Witten \cite{VW} showed the absence of 
the breaking of continuous flavor symmetry such as $\mathrm{SU}(N)$ with $N\ge 2$ under the assumption that 
the masses of the quarks are all strictly positive. Because of the $\mathrm{SU}(N)$ symmetry, all the masses of the quarks take the same strictly positive value. From this result, a reader might think that our results in the present paper contradict with 
Vafa--Witten result if our NJL model indeed contains an essence of  QCD theory as an effective one. However, this is not the case. Actually, the strictly positive masses stabilize a symmetric state 
because the mass term in the Hamiltonian control the state of the system as an external field. 
This is the physical meaning of Vafa--Witten result. 
In order to induce a symmetry breaking, we must choose a massless state as the initial state, 
and then apply an infinitesimally weak symmetry breaking field which has the form of the mass term in the Hamiltonian. 
In this strategy, Aoki and Gocksch \cite{AG} found numerical evidences for the existence of a phase 
in which both parity and flavor symmetries are spontaneously broken.  
However, Vafa--Witten result still makes many people believe that the flavor symmetry cannot be broken in the continuum limit of QCD-like theories.
As mentioned above, our model with the four-fermion interaction shows flavor symmetry breaking. 
Since Vafa and Witten also treated a similar effective fermion model with a four-fermion interaction and showed the absence of the flavor symmetry breaking by using the same argument as in their paper, a reader might think that our result gives a conterexample to Vafa--Witten result for a class of the lattice models. 
However, the flavor symmetry breaking is proved to occur only in a class of lattice models, and it does not necessarily imply the existence of the corresponding symmetry-breaking phase in usual continuum theory.
Therefore, if a suitable continuum limit is necessary to be taken, Vafa--Witten issue is still controversial.
While it is not clear whether or not our models have a physical continuum limit with the spontaneous breakdown of flavor symmetry, we believe that our results proved here can be useful for understanding of quantum field theories such as QCD.

The number of Nambu--Goldstone modes has been discussed in the physics literature.
Their arguments insist that the number of the linearly independent Nambu--Goldstone modes can be determined by the ground-state expectation values of the commutators between the symmetry generators\footnote{The generators $Q^{(a)}$ are ill-defined in the infinite-volume limit. However, we will always use them in finite-volume systems. Therefore, mathematical problems in the infinite-volume systems do not arise concerning the generators.} $Q^{(a)}$. 
Using the information about the commutators, the general formula for counting the number of Nambu--Goldstone modes has been given in~\cite{WM, Hidaka}. 
However it is difficult to obtain the expectation values of the commutators $ [Q^{(a)}, Q^{(b)}]$ in the ground states. In the present paper, we do not rely on the general counting rule.
Instead of the rule, we use a certain property of the staggered magnetization in a direct way, in order to prove the linear independence of the Nambu-Goldstone modes. As a result, we can count the number of the linearly independent Nambu-Goldstone modes.

The present paper is organized as follows: We give the precise definition of the SU(3) NJL model 
in Sec.~\ref{Sect.model}, and the particle-hole symmetry is discussed in Sec.~\ref{PHsymmetry}. 
Our results are given in Sec.~\ref{sect.results}. The reflection positivity and the Gaussian domination 
are proved in Sec.~\ref{sect.RP}. The existence of the long-range order is proved for non-zero temperatures in Sec.~\ref{sect.LRO1}
and for zero temperature in Sec.~\ref{LRO2}. We deal with the case of the SU(2) flavor symmetry in Sec.~\ref{Sect.SU(2)}. 
The existence of the Nambu--Goldstone modes is proved in Sec.~\ref{sect.NG}. 
Finally, Sec.~\ref{sect:NGn} deals with the problem of showing that the number of the Nambu--Goldstone modes equals the number of broken symmetry generators.
Appendices \ref{sect.SU(3)}--\ref{Appendix:eq.DuhamelEq} are devoted to the derivations of the relations associated to 
the SU(3) symmetry and of technical estimates.

%\ref{app.alg}
%\ref{Ap.U(1)}
%\ref{AppendixSymmetries}

%%%%%%%%%%%%%%%%%%%%%%%%%%%%%%%%%%%%%%%%%%%%%%%%%%
\section{Hamiltonian of a Lattice SU(3) Nambu--Jona-Lasinio model}
\label{Sect.model}

In order to describe the present model, we begin with some notation.
Let $\Lambda := \{x = (x^{(1)}, \dots, x^{(\nu)}) \in \mathbb{Z}^\nu \colon -L+1\le x^{(i)}\le  L, i=1,\dots,\nu \}$ denote 
the finite hypercubic lattice in $\mathbb{Z}^\nu$ with a positive integer $L$. 
Here, $\nu$ is the spatial dimension, and we impose the periodic boundary condition for the lattice.
Namely, we consider the $\nu$-dimensional torus as a finite lattice.
Let $\psi_1, \psi_2, \psi_3$ be three fermion operators satisfying the anti-commutation relations, 
\[
\{\psi_i(x), \psi_j^\dagger(y)\} = \delta_{x, y}\delta_{i, j}, \quad \{\psi_i(x), \psi_j(y)\}=0,
\] 
for the sites $x, y \in \Lambda$ and $i, j = 1, 2, 3$.
We write
\[\Psi(x):=
\begin{pmatrix}
	\psi_1(x)  \\
	\psi_2(x)  \\
	\psi_3(x) \\
\end{pmatrix},
\quad
\Psi^\dagger(x) :=(\psi_1^\dagger(x), \psi_2^\dagger(x), \psi_3^\dagger(x)).
\]
Let $\lambda^{(a)}$ be the $\mathrm{SU}(3)$ Gell-Mann matrices satisfying the commutation relations
\[
[\lambda^{(a)}, \lambda^{(b)}]=i\sum_{a=1}^8f_{abc}\lambda^{(c)},
\]
where $f_{abc}$ are the structure constants, each of which is totally anti-symmetric in the indices. 
See Appendix~\ref{sect.SU(3)} for details. 
Since we will deal with staggered fermions on the lattice $\Lambda$, 
the three types of the fermion operators $\psi_1, \psi_2, \psi_3$ can be interpreted as the three generations of the quarks, 
e.g., when $\nu =3$, the set of the operators ($\psi_1, \psi_2, \psi_3$) yields the three generations 
((u, d), (c, s), (t, b)) in the standard notation.
Each of the generations has two fermions (called `tastes') in the case of the three spatial dimensions $\nu=3$ 
because of the staggered fermions \cite{KS,Susskind}. 
In this paper, the term `flavors' refers to these three (or two) generations, whereas our previous paper did not differentiate between `taste' and `flavor'.

The Hamiltonian of the $\mathrm{SU}(3)$ Nambu--Jona-Lasinio model~\cite{Hatsuda} is defined by 
\begin{equation}
	\label{HamSU(3)}
	\begin{split}
		H^{(\Lambda)}(m)
		&:=i\kappa \sum_{x \in \Lambda} \sum_{\mu=1}^\nu(-1)^{\theta_\mu(x)}
		[\Psi^\dagger(x)\Psi(x+e_\mu) - \Psi^\dagger(x+e_\mu)\Psi(x)] \\
		&\quad+
		m\sum_{x\in \Lambda}(-1)^{x^{(1)} + \cdots +x^{(\nu)}}  S^{(2)}(x) + 
		g\sum_{x \in \Lambda}\sum_{\mu=1}^\nu\sum_{a=1}^8S^{(a)}(x) S^{(a)}(x+e_\mu),
	\end{split}
\end{equation}
on the finite lattice $\Lambda$,  
where 
the staggered amplitudes of the hopping Hamiltonian are determined by 
\begin{equation*}
	\theta_1(x):=
	\begin{cases}
		0 & \mbox{for \ } x^{(1)}\ne L;\\
		1 & \mbox{for \ } x^{(1)}=L,
	\end{cases}
\end{equation*}
and for $\mu=2,3,\ldots,\nu$, 
\begin{equation*}
	\theta_\mu(x):=
	\begin{cases}
		x^{(1)}+\cdots+x^{(\mu-1)} & \mbox{for \ } x^{(\mu)}\ne L;\\
		x^{(1)}+\cdots+x^{(\mu-1)}+1 & \mbox{for \ } x^{(\mu)}= L;\\
	\end{cases}
\end{equation*}
and we have written $S^{(a)}(x) := \Psi^\dagger(x) \lambda^{(a)} \Psi(x)$, ($a=1, \dots, 8$), $\kappa \in \mathbb{R}$ and $m, g \ge 0$.
The parameters $m$ and $g$ stand for the quark mass and coupling constant, respectively.
Moreover, $\kappa$ is interpreted as the lattice spacing $a\simeq \kappa^{-1}$.
We define the order parameter~\cite{GK}
\[
O^{(\Lambda)}:=\sum_{x\in \Lambda}(-1)^{x^{(1)} + \cdots +x^{(\nu)}}  S^{(2)}(x).
\]

As shown in Appendix~\ref{sect.SU(3)}, when $m=0$, the Hamiltonian $H^{(\Lambda)}(0)$ is invariant under 
the SU(3) rotations. The term $mO^{(\Lambda)}$ in the Hamiltonian $H^{(\Lambda)}(m)$ is introduced so as to apply 
an infinitesimally weak symmetry breaking field. As we will prove below, it indeed causes a SU(3) symmetry breaking. 
The choice of the order parameter $O^{(\Lambda)}$ is not specific. 
Actually we can use any $S^{(a)}$ instead of $S^{(2)}$ as the order parameter by Proposition~\ref{prop.rot}.
Physically, the operators $S^{(3)}$ and $S^{(8)}$ correspond to the mass difference between the generations of the quarks.  
Therefore, a natural choice of the order parameter is a superposition of $S^{(3)}$ and $S^{(8)}$.  
However, the operator $S^{(8)}$ is a nuisance to use the reflection positivity. 
We use the operator $S^{(2)}$ because it is easy to handle in the present formalism,
and can be transformed into $S^{(3)}$ by a rotation if necessary.

%%%%%%%%%%%%%%%%%%%%%%%%%%%%%%%%%%%%%%%%%%%%%%%%%%%%%%%%%%%%%%%%%%%%
\section{Particle-hole symmetry}
\label{PHsymmetry}

The reason why we use $S^{(2)}$ for the symmetry breaking field in the present paper is that 
such $H^{(\Lambda)}(m)$ has the particle-hole symmetry as follows.
For $x \in \Lambda$ and $i\in\{1,2,3\}$, we introduce the particle-hole transformation by
\begin{equation}
	\label{eq.minus}
	u_i(x) := \left[\prod_{\substack{y \in \Lambda,\ j\in\{1,2,3\}: \\ y \neq x\ {\rm or}\ j\ne i}}(-1)^{n_j(y)} \right]
	[\psi^\dagger_i(x) + \psi_i(x)],
\end{equation}
where we have written $n_i(x) := \psi^\dagger_i(x) \psi_i(x)$ with $i=1, 2, 3$. Since 
\begin{align*}
	e^{i\pi  n_j(x) } \psi_k(y) e^{- i\pi  n_j(x) } &= \psi_k(y) \,\, (k \neq j \text{ or } x\neq y);\\
	e^{i\pi n_j(x) } \psi_j(x) e^{-i\pi n_j(x) } &= -\psi_j(x),
\end{align*}
we have
\begin{equation}
	\label{eq.conj}
	u_i(x)^\dagger \psi_j(y)u_i(x)
	=
	\begin{cases}
		\psi_j^\dagger(x) & \text{if } y=x, j=i; \\
		\psi_j(y) & \text{otherwise}.
	\end{cases}
\end{equation}
Then we define the unitary operator for the particle-hole transformation by
\[
U_\mathrm{PH}^{(\Lambda)}
:=
\prod_{x \in \Lambda}\prod_{i=1}^3 u_i(x).
\]
For any $x \in \Lambda$ and $j=1,2,3$, we have
\[
\left(U_\mathrm{PH}^{(\Lambda)}\right)^\dagger
\psi_j(x)
U_\mathrm{PH}^{(\Lambda)}
=\psi_j^\dagger(x).
\]
Noting that 
\begin{equation}
	\label{eq.su}
	S^{(a)}(x)
	=
	\begin{cases}
		\psi_i^\dagger(x)\psi_j(x) + \psi^\dagger_j(x)\psi_i(x) & (a=1,4,6);\\
		i(\psi_i^\dagger(x)\psi_j(x) - \psi^\dagger_j(x)\psi_i(x)) &(a=2, 5, 7) \\
	\end{cases}
\end{equation}
for certain $i\neq j$,  
$$
S^{(3)}(x) = n_1(x)- n_2(x)\quad \mbox{and}\quad  S^{(8)}(x) = (n_1(x)+n_2(x)-2n_3(x))/\sqrt{3},
$$ 
we have that 
\begin{equation}
	\label{eq.ph}
	(U_\mathrm{PH}^{(\Lambda)})^\dagger S^{(a)}(x)U_\mathrm{PH}^{(\Lambda)} = 
	\begin{cases}
		S^{(a)}(x) & \text{if } a=2,5,7;\\
		-S^{(a)}(x) &\text{otherwise}
	\end{cases}
\end{equation}
by the anti-commutation relations. Further, we deduce from the anti-commutation relations that for all $j=1, 2, 3$
\begin{align*}
	\left(U_\mathrm{PH}^{(\Lambda)}\right)^\dagger
	&[\psi_j^\dagger(x)\psi_j(x+e_\mu) - \psi^\dagger_j(x+e_\mu)\psi_j(x)]
	U_\mathrm{PH}^{(\Lambda)}\\
	&=
	[\psi_j^\dagger(x)\psi_j(x+e_\mu) - \psi^\dagger_j(x+e_\mu)\psi_j(x)].
\end{align*}
Together with these properties, we can see that our Hamiltonian is invariant under the particle-hole transformation, 
i.e., for any $m$, 
\begin{equation}
	\label{eq.phHam}
	\left(U_\mathrm{PH}^{(\Lambda)}\right)^\dagger
	H^{(\Lambda)}(m)
	U_\mathrm{PH}^{(\Lambda)}
	=
	H^{(\Lambda)}(m).
\end{equation}
This invariance will be used in Section~\ref{sect.NG}. 

In passing, we remark the following. The usual mass term of the Hamiltonian is given by 
\begin{equation}
	m\sum_{x\in\Lambda}(-1)^{x^{(1)}+\cdots +x^{(\nu)}}[\psi_1^\dagger(x)\psi_1(x)+\psi_2^\dagger(x)\psi_2(x)
	+\psi_3^\dagger(x)\psi_3(x)]
\end{equation}
in the staggered fermion formalism. Clearly, this is not invariant under the particle-hole transformation. 
Therefore, a strictly positive mass parameter $m$ stabilizes a phase 
which is different from both of the massless phase and the symmetry-breaking phase which is induced from the massless phase.

%%%%%%%%%%%%%%%%%%%%%%%%%%%%%%%%%%%%%%%%%%%%%%%%%%%
\section{Main results} 
\label{sect.results}
Let us describe our main results in the present paper. We write 
$$
Z_{\beta, m}^{(\Lambda)} := \mathrm{Tr} \exp[-\beta H^{(\Lambda)}(m)]
$$ 
for the partition function, where $\beta\ge 0$ is the inverse temperature.
The expectation value is given by
\begin{equation}
	\label{eq.expv}
	\langle A \rangle_{\beta,m}^{(\Lambda)}:= \frac{1}{Z_{\beta, m}^{(\Lambda)}} \mathrm{Tr}\left\{A\exp[-\beta H^{(\Lambda)}(m)]\right\}
\end{equation}
for an observable $A$. 
For finite temperature field theories, see, e.g.,~\cite[Ch.17--20]{Rothe}.

The first goal of our paper is to prove the existence of long-range order as in~\cite{GK}.
We define the long-range order parameter~\cite{KT, Tasaki} by
\[
m_{\mathrm{LRO}}^{(\Lambda)}(\beta):= \vert \Lambda\vert^{-1}\sqrt{\langle [O^{(\Lambda)}]^2 \rangle^{(\Lambda)}_{\beta, 0}},
\]
and
\[
m_{\mathrm{LRO}}(\beta):=\lim_{\Lambda \nearrow \mathbb{Z}^\nu}m_{\mathrm{LRO}}^{(\Lambda)}(\beta).
\]
Then we will prove the following theorems:

\begin{theorem}
	\label{th.LT}
	For the spatial dimension $\nu \ge 5$, there exist a positive number $\alpha_0$ small enough and sufficiently large $\beta_0>0$ 
	such that  $m_{\mathrm{LRO}}(\beta)> 0$ for all $\vert\kappa\vert/g \le \alpha_0$ and $\beta \ge \beta_0$. 
	Namely, there appears a long-range order for a sufficiently strong coupling constant of the interactions and 
	for a sufficiently low temperatures. 
\end{theorem}

Let 
\[
\omega_m^{(\Lambda)}(A):=\lim_{\beta \nearrow \infty}\langle A \rangle_{\beta, m}^{(\Lambda)}
\]
be the expectation value for an observable $A$ in the ground state (the zero temperature limit). 
Then, there appears the corresponding long-range order also in the ground state as follows: 

\begin{theorem}
	\label{th.gs}
	For the spatial dimension $\nu \ge 5$, there exist a positive number $\alpha_0$ small enough such that  
	\[
	\lim_{\Lambda \nearrow \mathbb{Z}^\nu} \frac{1}{\vert\Lambda\vert^2}\omega_0^{(\Lambda)}([O^{(\Lambda)}]^2) >0
	\]
	for all $\vert\kappa\vert/g \le \alpha_0$.
\end{theorem}

Consequently, the $\mathrm{SU}(3)$ symmetry is spontaneously broken according to the Koma--Tasaki theorem~\cite{KT}.
Nevertheless, the ground state remains invariant under the subgroup generated by $Q^{(2)}$ and $Q^{(8)}$, where $Q^{(a)}:=\sum_{x\in \Lambda}S^{(a)}(x)$  (see also the first line of Sect.~\ref{sect:NGn}).
Thus, the $\mathrm{SU}(3)$ symmetry of our NJL model is spontaneously broken down to $\mathrm{SO}(2) \times \mathrm{U}(1)$.

\begin{remark}
	\label{rem.SU(2)}
	Theorems~\ref{th.LT} and~\ref{th.gs} are also valid for SU(2) model in the spatial dimension $\nu \ge 3$.
	The proof in this SU(2) case follows along the same line and is slightly simpler (see Section~\ref{Sect.SU(2)}).
\end{remark}

\begin{remark}
	\label{rem.dim}
	{From} the large N analysis for the $\mathrm{SU}(N)$ symmetry as in \cite{ABG}, one can expect that 
	the treatment of the large N case is much easier than that of the small N case. 
	Therefore, a reader feels that the condition $\nu \ge 5$ in Theorems~\ref{th.LT} and~\ref{th.gs} is somewhat strange, 
	compared with $\nu\ge 3$ in the case of SU(2) as in the above Remark~\ref{rem.SU(2)}. 
	This is due to the following technical reason: 
	Although our method can be generalized to the $\mathrm{SU}(N)$ models with $N \ge 4$,  
	the prefactor $24$ in the right-hand side of Eq. (\ref{eq.Dinteract}) below changes to a much larger value
	because of the commutation relations about the $\mathrm{SU}(N)$ algebra.
	In consequence, the spatial dimension $\nu$ in which our proof of the long-range order works becomes even larger.
	Thus, the condition $\nu \ge 5$ may not be optimal for the case of SU(3) symmetry.
\end{remark}

As mentioned in the introduction, our main purpose of the present paper is to prove the existence of Nambu--Goldstone mode. 
Since the existence of the long-range order implies the existence of a non-vanishing spontaneous magnetization 
in the infinite-volume limit \cite{KT}, the corresponding infinite-volume ground state shows a symmetry breaking. 
Let $\omega_0(A)$ be the symmetry breaking infinite-volume ground state which is defined by
\[
\omega_0(A) :=
\mathrm{weak}^*\text{-}\lim_{m\searrow 0}\mathrm{weak}^*\text{-}\lim_{\Lambda \nearrow \mathbb{Z}^\nu}\omega_m^{(\Lambda)}(A)
\]
in the sense of the weak$^\ast$ convergence. 
Then, as is well known, a Hilbert space is generated from the state $\omega_0(\cdot)$ 
by the Gelfand--Naimark--Segal (GNS) 
representation~\cite{BRI, BR}, and we can construct an excited state in the Hilbert space with an arbitrary small energy gap 
(see Section~\ref{sect.NG} for a more precise statement).
Namely, we will prove
\begin{theorem}
	\label{th.NG}
	Suppose that the spatial dimension $\nu \ge 5$ and $\vert\kappa\vert/g \le \alpha_0$, where $\alpha_0$ is given in Theorem~\ref{th.gs}.
	Then, there exists a gapless excitation above the infinite-volume ground state $\omega_0(\cdots)$. 
	Namely, there appears a Nambu--Goldstone mode associated with the SU(3) symmetry breaking above the infinite-volume ground state. 
\end{theorem}
In addition, Theorem~\ref{th.NG} is also true for the SU(2) model in the spatial dimension $\nu \ge 3$ by the same argument.

Up to now we have been concerned with demonstrating the existence of the phase transition and Nambu--Goldstone modes.
Finally, we study how many gapless modes emerge associated with the breakdown of the SU(3) symmetry. 
It should be noted that, in general, the number of gapless states, $N_\mathrm{NG}$, 
is not necessarily equal to the number of broken generators, $N_\mathrm{BS}$.
In particular, by the spin-wave theory, the Heisenberg ferromagnet has only one gapless mode (magnon) 
while two spin rotation symmetries are broken, i.e., $N_\mathrm{BS}=2$ and $N_\mathrm{NG}=1$.
We refer to~\cite{Hidaka, WM} for further details.

The following theorem shows that the number of Nambu--Goldstone modes is equal to the number of broken symmetry generators.
\begin{theorem}
	\label{cor.NGn}
	For the SU(3) NJL model, we have
	\[
	N_\mathrm{NG}=N_\mathrm{BS}=6.
	\]
	Moreover, for the SU(2) model, $N_\mathrm{NG}=N_\mathrm{BS}=2$ holds true.
\end{theorem}
This theorem is intuitively expected from the analogy to the quantum Heisenberg anti-ferromagnet, but proving this is not so trivial.
In fact, unlike the Heisenberg models, it is difficult to obtain the ground-state properties of our NJL model.
For instance, we could not directly apply the Marshall~\cite{Marshall} and Lieb--Mattis~\cite{LM} type theorem.
Therefore, the proof of Theorem~\ref{cor.NGn} is one of the main novel features of our paper.

%In order to prove these theorems it is first necessary to use reflection positivity as in~\cite{GK, DLS, FSS, Koma4}.
%In Section~\ref{sect.RP}, we verify the fermion reflection positivity~\cite{JP, Koma4, GK} for our model 
%and prove the Gaussian domination bound (Proposition~\ref{GaussianD} below).
%This allows us to use the infrared bounds by the general theory~\cite{FILS}.
%The proofs of Theorem~\ref{th.LT} and Theorem~\ref{th.gs} are given in Section~\ref{sect.LRO1} and Section~\ref{LRO2}, respectively.
%Section~\ref{Sect.SU(2)} is devoted to the definition of the SU(2) NJL model 
%and the proof of the existence of long-range order, as mentioned in Remark~\ref{rem.SU(2)}.
%In Section~\ref{sect.NG}, we prove Theorem~\ref{th.NG}.

%%%%%%%%%%%%%%%%%%%%%%%%%%%%%%%%%%%%%%%%%%%%%%%%%%%%%%%%%%%%%
\section{Reflection positivity and Gaussian domination}
\label{sect.RP}

In this section we prove a Gaussian domination bound. We begin with some definitions for reflection positivity.
Let $\Omega \subset \Lambda$ be a subset of the present hypercubic lattice $\Lambda$ 
and $\mathcal{A}(\Omega)$ be the algebra generated by the fermion operators 
$\psi_i^\dagger(x), \psi_j(y)$ for $x,y \in \Omega$ and $i, j =1,2,3$.
Clearly there are a decomposition $\Lambda = \Lambda_- \cup \Lambda_+$ with $\Lambda_- \cap \Lambda_+ = \emptyset$ 
and a reflection map $r:\Lambda_{\pm}\rightarrow \Lambda_{\mp}$ satisfying $r(\Lambda_{\pm})=\Lambda_{\mp}$. 
We write $\mathcal{A} := \mathcal{A}(\Lambda)$ and $\mathcal{A}_\pm := \mathcal{A}(\Lambda_\pm)$.
We take an anti-linear map\footnote{In the Euclidean formalism, treated in~\cite{OS, Seiler}, the corresponding map should be an anti-morphism due to the anti-symmetry of the Grassmann algebra. In any case, it should be noted again that these formalisms and ours differ in several respects.} $\vartheta : \mathcal{A}_\pm \to \mathcal{A}_\mp$ to satisfy 
\begin{align*}
	\vartheta(\psi_i(x))
	= \psi_i(r(x)), \quad
	\vartheta(\psi^\dagger_i(x))
	= \psi_i^\dagger(r(x)), \\
	\vartheta(AB) = \vartheta(A)\vartheta(B), \quad
	\vartheta(A)^\dagger=\vartheta(A^\dagger) \quad \text{for }A, B \in \mathcal{A}.
\end{align*}

First, we divide our finite lattice $\Lambda$ by the $x^{(1)} = 1/2$ hyperplane into 
$\Lambda_- :=\{x \in \Lambda \colon -L+1 \le x^{(1)} \le 0\}$ and 
$\Lambda_+ := \{x  \in \Lambda \colon 1 \le x^{(1)} \le L\}$.
As in~\cite{DLS, Koma4}, we introduce some unitary transformations. First, we define
\[
U_{1, j}:=\prod_{\substack{x \in \Lambda: \\ x^{(j)}=\text{even}}}\prod_{k=1}^3\exp\left[\frac{i\pi}{2} n_k(x)\right],
\] 
for $j=2, \dots, \nu$, where $n_k(x) := \psi_k(x)^\dagger \psi_k(x)$ is the number operator, and let
\begin{equation}
	\label{U1}
	U_1 := \prod_{j=2}^\nu U_{1, j}.
\end{equation}
Using the anti-commutation relations, we see that 
\begin{equation}
	\label{eq.u1}
	U^\dagger_{1, j}\psi_k(x)U_{1, j}
	=
	\begin{cases}
		i\psi_k(x)  &\text{if } x^{(j)} \text{ is even}; \\
		\psi_k(x)  &\text{if } x^{(j)} \text{ is odd}
	\end{cases}
\end{equation}
for any $k$ and $x \in \Lambda$ and for $j=2,\ldots,\nu$. 
The transformation for $\psi_k^\dagger$ is also obtained by $U^\dagger \psi^\dagger U =(U^\dagger \psi U)^\dagger$ 
for any $\psi$ and $U$.

Now we write the hopping term, $H_K^{(\Lambda)}$, in  (\ref{HamSU(3)}) as
\begin{equation*}
	H_K^{(\Lambda)}:= \sum_{\mu=1}^\nu H_{K, \mu}^{(\Lambda)},
\end{equation*}
where
\begin{equation}
	\label{eq.defhop}
	H_{K, \mu}^{(\Lambda)} := i\kappa \sum_{x \in \Lambda}(-1)^{\theta_\mu(x)}
	[\Psi^\dagger(x) \Psi(x+e_\mu) - \Psi^\dagger(x+e_\mu)\Psi(x)].
\end{equation}
Then, when $\mu \neq \lambda$,
\[
U^\dagger_{1, \lambda}H^{(\Lambda)}_{K, \mu}U_{1, \lambda} = H^{(\Lambda)}_{K, \mu},
\]
and
\begin{align*}
	U^\dagger_{1, \mu}H^{(\Lambda)}_{K, \mu}U_{1, \mu} 
	&= 
	\kappa\sum_{\substack{x \in \Lambda \\ x^{(\mu)}\neq L}}(-1)^{\sum_{j=1}^\mu x^{(j)}}[\Psi^\dagger(x)\Psi(x+e_\mu)
	+ \Psi^\dagger(x+e_\mu)\Psi(x)]\\
	&\quad-\kappa\sum_{\substack{x \in \Lambda \\ x^{(\mu)}= L}}(-1)^{\sum_{j=1}^\mu x^{(j)}}[\Psi^\dagger(x)\Psi(x+e_\mu)
	+ \Psi^\dagger(x+e_\mu)\Psi(x)].
\end{align*}
These imply that 
\begin{equation}
	U^\dagger_{1}H^{(\Lambda)}_{K, 1}U_{1}= H^{(\Lambda)}_{K, 1},
\end{equation}
and for $\mu=2, \dots, \nu$
\begin{equation}
	\begin{split}
		U^\dagger_{1}H^{(\Lambda)}_{K, \mu}U_{1} 
		&= 
		\kappa\sum_{\substack{x \in \Lambda \\ x^{(\mu)}\neq L}}(-1)^{\sum_{j=1}^\mu x^{(j)}}[\Psi^\dagger(x)\Psi(x+e_\mu)
		+ \Psi^\dagger(x+e_\mu)\Psi(x)]\\
		&\quad-\kappa\sum_{\substack{x \in \Lambda \\ x^{(\mu)}= L}}(-1)^{\sum_{j=1}^\mu x^{(j)}}[\Psi^\dagger(x)\Psi(x+e_\mu)
		+ \Psi^\dagger(x+e_\mu)\Psi(x)].
	\end{split}
\end{equation}
To introduce the second transformation, define $\Lambda_{\mathrm{odd}} 
:= \{x \in \Lambda \colon x^{(1)} + \cdots x^{(\nu)} = \text{ odd}\}$ and
\begin{equation}
	\label{eq.odd}
	U_\mathrm{odd}:=\prod_{x \in \Lambda_\mathrm{odd}}\prod_{i=1}^3u_i(x),
\end{equation}
where $u_i(x)$ is given in (\ref{eq.minus}).
By (\ref{eq.conj}), we see
\begin{equation}
	\label{eq.odd2}
	U_\mathrm{odd}^\dagger \psi_j(x) U_\mathrm{odd}
	=
	\begin{cases}
		\psi_j^\dagger(x) & \text{if } x \in \Lambda_{\mathrm{odd}}; \\
		\psi_j(x) & \text{otherwise},
	\end{cases}
\end{equation}
for any $j=1,2,3$.
Let $\tilde U_1 := U_1U_\mathrm{odd}$. Then, for $\mu = 2, \dots, \nu-1$,
\begin{equation*}
	\label{eq.UkineticA}
	\begin{split}
		\tilde H_{K, \mu}^{(\Lambda)} 
		&:= \tilde U_1^\dagger H_{K, \mu}^{(\Lambda)} \tilde U_1\\
		&=
		\kappa \sum_{\substack{x \in \Lambda \\ x^{(\mu)} \neq L}}(-1)^{\sum_{i=\mu +1}^\nu x^{(i)}}[\Psi^\dagger(x)\Psi(x+e_\mu)
		+\Psi^\dagger(x+e_\mu)\Psi(x)]\\
		&\quad-\kappa \sum_{\substack{x \in \Lambda \\ x^{(\mu)} = L}}(-1)^{\sum_{i=\mu +1}^\nu x^{(i)}}[\Psi^\dagger(x)\Psi(x+e_\mu)
		+\Psi^\dagger(x+e_\mu)\Psi(x)],
	\end{split}
\end{equation*}
and
\begin{equation*}
	\label{eq.UkineticB}
	\begin{split}
		\tilde H_{K, \nu}^{(\Lambda)} := \tilde U_1^\dagger H_{K, \nu}^{(\Lambda)} \tilde U_1
		&=
		\kappa \sum_{\substack{x \in \Lambda \\ x^{(\nu)} \neq L}}[\Psi^\dagger(x)\Psi(x+e_\nu) +\Psi^\dagger(x+e_\nu)\Psi(x)]\\
		&\quad-\kappa \sum_{\substack{x \in \Lambda \\ x^{(\nu)} = L}}[\Psi^\dagger(x)\Psi(x+e_\nu) +\Psi^\dagger(x+e_\nu)\Psi(x)].
	\end{split}
\end{equation*}
These terms have the form
\begin{equation*}
	\tilde H_{K, \mu}^{(\Lambda)} = \tilde H_{K, \mu}^+ + \tilde H_{K, \mu}^-
\end{equation*}
for $\mu =2, \dots, \nu$, with $\tilde H_{K, \mu}^\pm \in \mathcal{A}_\pm$ obeying 
$\vartheta(\tilde H_{K, \mu}^\pm)=\tilde H_{K, \mu}^\mp$.

In order to treat the term $H_{K, 1}^{(\Lambda)}$, we introduce the Majorana fermion operators $\xi_j(x), \eta_j(x)$ by
\[
\xi_j(x) := \psi^\dagger_j(x)+\psi_j(x), \quad \eta_j(x) := i(\psi_j^\dagger(x)-\psi_j(x)),
\]
for $j=1,2,3$ and $x \in \Lambda$.
These operators obey $\xi_i^\dagger(x) =\xi_i(x)$, $\eta_i^\dagger(x)=\eta_i(x)$, $i=1, 2, 3$, and the anti-communication relations
\begin{align*}
	\{\xi_i(x), \xi_j(y)\} &= 2 \delta_{x, y}\delta_{i, j}, \quad \{\eta_i(x), \eta_j(y)\} = 2 \delta_{x, y}\delta_{i, j},\\
	\{\xi_i(x), \eta_j(y)\}&=0.
\end{align*}
In addition, for $x \in \Lambda$ and $i=1,2,3$, we see
\[
(\vartheta\xi_i)(x) = \xi_i(\vartheta(x)), \quad (\vartheta \eta_i)(x)=-\eta_i(\vartheta(x)).
\]
We note that for $i=1,2,3$
\[
\psi_i^\dagger(x)\psi_i(x) - \psi_i^\dagger(y)\psi_i(x)
=\frac{\xi_i(x)\xi_i(y) + \eta_i(x)\eta_i(y)}{2}.
\]
Furthermore, for $i=1,2,3$, we see that $U^\dagger_\mathrm{odd} \xi_i(x) U_\mathrm{odd} = \xi_i(x)$ and
\[
U^\dagger_\mathrm{odd} \eta_i(x) U_\mathrm{odd} = 
\begin{cases}
	-\eta_i(x) &\text{if }x \in \Lambda_\mathrm{odd};\\
	\eta_i(x) &\text{otherwise}.
\end{cases} 
\]
Now we have
\begin{equation}
	\begin{split}
		\tilde H^{(\Lambda)}_{K, 1} := \tilde U_1^\dagger H^{(\Lambda)}_{K, 1}\tilde U_1
		&=
		\frac{i\kappa}{2}
		\sum_{\substack{x \in \Lambda \\ x^{(1)} \neq L}}\sum_{j=1}^3[\xi_j(x)\xi_j(x+e_1) -\eta_j(x)\eta_j(x+e_1)] \\
		&-\frac{i\kappa}{2}\sum_{\substack{x \in \Lambda \\ x^{(1)} =L}}\sum_{j=1}^3[\xi_j(x)\xi_j(x+e_1) -\eta_j(x)\eta_j(x+e_1)].
	\end{split}
\end{equation}
Thus $\tilde  H^{(\Lambda)}_{K, 1}$ can be decomposed into three parts as follows: 
\[
\tilde H^{(\Lambda)}_{K, 1}
=
\tilde H^{+}_{K, 1}+\tilde H^{-}_{K, 1}+\tilde H^{0}_{K, 1},
\]
where 
\[
\tilde H^{\pm}_{K, 1}
:=\frac{i\kappa}{2}
\sum_{\substack{x, x+e_1 \in \Lambda_\pm}}\sum_{j=1}^3[\xi_j(x)\xi_j(x+e_1) -\eta_j(x)\eta_j(x+e_1)]
\]
and
\begin{equation}
	\label{tildeH0K1}
	\tilde H^{0}_{K, 1}
	:=\frac{i\kappa}{2}
	\sum_{\substack{x \in \Lambda \\x^{(1)}=0, -L+1 }}\sum_{j=1}^3[\xi_j(x)(\vartheta\xi_j)(x) +\eta_j(x)(\vartheta\eta_j)(x)].
\end{equation}
We see from these expressions that
\begin{equation}
	\label{eq.kintrans}
	\tilde H_K
	:=
	\tilde U_1^\dagger H^{(\Lambda)}_{K}\tilde U_1 = \tilde H_K^+ +\tilde H_K^- +\tilde H_K^0,
\end{equation}
where $\tilde H_K^0$ is equal to the above $\tilde H^{0}_{K, 1}$ of (\ref{tildeH0K1}). 

Next, let us deal with the interaction term of the Hamiltonian.
Since the operators $S^{(a)}$ ($a=2, 5, 7$) are pure imaginary hermitian, 
we treat the corresponding interactions by separating them from the rest as in~\cite[Lemma 6.1]{DLS}.
For latter purpose, we now introduce real-valued functions $h=(h^{(1)},\dots, h^{(\nu)})$ with $h^{(\mu)}:\Lambda \to \mathbb{R}$.
It is convenient to write
\begin{equation}
	\label{Hinth}
	\begin{split}
		H_{\mathrm{int}}^{(\Lambda)}(h)
		&:=
		\frac{g}{2}\sum_{\mu=1}^\nu H_{\mathrm{int}, \mu, 3}^{(\Lambda)}(h)
		+\frac{g}{2}\sum_{x \in \Lambda}\sum_{\mu=1}^\nu\sum_{b = 1,4,6,8}[S^{(b)}(x)+S^{(b)}(x+e_\mu)]^2\\
		\quad& -\frac{g}{2}\sum_{x \in \Lambda}\sum_{\mu=1}^\nu\sum_{b =2, 5, 7}[S^{(b)}(x)- S^{(b)}(x+e_\mu)]^2\\
		&\quad-g\nu\sum_{a\neq 2, 5, 7}\sum_{x \in \Lambda}S^{(a)}(x)^2
		+g\nu\sum_{b =2, 5, 7}\sum_{x \in \Lambda}S^{(b)}(x)^2,
	\end{split}
\end{equation}
where
\begin{equation}
	\label{eq.int}
	H_{\mathrm{int}, \mu, 3}^{(\Lambda)}(h)
	:=
	\frac{g}{2} \sum_{x \in \Lambda}
	[S^{(3)}(x) +S^{(3)}(x+e_\mu) + (-1)^{\sum_{j=1}^\nu x^{(j)}} h^{(\mu)}(x)]^2.
\end{equation}
Clearly, $H_{\mathrm{int}}^{(\Lambda)}(0)$ is equal to the original interaction term. One has 
$$
S^{(a)}(x) =\sum_{i,j}\lambda_{ij}^{(a)}\psi^\dagger_i(x)\psi_j(x).
$$ 
{from} the definition $S^{(a)}(x)=\Psi^\dagger(x)\lambda^{(a)}\Psi(x)$. By combining this, (\ref{U1}) and (\ref{eq.u1}),
one obtains  
$$
U_1^\dagger S^{(a)}(x) U_1 = S^{(a)}(x)\quad \mbox{for all $a$}.
$$
Similarly, by (\ref{eq.odd2}), we have
\begin{equation}
	\label{eq.Sodd}
	U^\dagger_\mathrm{odd}S^{(a)}U_\mathrm{odd}
	=
	\begin{cases}
		-S^{(a)}(x) &\text{if } x \in \Lambda_\mathrm{odd}, a\neq 2,5, 7;\\
		S^{(a)}(x) &\text{otherwise}
	\end{cases}
\end{equation}
in the same way as (\ref{eq.ph}).
Combining these with $\tilde{U}_1=U_1U_{\rm odd}$, one has 
\begin{equation}
	\label{tildeHinth}
	\begin{split}
		\tilde H^{(\Lambda)}_\mathrm{int}(h)
		&:=\tilde U^\dagger_1H^{(\Lambda)}_\mathrm{int}(h)\tilde U_1\\
		&=
		\frac{g}{2}\sum_{\mu=1}^\nu \sum_{x \in \Lambda}
		\left[S^{(3)}(x) -S^{(3)}(x+e_\mu) + h^{(\mu)}(x)\right]^2\\
		&\quad +\frac{g}{2}\sum_{x \in \Lambda}\sum_{\mu=1}^\nu\sum_{b = 1,4,6,8}\left[S^{(b)}(x)-S^{(b)}(x+e_\mu)\right]^2
		\\
		&\quad -\frac{g}{2}\sum_{x \in \Lambda}\sum_{\mu=1}^\nu\sum_{b =2, 5, 7}\left[S^{(b)}(x)- S^{(b)}(x+e_\mu)\right]^2 \\
		&\quad -g\nu\sum_{a\neq 2, 5, 7}\sum_{x \in \Lambda}S^{(a)}(x)^2 +g\nu\sum_{b =2, 5, 7}\sum_{x \in \Lambda}S^{(b)}(x)^2.
	\end{split}
\end{equation}
Here the minus signs for the imaginary hermitian matrices are crucial for deriving the Gaussian domination bound below.
This expression $\tilde H^{(\Lambda)}_\mathrm{int}(h)$ of the interaction Hamiltonian can be decomposed into three parts 
as follows: 
\begin{equation}
	\label{decompotildeHinth}
	\tilde H^{(\Lambda)}_\mathrm{int}(h)
	=
	\tilde H^{+}_\mathrm{int}(h)+\tilde H^{-}_\mathrm{int}(h)+\tilde H^{0}_\mathrm{int}(h),
\end{equation}
where $\tilde H^{\pm}_\mathrm{int}(h) \in \mathcal{A}_\pm$ and
\begin{equation}
	\label{tildeH0inth}
	\begin{split}
		\tilde H^{0}_\mathrm{int}(h)
		&:=\frac{g}{2} \sum_{\substack{x \in \Lambda \\x^{(1)}=0, L}}
		\biggl\{\left[S^{(3)}(x) -S^{(3)}(x+e_1) + h^{(1)}(x)\right]^2\\
		&+\sum_{a=1,4,6,8}\left[S^{(a)}(x)- S^{(a)}(x+e_1)\right]^2 
		-\sum_{b=2,5,7}\left[S^{(b)}(x)- S^{(b)}(x+e_1)\right]^2\biggr\}.
	\end{split}
\end{equation}
Similarly, we obtain 
\begin{equation}
	\label{tildeOLambda}
	\tilde U^\dagger_1O^{(\Lambda)}\tilde U_1
	=
	O^{(\Lambda)},
\end{equation}
and
\begin{equation}
	\label{tildeHSBm}
	H_\mathrm{SB}(m)
	:=
	m O^{(\Lambda)}
	=
	H^+_\mathrm{SB}(m) + H^-_\mathrm{SB}(m),
\end{equation}
where $H^\pm_\mathrm{SB}(m) \in \mathcal{A}_\pm$ and $\vartheta(H^\pm_\mathrm{SB}(m))=H^\mp_\mathrm{SB}(m)$.

In order to obtain the Gaussian domination bound, we introduce 
\[
H^{(\Lambda)}(m, h):= H_K^{(\Lambda)} + H_{\mathrm{int}}^{(\Lambda)}(h) + mO^{(\Lambda)},
\]
where $H_K^{(\Lambda)}$ is the hopping term of the present Hamiltonian $H^{(\Lambda)}(m)$ of (\ref{HamSU(3)}) 
and $H_{\mathrm{int}}^{(\Lambda)}(h)$ is given by (\ref{Hinth}). 
Then, from (\ref{eq.kintrans}) and (\ref{tildeHinth})--(\ref{tildeHSBm}), we have 
\begin{equation}
	\label{eq.UH}
	\begin{split}
		\tilde H^{(\Lambda)}(m, h)
		&:=
		\tilde U^\dagger_1H^{(\Lambda)}(m, h)\tilde U_1\\
		&=
		\tilde H^{+}(m, h)
		+
		\tilde H^{-}(m, h)
		+\tilde H^{0}(m, h),
	\end{split}
\end{equation}
where $\tilde H^{\pm}(m, h) \in \mathcal{A}_\pm$ and
\begin{equation}
	\tilde H^{0}(m, h)
	:=
	\tilde H^{0}_K+\tilde H^{0}_\mathrm{int}(h).
\end{equation}

We next turn to the Gaussian domination bound.
\begin{proposition}[Gaussian domination]
	\label{GaussianD}
	For any real-valued $h = (h^{(1)}, \dots, h^{(\mu)})$, it follows that
	\begin{align}
		\label{GaussianDbund}
		\mathrm{Tr} \left\{\exp\left[- \beta H^{(\Lambda)}(m, h)\right] \right\}
		&\le
		\mathrm{Tr}\left\{\exp\left[- \beta H^{(\Lambda)}(m)\right]\right\},
	\end{align}
	where $H^{(\Lambda)}(m)=H^{(\Lambda)}(m, 0)$ is the present Hamiltonian (\ref{HamSU(3)}).
\end{proposition}

\begin{proof}
	First we treat $\tilde H^{(\Lambda)}(m, h)$.
	Using the Lie product formula, we have the identity
	$$
	\mathrm{Tr} \left\{\exp\left[- \beta \tilde H^{(\Lambda)}(m, h)\right]\right\}
	= \lim_{n \to \infty}\mathrm{Tr}\left( \alpha_n^n\right)
	$$ 
	with
	\begin{align*}
		\alpha_n
		&:=
		\left(1 - \frac{\beta}{n} \tilde H^{0}_{K, 1} \right)\left\{\prod_{\substack{x \in \Lambda: \\ x^{(1)} =0, L}} 
		\exp\left[{-\frac{\beta g}{2n}\left(S^{(3)}(x) - S^{(3)}(x+ e_1) + h^{(1)}(x)\right)^2}\right]\right\}
		\\
		&\quad \times
		\left\{\prod_{a=1,4,6,8}\prod_{\substack{x \in \Lambda: \\ x^{(1)} =0, L}} 
		\exp\left[-{\frac{\beta g}{2n}\left(S^{(a)}(x) - S^{(a)}(x+ e_1)\right)^2}\right]\right\}  \\
		&\quad \times \left\{\prod_{b=2,5,7}\prod_{\substack{x \in \Lambda: \\ x^{(1)} =0, L}} 
		\exp\left[{\frac{\beta g}{2n}\left(S^{(b)}(x) - S^{(b)}(x+ e_1)\right)^2}\right]\right\} \\
		&\quad \times
		\exp\left[- \frac{\beta}{n} \tilde H^{-}(m, h)\right] 
		\exp\left[-\frac{\beta}{n} \tilde H^+(m, h)\right].
	\end{align*}
	Let us consider the exponential function of the square of an operator. 
	As to the real hermitian matrices $S^{(a)}$ for $a\neq 2, 5, 7$, we use the operator identity
	\[
	e^{-D^2} = \int_\mathbb{R} \frac{dk}{\sqrt{4 \pi}}\; e^{ikD} e^{-k^2/4},
	\]
	for any hermitian matrix $D$.
	Then one has 
	\begin{equation}
		\label{IdS3}
		\begin{split}
			& \exp\left[{-\frac{\beta g}{2n}(S^{(3)}(x) - S^{(3)}(x+ e_1) + h^{(1)}(x))^2}\right]\\
			&=
			\int_\mathbb{R} \frac{dk}{\sqrt{4 \pi}} e^{-k^2/4} e^{ik \sqrt{\frac{\beta g}{2n}} S^{(3)}(x)} 
			e^{-ik \sqrt{\frac{\beta g}{2n}} S^{(3)}(x+ e_1)} e^{ik \sqrt{\frac{\beta g}{2n}}  h^{(1)}(x)},
		\end{split}
	\end{equation}
	and
	\begin{align}
		\label{IdSreal}
		&\exp\left[-{\frac{\beta g}{2n}(S^{(a)}(x) - S^{(a)}(x+ e_1))^2}\right]\\
		&=
		\int_\mathbb{R} \frac{dk}{\sqrt{4 \pi}} e^{-\frac{k^2}{4}} e^{ik \sqrt{\frac{\beta g}{2n}} S^{(a)}(x)} 
		e^{-ik \sqrt{\frac{\beta g}{2n}} S^{(a)}(x+ e_1)}
	\end{align}
	for $a=1,4,6,8$. As for the pure imaginary hermitian matrices $S^{(a)}$ for $a= 2, 5, 7$, we use a different operator identity
	\[
	e^{D^2} =  \int_\mathbb{R}\frac{dk}{\sqrt{4 \pi}} \; e^{kD} e^{-k^2/4},
	\]
	for any hermitian matrix $D$.
	Then one also obtains
	\begin{align}
		\label{IdSimaginary}
		\exp\left[{\frac{\beta g}{2n}\left(S^{(a)}(x) - S^{(a)}(x+ e_1)\right)^2}\right]
		&=
		\int_\mathbb{R} \frac{dk}{\sqrt{4 \pi}} e^{-k^2/4} e^{k \sqrt{\frac{\beta g}{2n}} S^{(a)}(x)} 
		e^{-k \sqrt{\frac{\beta g}{2n}} S^{(a)}(x+ e_1)}
	\end{align}
	for $a= 2, 5, 7$. By relying on these indentities, we want to rewrite $\mathrm{Tr}\left( \alpha_n^n\right)$. 
	For this purpose, we introduce 
	\begin{align*}
		&K_-(k, m, h) := A_-(k) e^{-\beta \tilde H^- (m, h)/n}, \\
		&K_+(k, m, h) := \underbrace{\vartheta(A_-(k))}_{=: A_+(k)} e^{-\beta \tilde H^+ (m, h)/n},
	\end{align*}
	where we have written 
	\begin{align*}
		&A_-(k)\\ 
		&:=  \left\{\prod_{a\in \mathcal{S}}%=1,3,4,6,8}
	\prod_{\substack{x \in \Lambda, \\ 
			x^{(1)} =0}} \exp\left[ik^{(a)}(x) \sqrt{\frac{\beta g}{2n}} S^{(a)}(x)\right]
	\prod_{\substack{x \in \Lambda, \\ x^{(1)} =-L+1}} \exp\left[-ik^{(a)}(x) \sqrt{\frac{\beta g}{2n}} S^{(a)}(x)\right]\right\}\\
	&\quad \times
	\left\{\prod_{a=2, 5, 7}
	\prod_{\substack{x \in \Lambda, \\ x^{(1)} =0}} \exp\left[k^{(a)}(x) \sqrt{\frac{\beta g}{2n}} S^{(a)}(x)\right]
	\prod_{\substack{x \in \Lambda, \\ x^{(1)} =-L+1}}   \exp\left[-k^{(a)}(x) \sqrt{\frac{\beta g}{2n}} S^{(a)}(x)\right]\right\}
\end{align*}
with the set of the indices, $\mathcal{S}:=\{1,3,4,6,8\}$. 
We note that
\[
\vartheta(S^{(a)}(0, x^{(2)}, \dots, x^{(\nu)}))
=
\begin{cases}
	S^{(a)}(1, x^{(2)}, \dots, x^{(\nu)}) & (a\neq 2, 5, 7);\\
	-S^{(a)}(1, x^{(2)}, \dots, x^{(\nu)}) & (a= 2, 5, 7),
\end{cases}
\]
and
\[
\vartheta(S^{(a)}(-L+1, x^{(2)}, \dots, x^{(\nu)}))
=
\begin{cases}
	S^{(a)}(L, x^{(2)}, \dots, x^{(\nu)}) & (a\neq 2, 5, 7);\\
	-S^{(a)}(L, x^{(2)}, \dots, x^{(\nu)}) & (a= 2, 5, 7),
\end{cases}
\]
where we have used anti-linearity of the reflection map $\vartheta$.
By combining these with the identities (\ref{IdS3}), (\ref{IdSreal}) and (\ref{IdSimaginary}), we obtain 
\begin{equation}
	\label{eq.trid}
	\begin{split}
		\mathrm{Tr} \left(\alpha_n^n\right)
		=
		\int &d\mu(k_1)\cdots d\mu(k_n)
		\mathrm{Tr} \left[\prod_{i=1}^n \left(1 - \frac{\beta}{n} \tilde H^{0}_{K, 1} \right) K_-(k_i, m, h) K_+(k_i, m, h)
		\right] \\
		& \times \prod_{\substack{x \in \Lambda, \\ x^{(1)} =0, L}}  
		\exp\left[i \left\{k^{(1)}_1(x)+ \cdots+ k^{(1)}_n(x)\right\} \sqrt{\frac{\beta g}{2n}}  h^{(1)}(x)\right]
	\end{split}
\end{equation}
with 
\begin{align*}
	\int d\mu(k)
	&:=
	\prod_{\substack{x \in \Lambda, \\ x^{(1)} =0, L}} \prod_{a=1}^8\int_\mathbb{R} \frac{dk^{(a)}(x)}{\sqrt{4\pi}} e^{- k^{(a)}(x)^2/4}.
\end{align*}
The rest of the proof is the same as in~\cite[Proposition~3.1]{GK}. From the expression (\ref{tildeH0K1}) of $\tilde{H}_{K,1}^0$, 
we notice that the operator in the trace in (\ref{eq.trid}) can be written as the sum of the terms, 
each of which has the following form: 
\begin{equation}
	\begin{split}
		\label{Eq.integrand}
		\left(- \frac{i\beta \kappa}{2n} \right)^j&
		K_-(k_1, m, h) K_+(k_1, m, h) \cdots K_-(k_{l_1}, m, h) K_+(k_{l_1}, m, h) \gamma_1 \vartheta(\gamma_1) \\
		& \times K_-(k_{l_1+1}, m, h) K_+(k_{l_1+2}, m, h) \cdots K_-(k_{l_2}, m, h) K_+(k_{l_2}, m, h) \gamma_2 \vartheta(\gamma_2) \\
		& \times K_-(k_{l_2+1}, m, h) K_+(k_{l_2+2}, m, h) \cdots K_-(k_{l_j}, m, h) K_+(k_{l_j}, m, h) \gamma_j \vartheta(\gamma_j)\\
		&\times K_-(k_{l_j+1}, m, h) K_+(k_{l_j+2}, m, h) \cdots K_-(k_{n}, m, h) K_+(k_{n}, m, h),
	\end{split}
\end{equation}
where $\gamma_l \in \{\xi_j(x), \eta_j(x)\}$, for $x \in \Lambda$ with $x^{(1)} = 0$ or $x^{(1)}= -L+1$.
Since both of $\tilde H^- (m, h)$ and $S^{(a)}$ have even fermion parity, $K_-(k, m, h)$ also has even fermion parity 
from the definition. By applying~\cite[Prop.~1]{JP}, one can show that the trace of (\ref{Eq.integrand}) vanishes 
unless there appear an even number of Majorana operators $\gamma_\ell$ on the half lattice.
Thus, there remains only the case that $j$ is even.
Since $K_\pm(k, m, h) \in \mathcal{A}_\pm$ has even fermion parity, 
we see $[K_-, K_+] = 0$, $[K_+, \gamma_i] =0$, and $[K_-, \vartheta(\gamma_i)] =0$.
Combining these with $\vartheta(\gamma_1) \gamma_2 = - \gamma_2\vartheta(\gamma_1)$, we have that
\begin{align*}
	&K_-(k_1, m, h) K_+(k_1, m, h) \cdots K_-(k_{l_1}, m, h) K_+(k_{l_1}, m, h) \gamma_1 \vartheta(\gamma_1) \\
	& \times K_-(k_{l_1+1}, m, h) K_+(k_{l_1+2}, m, h) \cdots K_-(k_{l_2}, m, h) K_+(k_{l_2}, m, h) \gamma_2 \vartheta(\gamma_2) \\
	&=
	-K_-(k_1, m, h)  \cdots K_-(k_{l_1}, m, h) \gamma_1 K_-(k_{l_1 +1 }, m, h) \cdots K_-(k_{l_2}, m, h) \gamma_2  \\
	& \times K_+(k_{1}, m, h)  \cdots K_+(k_{l_1}, m, h) \vartheta(\gamma_1) K_+(k_{l_1 +1 }, m, h) 
	\cdots K_+(k_{l_2}, m, h) \vartheta(\gamma_2)\\
	&=: -X_-(1)X_+(1),
\end{align*}
where $X_\pm(1) \in \mathcal{A}_\pm$.
Hence, for $j = 2m$, the term (\ref{Eq.integrand}) can be written in the form, 
\begin{align*}
	&\left(\frac{\beta \kappa}{2n} \right)^{2m} X_-(1)X_+(1) \cdots X_-(m)X_+(m) \\
	&= \left(\frac{\beta \kappa}{2n} \right)^{2m} X_-(1)X_-(2) \cdots X_-(m) X_+(1) X_+(2) \cdots X_+(m),
\end{align*}
where $X_\pm(i) \in \mathcal{A}_\pm, i=1,2,\ldots,m$, and we have used $[X_-(i), X_+(j)]=0$ for all $i,j$.

Consequently, we have
\begin{equation}
	\label{exprestralphann}
	\begin{split}
		\mathrm{Tr}(\alpha_n^n)
		=\int &d\mu(k_1) \cdots d\mu(k_n) \sum_{j} \mathrm{Tr} \left( W_-(j) W_+(j)\right) \\
		& \times \prod_{\substack{x \in \Lambda, \\ x^{(1)} =0, L}}  
		\exp\left[i \left\{k^{(1)}_1(x)+ \cdots+ k^{(1)}_n(x)\right\} \sqrt{\frac{\beta g}{2n}}  h^{(1)}(x)\right],
	\end{split}
\end{equation}
with
\[
W_\pm(j) := \left(\frac{\beta \kappa}{2n} \right)^{m_j} X_\pm(1, j)X_\pm(2, j) \cdots X_\pm(m_j, j).
\]
Here $m_j \ge 0$ is an integer and $X_{\pm}(i, j) \in \mathcal{A}_\pm$ for $i = 1, \dots, m_j$. 

We next require the following lemma~\cite[Prop.~2]{JP}. 
\begin{lemma}
	\label{lem.RP}
	For any operator $A \in \mathcal{A}_\pm$ it holds that
	\[
	\mathrm{Tr} (A \vartheta(A)) \ge 0.
	\]
\end{lemma}
This allows us to establish the Cauchy--Schwarz inequality 
for the trace: 
\begin{equation}
	\vert\mathrm{Tr}(A \vartheta(B))\vert^2 \le \mathrm{Tr}(A \vartheta(A))\cdot \mathrm{Tr}(B \vartheta(B))
	\quad \mbox{for \ } A, B \in \mathcal{A}_-.
\end{equation} 
Combining this inequality with the Cauchy--Schwarz inequality for the sum and integrations, we can evaluate 
$\mathrm{Tr}(\alpha_n^n)$ of (\ref{exprestralphann}) as follows: 
\begin{align*}
	\vert\mathrm{Tr}(\alpha_n^n)\vert
	&\le
	\int d\mu(k_1) \cdots d\mu(k_n) \sum_{j} \left\vert\mathrm{Tr} \left( W_-(j) W_+(j)\right)\right\vert \\
	&\le
	\int d\mu(k_1) \cdots d\mu(k_n) \sum_{j} \sqrt{\mathrm{Tr}\left[W_-(j) \vartheta (W_-(j))\right]}  
	\sqrt{\mathrm{Tr}\left[W_+(j) \vartheta ( W_+(j))\right]}\\
	&\le
	\int d\mu(k_1) \cdots d\mu(k_n) 
	\left(\sum_j \mathrm{Tr}\left[W_-(j) \vartheta (W_-(j))\right] \right)^{1/2}\\
	&\quad \times \left(\sum_j \mathrm{Tr}\left[W_+(j) \vartheta ( W_+(j))\right] \right)^{1/2} \\
	&\le
	\left(\int d\mu(k_1) \cdots d\mu(k_n) \sum_{j} \mathrm{Tr} \left[ W_-(j) \vartheta(W_-(j))\right]\right)^{1/2} \\
	&\quad \times \left(\int d\mu(k_1) \cdots d\mu(k_n) \sum_{j} \mathrm{Tr} \left[ W_+(j) \vartheta( W_+(j))\right]\right)^{1/2},
\end{align*}
where we have used $W_+(j) = \vartheta (\vartheta (W_+(j)))$ with $\vartheta (W_+(j)) \in \mathcal{A}_-$.

Undoing the above steps, we find 
\begin{equation}
	\begin{split}
		\label{Eq.trace}
		\left\{\mathrm{Tr} \exp\left[- \beta \tilde H^{(\Lambda)}(m, h)\right] \right\}^2
		&\le
		\mathrm{Tr} \exp\left[- \beta ( \tilde H^{(-)}(m, h) +  \vartheta(\tilde H^{(-)}(m, h))+ \tilde H^{(0)}(0)) \right]\\
		& \quad \times 
		\mathrm{Tr} \exp\left[- \beta (\vartheta(\tilde H^{(+)}(m, h))+ \tilde H^{(+)}(m, h)  + \tilde H^{(0)}(0)) \right].
	\end{split}
\end{equation}
For $h = (h^{(1)}, \dots, h^{(\nu)})$, we define $h_\pm = (h_\pm^{(1)}, \dots, h_\pm^{(\nu)})$ by
\begin{equation*}
	h_\pm^{(1)} :=
	\begin{cases}
		h^{(1)}(x) & \mbox{if } x \in \Lambda_\pm, \, x^{(1)}\neq 0,  L; \\
		-h^{(1)}(\vartheta(x+e_1)) & \mbox{if } x \in \Lambda_\mp, \, x^{(1)} \neq 0, L; \\
		0  & \mbox{otherwise, \ }
	\end{cases}
\end{equation*}
and for $i = 2, \dots, \nu$
\begin{equation*}
	h_\pm^{(i)} :=
	\begin{cases}
		h^{(i)}(x) & \mbox{if } x \in \Lambda_\pm; \\
		h^{(i)}(\vartheta(x)) & \mbox{if } x \in \Lambda_\mp.
	\end{cases}
\end{equation*}
Then, by (\ref{Eq.trace}), (\ref{eq.UH}) 
and $U e^{A} U^{-1} = e^{UAU^{-1}}$ for any matrix $A$ and any invertible $U$, we have
\begin{equation}
	\label{Eq.Gaussian0}
	\left\{\mathrm{Tr} \exp\left[-\beta H^{(\Lambda)}(m, h)\right] \right\}^2
	\le  \left\{\mathrm{Tr} \exp\left[-\beta H^{(\Lambda)}(m, h_-)\right]\right\}   
	\left\{\mathrm{Tr} \exp\left[-\beta H^{(\Lambda)}(m, h_+)\right]\right\}. 
\end{equation}

Following~\cite[Sect.~4]{Koma4}, we want to prove that (\ref{Eq.Gaussian0}) holds for reflections across any hyperplane.
To do so, we introduce a unitary operator
\begin{equation}
	\label{eq.gaugeHA}
	U_{\mathrm{HA}}(j \to 1) := U_{\mathrm{HA}}(j, j-1) U_{\mathrm{HA}}(j, j-2) \cdots U_{\mathrm{HA}}(j, 1)
\end{equation}
with
\[
U_{\mathrm{HA}}(i, j)  :=\prod_{k=1}^3 \prod_{\substack{x \in \Lambda\\  x^{(i)}, x^{(j)} 
		= \text{ odd} }} \exp\left[i\pi n_k(x)\right],
\]
for $i \neq j$.
Then one has
\begin{equation}
	\label{eq.invHA}
	U^\dagger_{\mathrm{HA}}(j \to 1) H_K^{(\Lambda)} U_{\mathrm{HA}}(j \to 1)
	=i\kappa \sum_{x \in \Lambda} \sum_{\mu =1}^\nu (-1)^{\tilde \theta_\mu(x)}
	[\Psi^\dagger(x) \Psi(x+e_\mu) - \Psi^\dagger(x+e_\mu)\Psi(x)] 
\end{equation}
where
\begin{equation*}
	\tilde \theta_j(x):=
	\begin{cases}
		0 & \mbox{if \ } x^{(j)}\ne L;\\
		1 & \mbox{if \ } x^{(j)}=L,
	\end{cases}
\end{equation*}
$\tilde \theta_\mu = \theta_\mu$ for $\mu > j$, and for $\mu < j$, 
\begin{equation*}
	\tilde \theta_\mu(x):=
	\begin{cases}
		x^{(1)}+\cdots+x^{(\mu-1)} +x^{(j)} & \mbox{if \ } x^{(\mu)}\ne L;\\
		x^{(1)}+\cdots+x^{(\mu-1)}+ x^{(j)}+1 & \mbox{if \ } x^{(\mu)}= L.
	\end{cases}
\end{equation*}
This shows that the hopping amplitudes are gauge equivalent in all directions.

Next, define the unitary operator by
\begin{equation}
	\label{eq.BC}
	U_{\mathrm{BC}, i} (L \to l) := \prod_{k=1}^3\prod_{\substack{x \in \Lambda\\ l \le x^{(i)} \le L}} \exp\left[i\pi n_k(x)\right].
\end{equation}
If $l \le x^{(i)} \le L$, then, for any $k$,
\begin{equation}
	\label{eq.BCtr}
	U^\dagger_{\mathrm{BC}, i} (L \to l) \psi_k(x) U_{\mathrm{BC}, i} (L \to l) = - \psi_k(x),
\end{equation}
Thus, this transformation changes the antiperiodic boundary condition for the boundary bond 
$\{(-L+1, x^{(2)}, \dots, x^{(\nu)}), (L, x^{(2)}, \dots, x^{(\nu)})  \}$ to that for 
the bond at the location $x^{(1)}=l-1,l$, i.e., $\{(l-1, x^{(2)}, \dots, x^{(\nu)}), (l, x^{(2)}, \dots, x^{(\nu)})  \}$. 
If we put above results together, we obtain that all the reflections across any hyperplane are equivalent 
in the sense of these gauge equivalence.

Finally, we show that $h \equiv 0$ maximizes $\mathrm{Tr} \exp[-\beta H^{(\Lambda)}(m, h)]$ by using above equivalences.
Since $\mathrm{Tr} \exp[-\beta H^{(\Lambda)}(m, h)]$ is bounded and continuous in $h$, there is at least one maximizer.
If a maximizer $h_0$ contains the maximal number of zeros, it should be $h_0^{(i)}(x) = 0$ for all $i, x$, as follows.
If $h_0$ does not vanish for some $i$ and $x$, we can assume that $h_0^{(1)}(L, x^{(2)}, \dots, x^{(\nu)}) \neq 0$ 
for some $(x^{(2)}, \dots, x^{(\nu)})$ by the above arguments.
Then, inequality (\ref{Eq.Gaussian0}) implies that $h_\pm$ obtained from $h_0$ are also maximizers.
This contradicts the assumption on $h_0$, because $h_\pm$ must contain strictly more zero than $h_0$ by definitions.
Thus, one sees $h_0 \equiv 0$.
\end{proof}

%%%%%%%%%%%%%%%%%%%%%%%%%%%%%%%%%%%%%%%%%%%%%%%%%%%%%%%%%%%%%%
\section{Long-range order at non-zero temperature}
\label{sect.LRO1}
In this section we give a proof of the existence of the long-range order for non-zero temperatures in the dimensions $\nu\ge 3$.
To obtain the infrared bound~\cite{FSS, DLS}, we need some more notation. 
By using the unitary transformation $U_{\rm odd}$ of (\ref{eq.odd}), 
let $H_\mathrm{odd}(m, h) := U_\mathrm{odd}^\dagger H^\mathrm{(\Lambda)}(m, h) U_\mathrm{odd}$.
{From} the Gaussian domination bound (\ref{GaussianDbund}), we obtain
\begin{equation}
\label{Eq.Gaussian1}
Z_\mathrm{odd}(h) :=
\mathrm{Tr}\left\{ \exp[-\beta H_\mathrm{odd}(m, h)]\right\}
\le
Z_\mathrm{odd}(0).
\end{equation}
For any pairs of operators $A$ and $B$, we define the Duhamel two-point function by
\[
(A, B) := Z_\mathrm{odd}(0)^{-1}\int_0^1 ds\; 
\mathrm{Tr}\left[ e^{-s \beta H_\mathrm{odd}(m, 0)} A e^{-(1-s)\beta H_\mathrm{odd}(m, 0) } B \right].
\]
Since $U_\mathrm{odd}^\dagger S^{(3)}(x) U_\mathrm{odd} = (-1)^{x_1 + \cdots + x_\nu} S^{(3)}(x)$, one has 
\begin{align*}
U_\mathrm{odd}^\dagger H_{\mathrm{int}, \mu, 3}^{(\Lambda)}(h) U_\mathrm{odd} 
&=
\frac{g}{2} \sum_{x \in \Lambda} [S^{(3)}(x) - S^{(3)}(x+e_\mu) +h^{(\mu)}(x)]^2\\
&=\frac{g}{2} \sum_{x \in \Lambda} \left[(S^{(3)}(x) - S^{(3)}(x+e_\mu))^2 
+ h^{(\mu)}(x)^2\right. \\
&\quad \quad \quad \quad \quad \left. +2S^{(3)}(x) (h^{(\mu)}(x) - h^{(\mu)}(x-e_\mu))\right],
\end{align*}
where $H_{\mathrm{int}, \mu, 3}^{(\Lambda)}(h)$ is given by (\ref{eq.int}), and 
we have used 
$$
\sum_{x \in \Lambda} [S^{(3)}(x) - S^{(3)}(x+e_\mu)]h^{(\mu)}(x) 
= \sum_{x \in \Lambda}S^{(3)}(x) [h^{(\mu)}(x) - h^{(\mu)}(x-e_\mu)].
$$ 

Now we write $\partial_j h^{(\mu)} (x) := h^{(\mu)}(x) -h^{(\mu)}(x-e_j)$ and $S^{(3)}[f] := \sum_x S^{(3)}(x)f(x)$.
Then we show the following inequality: For  any complex-valued functions $h^{(\mu)}$
\begin{equation}
\label{Eq.infrared}
\left(S^{(3)}\left[\overline{\sum_\mu \partial_\mu h^{(\mu)}}\right], S^{(3)}\left[\sum_\mu \partial_\mu h^{(\mu)}\right] \right)
\le
\frac{1}{\beta g} \sum_{\mu=1}^\nu \sum_{x \in \Lambda} \vert h^{(\mu)}(x)\vert^2,
\end{equation} 
where $\overline{z}$ denotes the complex conjugate of $z\in\mathbb{C}$. 
Using $d^2 Z_\mathrm{odd}(\varepsilon h)/d\varepsilon^2 \vert_{\varepsilon = 0} \le 0$ by (\ref{Eq.Gaussian1})  
and the identity followed from Duhamel's formula
\[
\left.\frac{d^2}{d \varepsilon^2} \mathrm{Tr} \left[\exp(-\beta H_\mathrm{odd}(m, 0) + \varepsilon A ) \right] 
\right\vert_{\varepsilon=0}=
(A, A) Z_\mathrm{odd}(0),
\]
we obtain (\ref{Eq.infrared}) for real-valued functions $h^{(\mu)}$.
With the help of the fact that  $(A^\dagger, A) = (A_1, A_1) + (A_2, A_2)$  
for $A = A_1 + i A_2$ with $A_i^\dagger = A_i$, $i=1, 2$, the inequality (\ref{Eq.infrared}) holds 
for any complex-valued functions $h^{(\mu)}$. 

We write $\Lambda^\ast$ for the dual lattice of $\Lambda$. We choose 
$$
h^{(\mu)}(x) = \vert\Lambda\vert^{-1/2} \{\exp[ip \cdot(x+ e_\mu)]- \exp[ip\cdot x]\}
$$ 
with $p= (p^{(1)}, \dots, p^{(\nu)})\in\Lambda^\ast$. Then, we have
\[
\partial_\mu h^{(\mu)}(x) = -2\vert\Lambda\vert^{-1/2}e^{ip\cdot x}(1- \cos p^{(\mu)})
\]
and
\begin{equation}
\label{Ep}
\frac{1}{2}\sum_{x \in \Lambda} \sum_{\mu=1}^\nu \vert h_\mu(x)\vert^2
=\sum_{\mu=1}^\nu(1- \cos p^{(\mu)}) =: E_p.
\end{equation}
For $p \in \Lambda^\ast$, 
let $\tilde S^{(3)}_p := \vert\Lambda\vert^{-1/2} \sum_x S^{(3)}(x) \exp[ip \cdot x]$.
By substituting these into the inequality (\ref{Eq.infrared}), we obtain the desired infrared bound,
\begin{equation}
\label{Eq.IB}
(\tilde S^{(3)}_p, \tilde S^{(3)}_{-p})_{\beta, m} \le \frac{1}{2 \beta g E_{p+Q}},
\end{equation}
where we have used $U_\mathrm{odd}^\dagger S^{(3)}(x) U_\mathrm{odd} = (-1)^{x_1 + \cdots + x_\nu} S^{(3)}(x)$, 
$Q=(\pi, \dots, \pi)$ and the Duhamel two-point function for the Hamiltonian $H^{(\Lambda)}(m)$ in the right-hand side is given by 
\begin{equation}
\label{eq.Duhameldef}
(A, B)_{\beta, m} := \frac{1}{Z_{\beta, m}^{(\Lambda)}}\int_0^1 ds\;  
\mathrm{Tr}\left[ e^{-s \beta H^{(\Lambda)}(m)} A e^{-(1-s)\beta H^{(\Lambda)}(m) } B \right].
\end{equation}

Let $C_p := \langle [\tilde S^{(3)}_{p}, [H^{(\Lambda)}(0), \tilde S^{(3)}_{-p}]] \rangle_{\beta, 0}^{(\Lambda)} $ 
be the expectation value (\ref{eq.expv}) of the double commutator with $m=0$. 
Then $C_p \ge 0$ follows by an eigenfunction expansion (see the next line of \cite[Eq.~(28)]{DLS}).
Applying~\cite[Thm.~3.2 \& Cor~3.2]{DLS} and the infrared bound (\ref{Eq.IB}), we have
\begin{equation}
\label{eq.preLRO}
\begin{split}
	\left\langle \tilde S^{(3)}_{p} \tilde S^{(3)}_{-p}  + \tilde S^{(3)}_{-p} \tilde S^{(3)}_{p} \right\rangle_{\beta, 0}^{(\Lambda)}
	&\le
	\sqrt{\frac{C_p}{2gE_{p+Q}}} \coth\left(\sqrt{\frac{C_p \beta^2 g E_{p+Q} }{2 } }\right) \\
	&\le
	\sqrt{ \frac{C_p}{2 g E_{p+Q}}} + \frac{1}{\beta g E_{p+Q}},
\end{split}
\end{equation}
where we have used the inequality $\coth x \le 1+ 1/x$. 

In order to obtain a lower bound for the long-range order, 
we want to use an upper bound for the expectation value of the interaction Hamiltonian, following \cite{KLS1}. 
Actually, for the left-hand side of (\ref{eq.preLRO}), one has 
\begin{equation}
\label{eq.LHSLRO}
\begin{split}
	&\sum_{p \in \Lambda^\ast} \left\langle \tilde S^{(3)}_{p} \tilde S^{(3)}_{-p}  \right\rangle_{\beta, 0}^{(\Lambda)}\cos p^{(\mu)}
	\\
	&=
	\sum_{p \in \Lambda^\ast} \sum_{x, y \in \Lambda}
	\frac{\left\langle S^{(3)}(x) S^{(3)}(y) \right\rangle_{\beta, 0}^{(\Lambda)}}{ 2\vert\Lambda\vert} \left(e^{ip\cdot(x-(y-e_\mu))} 
	+e^{ip\cdot((x-e_\mu) - y)}\right)\\
	&=
	\sum_{p \in \Lambda^\ast} e^{ip\cdot(x-y)}\sum_{x, y \in \Lambda}\frac{\left\langle S^{(3)}(x) S^{(3)}(y+e_\mu) 
		+ S^{(3)}(x+e_\mu) S^{(3)}(y)  \right\rangle_{\beta, 0}^{(\Lambda)}}{ 2\vert\Lambda\vert}\\
	&=
	\sum_{x \in \Lambda}\left\langle S^{(3)}(x) S^{(3)}(x+e_\mu)  \right\rangle_{\beta, 0}^{(\Lambda)}.
\end{split}
\end{equation}
The rotational symmetry of SU(3) and the spatial symmetry 
imply that this right-hand side is equal to the expectation value of 
the interaction Hamiltonian divided by $8\nu$ except for the coupling constant $g$. 
(For the symmetries, see Appendix~\ref{AppendixSymmetries}.)
From the bound (\ref{eq.preLRO}), one has    
\begin{equation}
\label{Eq.LRO}
\begin{split}
	&\frac{1}{ \vert\Lambda\vert} \sum_{p \in \Lambda^\ast}    \left\langle \tilde S^{(3)}_{p} \tilde S^{(3)}_{-p}  
	+ \tilde S^{(3)}_{-p} \tilde S^{(3)}_{p} \right\rangle_{\beta, 0}^{(\Lambda)}
	\frac{1}{\nu}\sum_{\mu=1}^\nu \left(-\cos p^{(\mu)}\right)\\
	&\le \vert\Lambda\vert^{-1}\sum_{p \neq Q}\left(  \frac{1}{  \beta g E_{p+Q}} 
	+ \sqrt{\frac{C_p}{2g E_{p+Q}}}\right)\frac{1}{\nu} \left(-\sum_{\mu=1}^\nu\cos p^{(\mu)}\right)_+\\
	&\quad	+\frac{2}{\vert\Lambda\vert}{\left\langle \tilde S^{(3)}_{Q} \tilde S^{(3)}_{Q}  \right\rangle_{\beta,  0}^{(\Lambda)}},
\end{split}
\end{equation}
where $F_+ := \max(0, F)$. Further, from (\ref{eq.LHSLRO}) and (\ref{eq.indinv}) in Appendix~\ref{Appendix:SpatialSymmetry}, 
this left-hand side is written 
\begin{equation}
\label{Eq.LROequi}
\begin{split}
	&\frac{1}{ \vert\Lambda\vert} \sum_{p \in \Lambda^\ast}    \left\langle \tilde S^{(3)}_{p} \tilde S^{(3)}_{-p}  
	+ \tilde S^{(3)}_{-p} \tilde S^{(3)}_{p} \right\rangle_{\beta, 0}^{(\Lambda)}
	\frac{1}{\nu}\sum_{\mu=1}^\nu \left(-\cos p^{(\mu)}\right)\\
	&=-2\frac{1}{ \vert\Lambda\vert} \sum_{x \in \Lambda}\left\langle S^{(3)}(x) S^{(3)}(x+e_\mu)  \right\rangle_{\beta, 0}^{(\Lambda)}
\end{split}
\end{equation}
for any $\mu=1,2,\ldots,\nu$. 

We next consider the right-hand side of (\ref{Eq.LRO}).
Using Proposition~\ref{prop.rot}, the last term in the right-hand side of (\ref{Eq.LRO}) 
can be written in terms of the long-range order parameter defined by $m_\mathrm{LRO}^{(\Lambda)} 
:=  \vert\Lambda\vert^{-1}\sqrt{\langle [O^{(\Lambda)}]^2  \rangle_{\beta,  0}^{(\Lambda)}}$.
Indeed, by Proposition~\ref{prop.rot}, the long-range order parameter can be written in the form
\begin{equation}
\label{MLROS3Q}
\begin{split}
	(m_\mathrm{LRO}^{(\Lambda)})^2
	&=
	\vert\Lambda\vert^{-2} \sum_{x, y \in \Lambda} (-1)^{x^{(1)}+\cdots+ x^{(\nu)}}(-1)^{y^{(1)}+\cdots +y^{(\nu)}} 
	\left\langle S^{(2)}(x)S^{(2)}(y) \right\rangle_{\beta,  0}^{(\Lambda)}\\
	&=
	\vert\Lambda\vert^{-2} \sum_{x, y \in \Lambda} (-1)^{x^{(1)}+\cdots +x^{(\nu)}}(-1)^{y^{(1)}+\cdots +y^{(\nu)}} 
	\left\langle S^{(3)}(x)S^{(3)}(y) \right\rangle_{\beta,  0}^{(\Lambda)}\\
	&=
	\vert\Lambda\vert^{-1} \left\langle  \tilde S^{(3)}_{Q} \tilde S^{(3)}_{Q} \right\rangle_{\beta,  0}^{(\Lambda)}.
\end{split}
\end{equation}

In order to estimate the first sum in the right-hand side of (\ref{Eq.LRO}), we need to evaluate the double commutator in $C_p$. 
We first consider the free part $H_K^{(\Lambda)}$ of the present Hamiltonian $H^{(\Lambda)}(m)$ with the mass parameter $m=0$. 
Using the commutation relations, one has 
\begin{equation}
\label{eq.doubleCK}
\begin{split}
	\left\| \left[\tilde S^{(3)}_{p}, \left[H_K^{(\Lambda)}, \tilde S^{(3)}_{-p}\right]\right]   \right\|
	&\le
	4\vert\kappa\vert \nu \|\Psi^\dagger(x) \Psi(y) - \Psi^\dagger(y)\Psi(x)  \|\\
	&\le
	24\vert\kappa\vert \nu. 
\end{split}
\end{equation}
Next, we give a bound for the interaction part $H_\mathrm{int}^{(\Lambda)} = H_\mathrm{int}^{(\Lambda)}(0)$.
By a direct calculation, we have that for $x \neq y$
\begin{equation*}
\label{eq.Dint}
\begin{split}
	\sum_{a=1}^8&\left[S^{(a)}(x)S^{(a)}(y), S^{(3)}(x) \right]\\
	&=
	\sum_{a,b=1}^8\left[if_{a3b}S^{(b)}(x)S^{(a)}(y) \right]\\
	&=
	i\left[-2S^{(2)}(x)S^{(1)}(y) +2S^{(1)}(x)S^{(2)}(y)-S^{(5)}(x)S^{(4)}(y)\right.\\
	&\quad \quad \left.+S^{(4)}(x)S^{(5)}(y)+S^{(7)}(x)S^{(6)}(y)-S^{(6)}(x)S^{(7)}(y) \right],
\end{split}
\end{equation*}
where we have used (\ref{Scommu}) (see also Appendix~\ref{sect.SU(3)} for the value of $f_{abc}$).
Similarly, we have
\begin{equation*}
\label{eq.Dinte}
\begin{split}
	&\left[ S^{(3)}(x), \sum_{a=1}^8\left[S^{(a)}(x)S^{(a)}(y), S^{(3)}(x) \right]\right]\\
	&=
	-\left[4S^{(1)}(x)S^{(1)}(y) +4S^{(2)}(x)S^{(2)}(y)+S^{(4)}(x)S^{(4)}(y)\right.\\
	&\quad \quad \left.+S^{(5)}(x)S^{(5)}(y)+S^{(6)}(x)S^{(6)}(y)+S^{(7)}(x)S^{(7)}(y) \right],
\end{split}
\end{equation*}
and
\begin{equation*}
\label{eq.Dinter}
\begin{split}
	&\left[ S^{(3)}(y), \sum_{a=1}^8\left[S^{(a)}(x)S^{(a)}(y), S^{(3)}(x) \right]\right]\\
	&=
	\left[4S^{(1)}(x)S^{(1)}(y) +4S^{(2)}(x)S^{(2)}(y)+S^{(4)}(x)S^{(4)}(y)\right.\\
	&\quad \quad \left.+S^{(5)}(x)S^{(5)}(y)+S^{(6)}(x)S^{(6)}(y)+S^{(7)}(x)S^{(7)}(y) \right].
\end{split}
\end{equation*}
Since $[S^{(a)}(x), S^{(b)}(y)] =0$ for $x \neq y$, one has
\begin{equation}
\label{eq.Dintera}
\begin{split}
	&\left[\tilde S^{(3)}_{p}, \left[H_\mathrm{int}^{(\Lambda)}, \tilde S^{(3)}_{-p}\right]\right]\\
	&=
	\frac{g}{\vert\Lambda\vert}\sum_{\mu=1}^\nu\sum_{x \in \Lambda}
	\Big[S^{(3)}(p;x,x+e_\mu), 
	\sum_{a=1}^8\left[S^{(a)}(x)S^{(a)}(x+e_\mu), S^{(3)}(-p;x,x+e_\mu) \right]\Bigr]\\
	%&\hspace{110mm} \left. \left.+e^{-ip\cdot (x+e_\mu)} S^{(3)}(x+e_\mu)\right]\right]\\
	&=
	-2g\vert\Lambda\vert^{-1}\sum_{\mu=1}^\nu\left(1-\cos p^{(\mu)}\right)\\
	&\quad \times\sum_{x \in \Lambda}
	%\left
	[4S^{(1)}(x)S^{(1)}(x+e_\mu) +4S^{(2)}(x)S^{(2)}(x+e_\mu)+S^{(4)}(x)S^{(4)}(x+e_\mu)\\%\right.\\
	&\hspace{15mm}\left.+S^{(5)}(x)S^{(5)}(x+e_\mu)+S^{(6)}(x)S^{(6)}(x+e_\mu)+S^{(7)}(x)S^{(7)}(x+e_\mu) \right],
\end{split}
\end{equation}
where we have written 
$$
S^{(3)}(p;x,x+e_\mu):=e^{ip\cdot x} S^{(3)}(x) +e^{ip\cdot (x+e_\mu)} S^{(3)}(x+e_\mu).
$$
Thus, Proposition~\ref{prop.rot} and (\ref{eq.indinv}) imply that
\begin{equation}
\label{eq.Dinteract}
\begin{split}
	&\left\langle\left[\tilde S^{(3)}_{p}, \left[H_\mathrm{int}^{(\Lambda)}, \tilde S^{(3)}_{-p}\right]\right]  
	\right\rangle_{\beta, 0}^{(\Lambda)}\\
	&=
	-\frac{24g}{\vert\Lambda\vert}\sum_{x \in \Lambda}
	\left\langle S^{(3)}(x)S^{(3)}(x+e_1)  \right\rangle_{\beta, 0}^{(\Lambda)} \sum_{\mu=1}^\nu\left(1-\cos p^{(\mu)}\right)\\
	&=-\frac{24gE_p}{\vert\Lambda\vert}\sum_{x \in \Lambda}
	\left\langle S^{(3)}(x)S^{(3)}(x+e_1)  \right\rangle_{\beta, 0}^{(\Lambda)}, 
\end{split}
\end{equation}
where we have used the expression (\ref{Ep}) of $E_p$. 
Combining (\ref{eq.doubleCK}) and (\ref{eq.Dinteract}) yields
\begin{align*}
C_p&=\langle [\tilde{S}_p^{(3)},[H^{(\Lambda)}(0),\tilde{S}_{-p}^{(3)}]]\rangle_{\beta,0}^{(\Lambda)}\\
&=\langle [\tilde{S}_p^{(3)},[H_K^{(\Lambda)},\tilde{S}_{-p}^{(3)}]]\rangle_{\beta,0}^{(\Lambda)}
+\langle [\tilde{S}_p^{(3)},[H_{\rm int}^{(\Lambda)},\tilde{S}_{-p}^{(3)}]]\rangle_{\beta,0}^{(\Lambda)}\\
&\le 24\vert\kappa\vert\nu +24gE_p\mathcal{E}_0^{(\Lambda)},
\end{align*}
where 
\[
\mathcal{E}_0^{(\Lambda)}:= -\vert\Lambda\vert^{-1}\sum_{x \in \Lambda}
\left\langle S^{(3)}(x)S^{(3)}(x+e_1)  \right\rangle_{\beta, 0}^{(\Lambda)}.
\]
We will show $\mathcal{E}_0^{(\Lambda)} \ge 0$ later. From this bound, one has 
\begin{equation}
\label{eq.Cpbound}
\begin{split}
	\sqrt{\frac{C_p}{2  g E_{p+Q}}}&\le 2\sqrt{3}
	\sqrt{\frac{\vert\kappa\vert\nu+gE_p\mathcal{E}_0^{(\Lambda)}}{gE_{p+Q}}}\\
	&\le {2} \sqrt{\frac{3\vert\kappa\vert \nu}{g E_{p+Q}}}
	+ {2\sqrt{3\mathcal{E}_0^{(\Lambda)}}} \sqrt{\frac{E_{p}}{E_{p+Q}}}.
\end{split}
\end{equation}
Substituting this into (\ref{Eq.LRO}), and taking the infinite-volume limit, we obtain
\begin{equation}
\label{eq.LROfin}
\begin{split}
	\mathcal{E}_0
	&:=
	\lim_{\Lambda \nearrow \mathbb Z^{\nu}}\mathcal{E}_0^{(\Lambda)}
	\le \frac{I_\nu}{2\beta g} + \sqrt{\frac{3\vert\kappa\vert \nu}{g} }J_\nu + \sqrt{3\mathcal{E}_0}K_\nu
	+\underbrace{\lim_{\Lambda \nearrow \mathbb Z^{\nu}}(m^{(\Lambda)}_\mathrm{LRO})^2}_{=: m_\mathrm{LRO}^2},
\end{split}
\end{equation}
with  $I_\nu$, $J_\nu$, and $K_\nu$ given by
\begin{align*}
I_\nu &:= \frac{1}{(2\pi)^\nu}\int_{[-\pi, \pi]^\nu} \frac{dp}{E_p}, \quad
J_\nu := \frac{1}{(2\pi)^\nu}\int_{[-\pi, \pi]^\nu} \frac{dp}{\sqrt{E_p}},\\
K_\nu &:=  \frac{1}{(2\pi)^\nu}\int_{[-\pi, \pi]^\nu} \, \frac{dp}{\nu} \sqrt{\frac{E_p}{E_{p+Q}}}
\left(-\sum_{\mu=1}^\nu\cos p^{(\mu)}\right)_+.
\end{align*}
Here, we have used (\ref{Eq.LROequi}), (\ref{MLROS3Q}) and 
$$
0\le \frac{1}{\nu}\left(-\sum_{\mu=1}^\nu\cos p^{(\mu)}\right)_+\le 1.
$$  
Since $I_\nu$ and $J_\nu$ are finite for $\nu \ge 3$, we can prove the existence of the long-range order, 
i.e., $m_\mathrm{LRO}>0$ in the infinite-volume limit if $\mathcal{E}_0$ satisfies
\begin{equation}
\label{eq.gsbound}
\sqrt{\mathcal{E}_0} \left(\sqrt{\mathcal{E}_0}- \sqrt{3}K_\nu\right) >0
\end{equation}
for $\beta$ and $g/\vert\kappa\vert$ both of which are sufficiently large . 

The lower bound for $\mathcal{E}_0$ can be obtained by thermodynamic considerations as in~\cite[Proof of (5.25)]{FILS}.
Using Proposition~\ref{prop.rot}, for any $b \neq 8$, we can find a unitary transformation $U_b$ 
satisfying $[U_b, H^{(\Lambda)}(0)]=0$ and $U_b^\dagger S^{(b)}(x) U_b = S^{(3)}$.
By relying on this property, we consider the following Hamiltonian:
\begin{equation}
\label{eq.ThHam}
\begin{split}
	\tilde{H}_{3,8}^{(\Lambda)}&:= H_K^{(\Lambda)} +7g\sum_{x \in \Lambda}\sum_{\mu=1}^\nu S^{(3)}(x) S^{(3)}(x+e_\mu) 
	+g\sum_{x \in \Lambda}\sum_{\mu=1}^\nu S^{(8)}(x) S^{(8)}(x+e_\mu).
\end{split}
\end{equation}
By $[U_b, H^{(\Lambda)}(0)]=0$ and the SU(3) rotational invariance of the expectation value, 
one has $\langle H^{(\Lambda)}(0) \rangle_{\beta, 0}=\langle \tilde{H}_{3,8}^{(\Lambda)}\rangle_{\beta, 0}$.
Further, using $\|H_K^{(\Lambda)}\| \le C\nu\vert\kappa\vert\vert\Lambda\vert$ with a positive constant $C$ and the spatial symmetry (\ref{eq.indinv}), 
we  have the upper bound
\begin{equation}
\label{eq.PHup}
\begin{split}
	\left\langle  -H^{(\Lambda)}(0)  \right\rangle_{\beta, 0}
	\le& C\nu\vert\kappa\vert\vert\Lambda\vert -7g\nu\sum_{x \in \Lambda} \left\langle S^{(3)}(x) S^{(3)}(x+e_1) \right\rangle_{\beta, 0}\\
	& -g\nu\sum_{x \in \Lambda}\left\langle S^{(8)}(x) S^{(8)}(x+e_1)\right\rangle_{\beta, 0}.
\end{split}
\end{equation}
To obtain the lower bound for $\mathcal{E}_0$, we use the following N\'eel state as a trial state:   
\begin{equation}
\Phi:=\Biggl[\;\prod_{x \in \Lambda_\mathrm{odd}}\psi_1^\dagger(x)\psi_2^\dagger(x)\Biggr]
\Biggl[\;\prod_{y \in \Lambda \backslash \Lambda_\mathrm{odd}} \psi_3^\dagger(y)\Biggr]\vert 0\rangle,
\end{equation}
where $\vert 0\rangle$ is the vacuum for fermions, namely $\psi_i(x) \vert 0\rangle =0 $ for all $i=1,2,3$ and $x \in \Lambda$.
We note that for any $x \in \Lambda$ and $\mu=1, \dots, \nu$
\begin{equation}
\label{eq.Sev}
\begin{split}
	\left\langle \Phi, S^{(a)}(x)S^{(a)}(x+e_\mu) \Phi\right\rangle 
	=
	\begin{cases} 0 & (a\neq 8);\\
		-\frac{4}{3} & (a=8).
	\end{cases}
\end{split}
\end{equation}
We rely on Peierls's inequality which is followed from the convexity (see, e.g., \cite[Proposition~2.5.4]{Ruelle}): 
\begin{lemma}[Peierls's inequality]
Let $A$ be a hermitian matrix and $\{\phi_i\}_i$ an orthonormal family.
Then it holds that
\[
\sum_i \exp\left[-\langle \phi_i, A \phi_i\rangle\right]\le 
\mathrm{Tr} \exp(-A)
\]
\end{lemma}
Using this, (\ref{eq.Sev}) and $\langle \Phi,H_K^{(\Lambda)}\Phi\rangle=0$ for the free part $H_K^{(\Lambda)}$ of 
the present Hamiltonian, we obtain
\begin{align*}
\mathrm{Tr} \,\exp\left[-\beta H^{(\Lambda)}(0)\right]
&\ge
\exp\left(\left\langle \Phi, -\beta H^{(\Lambda)}(0)\Phi \right\rangle\right)\\
&\ge \exp\left[\frac{4}{3}\beta g\nu\vert\Lambda\vert \right].
\end{align*}
Hence
\begin{equation}
\label{eq.PHlow}
\ln \mathrm{Tr} \, \exp\left[-\beta H^{(\Lambda)}(0)\right]
\ge
\frac{4}{3}\beta g\nu\vert\Lambda\vert.
\end{equation}
By the principle of maximum entropy for the Gibbs states (see, e.g.,~\cite[p.~90]{BR}), the following formula holds:
\[
\ln \mathrm{Tr} \, \exp\left[-\beta H^{(\Lambda)}(0)\right]
= \langle -\beta H^{(\Lambda)}(0)\rangle_{\beta, 0}^{(\Lambda)} - \mathrm{Tr} \left[\rho \ln \rho\right],
\]
where $\rho := e^{-\beta H^{(\Lambda)}(0)}/Z_{\beta, 0}^{(\Lambda)}$.
By the concavity for the function $S: t \mapsto -t \ln t$,  we notice that for any $\sum_{j=1}^n \lambda _j=1$
\begin{align*}
-\frac{1}{n} \sum_{j=1}^n\lambda_j \ln \lambda_j
=
\frac{1}{n}\sum_{j=1}^nS(\lambda_j)
&\le
S\left(\sum_{j=1}^n\frac{\lambda_j}{n} \right)
=
-\left(\frac{1}{n} \right) \ln \left(\frac{1}{n} \right)=\frac{1}{n}\ln(n),
\end{align*}
which implies that
\[
-  \mathrm{Tr} \left[\rho \ln \rho\right]
\le \ln \mathrm{Tr}(1) =\ln 2^{3\vert\Lambda\vert}.
\]
Together with (\ref{eq.PHup}) and (\ref{eq.PHlow}), we arrive at
\begin{equation}
\label{eq.fin}
\begin{split}
	&-7g\sum_{x \in \Lambda}\left\langle S^{(3)}(x) S^{(3)}(x+e_1) \right\rangle_{\beta,0}^{(\Lambda)}
	-g\sum_{x \in \Lambda}\left\langle S^{(8)}(x) S^{(8)}(x+e_1) \right\rangle_{\beta,0}^{(\Lambda)}\\
	&\ge
	\left(\frac{4}{3}g- C \vert\kappa\vert\right) \vert \Lambda \vert - \frac{3}{\beta \nu} \vert \Lambda \vert\ln 2.
\end{split}
\end{equation}
Substituting (\ref{eq.S8}) into (\ref{eq.fin}) leads to
\begin{equation}
\label{eq.S1expv}
-\frac{8}{\vert\Lambda\vert}\sum_{x \in \Lambda}\left\langle S^{(3)}(x) S^{(3)}(x+e_1)\right\rangle_{\beta,0}^{(\Lambda)} 
\ge
\frac{4}{3}  -\frac{C \vert\kappa\vert}{g} -\frac{3}{\beta \nu g}\ln 2.
\end{equation}
This shows 
$$
\mathcal{E}_0^{(\Lambda)}=-\frac{1}{\vert\Lambda\vert}\sum_x\langle S^{(3)}(x) S^{(3)}(x+e_1) \rangle_{\beta, 0}^{(\Lambda)} 
\ge \frac{1}{6}-\varepsilon
$$ 
with a small positive $\varepsilon$ which depends on $\vert\kappa\vert/g$ and $1/(\beta g)$ both of which are sufficiently small.
In the spatial dimension $\nu =5$, the numerical value $K_5$ in (\ref{eq.LROfin}) is given by 
$K_5 = 0.2069...$ and thus $\sqrt{3} K_5 = 0.359...$.
Combining this with  (\ref{eq.LROfin}), the long-range order $m_\mathrm{LRO} >0$ exists 
for sufficiently large $\beta g$ and sufficiently small $\vert\kappa\vert/g$ since one has $1/\sqrt{6} = 0.4082...$.
Since $K_\nu$ is monotone decreasing~\cite{KLS2} in the spatial dimension $\nu$, the long-range order exists for $\nu \ge 5$.
This completes the proof. \qed

%%%%%%%%%%%%%%%%%%%%%%%%%%%%%%%%%%%%%%%%%%%%%%%%%%%%%%%%%%%%%%%%%%%
\section{Long-range order at zero temperature}
\label{LRO2}

In this section we prove the existence of long-range order in the ground states~\cite{Neves1986LongRO, KLS1, KLS2} 
for the spatial dimensions $\nu \ge 5$. 

For any observable $A$, the expectation value in the grand state can be defined as
\[\omega_0^{(\Lambda)}(A )
:=\lim_{\beta \to \infty} \langle  A  \rangle_{\beta,  0}^{(\Lambda)}
\]
Taking $\beta \to \infty$ for (\ref{Eq.LRO}), we have 
\begin{equation}
\sqrt{\mathcal{E}_{0,\infty}^{(\Lambda)}}\left(\sqrt{\mathcal{E}_{0,\infty}^{(\Lambda)}}-\sqrt{3}K_\nu^{(\Lambda)}\right)
\le \sqrt{\frac{{3\vert\kappa\vert \nu}}{g}} J_\nu^{(\Lambda)}
+ \frac{1}{\vert\Lambda\vert}\omega_0^{(\Lambda)}\left(\tilde S^{(3)}_{Q} \tilde S^{(3)}_{Q} \right),
\end{equation}
where we have used the sum rule (\ref{Eq.LROequi}), (\ref{eq.Cpbound}) and we have written 
\begin{equation}
\mathcal{E}_{0,\infty}^{(\Lambda)}:=-\frac{1}{\vert\Lambda\vert}\sum_{x\in\Lambda}\omega_0^{(\Lambda)}(S^{(3)}(x)S^{(3)}(x+e_\mu)),
\end{equation}
\begin{equation}
J_\nu^{(\Lambda)}:=\frac{1}{\vert\Lambda\vert} \sum_{p \neq Q} \sqrt{\frac{1}{E_{p+Q}}}
\end{equation}
and 
\begin{equation}
K_\nu^{(\Lambda)}:=\frac{1}{\vert\Lambda\vert}\sum_{p \neq Q} \sqrt{\frac{E_{p}}{E_{p+Q}}} 
\left( -\frac{1}{\nu}\sum_{\mu=1}^\nu\cos p^{(\mu)}\right)_+
\end{equation}

On the other hand, from (\ref{eq.S1expv}), we have
\begin{equation}
\mathcal{E}_{0,\infty}^{(\Lambda)}\ge {\frac{1}{6} -\frac{C \vert\kappa\vert}{8g}}. 
\end{equation}
Hence, in the infinite-volume limit, we obtain
\begin{align*}
m^2_\mathrm{GSLRO}
&:=
\lim_{\Lambda \nearrow \mathbb Z^\nu} \frac{\omega_0^{(\Lambda)}\left(\tilde S^{(3)}_{Q} \tilde S^{(3)}_{Q} \right)}{\vert\Lambda\vert}\\
&\ge\sqrt{\frac{1}{6} -\frac{C \vert\kappa\vert}{8g}}
\left(\sqrt{\frac{1}{6} -\frac{C \vert\kappa\vert}{8g}} - \sqrt{3}K_\nu\right) - \sqrt{\frac{3\vert\kappa\vert\nu}{g}} J_\nu.
\end{align*}
This proves the existence of the long-range order $m_\mathrm{GSLRO} >0$  in the spatial dimension $\nu \ge 5$ 
for sufficiently small $\vert\kappa\vert/g$ because of the bound $1/\sqrt{6}>\sqrt{3}K_\nu$. \qed

%%%%%%%%%%%%%%%%%%%%%%%%%%%%%%%%%%%%%%%%%%%%%%%%%%%%%%%%
\section{The case of SU(2) symmetry}
\label{Sect.SU(2)}

Instead of SU(3) symmetry, we require SU(2) symmetry for the present Hamiltonian. 
Namely, it has two types of fermions, $\psi_1$ and $\psi_2$. Therefore, we write 
\begin{equation}
\Psi(x):=
\begin{pmatrix}
	\psi_1(x)  \\
	\psi_2(x)  \\
\end{pmatrix}
\end{equation}
for the site $x\in\Lambda$, and define only three operators as follows:  
\begin{equation}
S^{(a)}(x):=\Psi^\dagger(x)\lambda^{(a)}\Psi(x)\quad \mbox{for \ } a=1,2,3
\end{equation}
with the usual Pauli matrices. The interaction Hamiltonian is given by 
\begin{equation}
H_{\rm int}^{(\Lambda)}:=g\sum_{x \in \Lambda}\sum_{\mu=1}^\nu\sum_{a=1}^3S^{(a)}(x) S^{(a)}(x+e_\mu). 
\end{equation}
In this case, the same calculation as the case of SU(3) symmetry shows that 
the expectation value (\ref{eq.Dinteract}) of the double commutator is replaced by 
\begin{equation}
\left\langle\left[\tilde S^{(3)}_{p}, \left[H_\mathrm{int}^{(\Lambda)}, \tilde S^{(3)}_{-p}\right]\right]  
\right\rangle_{\beta, 0}^{(\Lambda)}
=-\frac{16gE_p}{\vert\Lambda\vert}\sum_{x \in \Lambda}
\left\langle S^{(3)}(x)S^{(3)}(x+e_1)  \right\rangle_{\beta, 0}^{(\Lambda)}.  
\end{equation}
The difference appears as the prefactor $16$ instead of $24$. Therefore, in the same way, the condition (\ref{eq.gsbound}) 
for the existence of the long-range order is replaced by 
\begin{equation}
\label{E0condition}
\sqrt{\mathcal{E}_0}>\sqrt{2}K_\nu. 
\end{equation}

On the other hand, in order to obtain a lower bound for $\mathcal{E}_0$, we use the usual N\'eel state, 
\begin{equation}
\Phi=\left[\prod_{x \in \Lambda_\mathrm{odd}}\psi_1^\dagger(x)\right]
\left[\prod_{x \in \Lambda \backslash \Lambda_\mathrm{odd}} \psi_2^\dagger(y)\right]\vert 0\rangle, 
\end{equation}
as a trial state. Then, in the same way as in the case of SU(3) symmetry, we have 
\begin{equation}
\mathcal{E}_0>\frac{1}{3}-\varepsilon, 
\end{equation}
where $\varepsilon$ is a small positive number which depends on the strength of the coupling constant $g/\vert\kappa\vert$ 
and the inverse temperature $\beta$. 
For a sufficiently large coupling constant $g/\vert\kappa\vert>0$ and a sufficiently large inverse temperature $\beta$,
the quantity $\varepsilon$ indeed takes a sufficiently small value. 
By combining this with the above condition (\ref{E0condition}), we have the condition for the existence of the long-range order, 
\begin{equation}
1>\sqrt{6}K_\nu, 
\end{equation}
for the strength of the coupling constant $g/\vert\kappa\vert>0$ and the inverse temperature $\beta$ both of which are 
sufficiently large. 
For the spatial dimension $\nu=3$, the numerical value is given by $K_3= 0.3498...$. 
One can easily check that this value satisfies the above bound. Thus, in the spatial dimension $\nu\ge 3$, 
there appears the long-range order for a sufficiently large coupling constant $g/\vert\kappa\vert>0$ 
and a sufficiently large inverse temperature $\beta$. This result also holds for the ground state 
as in the case of SU(3) symmetry.

%%%%%%%%%%%%%%%%%%%%%%%%%%%%%%%%%%%%%%%%%%%%%%
\section{Nambu--Goldstone modes}
\label{sect.NG}
In this section we construct a low energy excitation above the infinite-volume ground state 
as in~\cite{LSM, Koma1, Koma2, Koma3, Koma4, Momoi}.
First of all, we introduce the infinite-volume ground state $\omega_m$.
From now on, $\mathfrak{A}_\Lambda$ denotes the algebra of observables on $\Lambda$.
When two finite lattices, $\Lambda_1$ and $\Lambda_2$, satisfy $\Lambda_1\subset\Lambda_2$, 
the algebra $\mathfrak{A}_{\Lambda_1}$ is embedded in $\mathfrak{A}_{\Lambda_2}$ by 
the tensor product $\mathfrak{A}_{\Lambda_1}\otimes I_{\Lambda_2\backslash\Lambda_1}\subset 
\mathfrak{A}_{\Lambda_2}$ with the identity $I_{\Lambda_2\backslash\Lambda_1}$. 
We define the set of local observables as
\[
\mathfrak{A}_\mathrm{loc}
:=
\bigcup_{\Lambda \subset \mathbb{Z}^\nu;\vert\Lambda\vert<\infty}\mathfrak{A}_\Lambda.
\]
Then the quasi-local algebra, $\mathfrak{A}$, is given by the norm completion of $\mathfrak{A}_\mathrm{loc}$ 
with respect to the operator norm.
Let $\{\Psi^{(\Lambda)}_i \}_{i=1}^d$ denotes the ground states for $H^{(\Lambda)}(m)$ with $d$-multiplicity.
As the zero temperature limit $\beta \to \infty$, the thermal expectation value 
of any local observable $A \in \mathfrak{A}_\Lambda$ converges to the expectation value in the ground states:
\[
\lim_{\beta \to \infty} \langle  A  \rangle_{\beta,  m}^{(\Lambda)}
=
\frac{1}{d}\sum_{i=1}^d \langle \Psi^{(\Lambda)}_i , A  \Psi^{(\Lambda)}_i  \rangle 
=:
\omega_m^{(\Lambda)}(A ).
\]
This fact follows from the standard argument as in (\ref{eq.infPart}) below.
For taking the infinite-volume limit $\Lambda\nearrow\mathbb{Z}^\nu$, 
the Banach--Alaoglu theorem allows us to take a subsequence such that $\omega_m^{(\Lambda)}$ converges to 
a state $\omega_m$ on $\mathfrak{A}$ in the sense of the weak* topology. 
In addition to taking the infinite-volume limit, we take the weak limit of the symmetry breaking field, i.e., $m \searrow 0$. 
Namely, we define the symmetry breaking infinite-volume ground state as
\begin{equation}
\label{eq.infGS}
\omega_0 (A):=\text{weak*-}\lim_{m \searrow 0} \text{weak*-}\lim_{\Lambda \nearrow \mathbb{Z}^\nu}\omega_m^{(\Lambda)}(A)
\end{equation}
for any quasi-local observable $A \in \mathfrak{A}$. Here, as usual, for a finite $\Lambda$ in the subsequence, we need to approximate 
$A$ by an observable $A^{(\Lambda)}\in\mathfrak{A}_\Lambda$, by relying on the operator norm topology.

To calculate the excitation energy above the ground state, we define the time evolution by  
\begin{equation}
\label{alphatmA}
\alpha_{t, m}^{(\Lambda)}(A) := \exp[it H^{(\Lambda)}(m)]A \exp[-itH^{(\Lambda)}(m)]
\end{equation} 
for a local operator $A \in \mathfrak{A}_\Lambda$, and
\begin{equation}
\label{eq.timelocal}
\tau_{f, m}^{(\Lambda)}(A) := \int_{-\infty}^\infty dt\, \hat f(t) \alpha_{t, m}^{(\Lambda)}(A),
\end{equation}
where $ \hat f$ is the Fourier transform of the function $f \in C^\infty(\mathbb{R})$ 
and satisfies $\mathrm{supp}\,  f \subset (0, \gamma)$ with a positive constant $\gamma >0$.
Note that $\hat f$ is a rapidly decreasing function.
Then the following limit exists (see, e.g.,~\cite[Proposition~6.2.3]{BR}):
\[
\lim_{\Lambda \to \mathbb{Z}^\nu}\left\|\alpha_{t, m}^{(\Lambda)}(A)- \alpha_{t, m}(A) \right\| =0
\]
for all  $A \in \mathfrak{A}$, and uniformly for $t$ in a compact set.
In addition, we define
\begin{equation}
\label{eq.limtime}
\tau_{f,0}(A) :=\lim_{m \searrow 0}\lim_{\Lambda \nearrow \mathbb{Z}^\nu}\tau_{f, m}^{(\Lambda)}(A).
\end{equation}
(See Appendix~B of \cite{Koma2} about the existence of the limit $m\searrow 0$.) 
Then the infinite-volume ground state $\omega_0$ has a gapless excitation~\cite{AL} if there is a suitable local operator $A$ such that, for any given $\varepsilon >0$, 
the following holds: 
$$
\omega_0\left(\tau_{f, 0}(A)^\ast \tau_{f, 0}(A)\right) >0
$$ 
and
\begin{equation}
\label{eq.gapless}
\lim_{\Lambda \nearrow \mathbb{Z}^\nu}
\frac{\omega_0\left(\tau_{f, 0}(A)^\ast[\mathcal{H}^{(\Lambda)}(0), \tau_{f, 0}(A)]\right)}{\omega_0\left(\tau_{f, 0}(A)^\ast 
	\tau_{f, 0}(A)\right)}\le \varepsilon,
	\end{equation}
	where we have written $\mathcal{H}^{(\Lambda)}(m) := H^{(\Lambda)}(m) - E^{(\Lambda)}(m)$. 
	
	Now we explain why this condition means the existence of a gapless excitation.
	By GNS representation \cite{BRI}, there exist a Hilbert space $\mathbb{H}_\omega$, a normalized vector $\Omega_\omega$, 
	and a representation $\pi_\omega$ of the observables on $\mathbb{H}_\omega$ such that for any $A \in \mathfrak{A}$
	\[
	\omega_0(A)=(\Omega_\omega, \pi_\omega(A)\Omega_\omega).
	\]
	Moreover, there exists a self-adjoint operator $H_\omega \ge 0$ corresponding to 
	the infinite volume limit of $\mathcal{H}^{(\Lambda)}(0)$ 
	such that
	\[
	H_\omega \Omega_\omega =0, \quad \pi_\omega(\alpha_t(A))=\exp[itH_\omega]\pi_\omega(A)\exp[-itH_\omega].
	\]
	The corresponding state vector in the condition (\ref{eq.gapless}) is written 
	\begin{equation}
\label{eq.excstate}
\begin{split}
	\pi_\omega(\tau_f(A)) \Omega_\omega
	=
	\int_{-\infty}^\infty dt\, \hat f(t)  e^{itH_\omega}\pi_\omega(A) \Omega_\omega.
\end{split}
\end{equation}
Let us consider the spectral decomposition for the Hamiltonian $H_\omega$. 
We write $E_\omega$ for the specral parameter (i.e., energy) for the Hamiltonian $H_\omega$. Then, one has 
\begin{equation}
\int_{-\infty}^\infty dt\, \hat f(t)  e^{itE_\omega}=f(E_\omega). 
\end{equation}
{From} the assumption $\mathrm{supp}\,  f \subset (0, \gamma)$ for the function $f$, these observations imply that 
the state vector $\pi_\omega(\tau_f(A)) \Omega_\omega$ has the energy range  
which is strictly greater than zero unless its norm is zero. 
Therefore, it is an excitation state above the ground state $\Omega_\omega$. Further,  
the condition (\ref{eq.gapless}) is written 
\[
\frac{(\pi_\omega(\tau_{f,0}(A)) \Omega_\omega, H_\omega\pi(\tau_{f,0}(A)) \Omega_\omega) }{(\pi_\omega(\tau_{f,0}(A)) \Omega_\omega, 
\pi_\omega(\tau_{f, 0}(A)) \Omega_\omega) }
\le \varepsilon,
\]
This implies the existence of a gapless excitation above the ground state 
because we can take $\varepsilon$ to be any small value.

As a preparation for the construction of the gapless excitation, we recall the expression (\ref{Hinth}) of 
the interaction part $H_{\mathrm{int}}^{(\Lambda)}(h)$ of the total Hamiltonian $H^{(\Lambda)}(m,h)$ 
and the Gaussian domination bound (\ref{GaussianDbund}). The function $h$ is contained only in 
the Hamiltonian $H_{\mathrm{int}, \mu, 3}^{(\Lambda)}(h)$ of (\ref{eq.int}).  
The explicit form is given by 
\[
H_{\mathrm{int}, \mu, 3}^{(\Lambda)}(h)=
\frac{g}{2} \sum_{x \in \Lambda}
[S^{(3)}(x) +S^{(3)}(x+e_\mu) + (-1)^{\sum_{j=1}^\nu x^{(j)}} h^{(\mu)}(x)]^2.
\] 
For a given real-valued function $h'$ on $\Lambda$, we choose $h^{(1)}(x)$ so as to satisfy 
\begin{equation}
h'(x)=(-1)^{\sum_{j=1}^\nu x^{(j)}} h^{(1)}(x). 
\end{equation}
Namely, we choose 
$$
h^{(1)}(x)=(-1)^{\sum_{j=1}^\nu x^{(j)}} h'(x).
$$
We also choose $h^{(\mu)}=0$ for $\mu=2,\ldots,\nu$. 
Then, we have 
\[
\sum_{\mu=1}^\nu H_{\mathrm{int}, \mu, 3}^{(\Lambda)}(h)=
\frac{g}{2} \sum_{x \in \Lambda}\sum_{\mu=1}^\nu [S^{(3)}(x) +S^{(3)}(x+e_\mu)]^2 
+ g S^{(3)}[\tilde{h}'] +\frac{g}{2} \sum_{x\in\Lambda} [h'(x)]^2,
\] 
where we have written 
\begin{equation}
\label{tildeh'}
\tilde{h}'(x):=h'(x)+h'(x-e_1) 
\end{equation}
and 
\[
S^{(3)}[\tilde{h}']:=\sum_{x \in \Lambda}S^{(3)}(x)\tilde{h}'(x).
\]
By combining this expression of the Hamiltonian and the Gaussian domination bound (\ref{GaussianDbund}), 
we find that 
\begin{equation}
\label{eq.Duhamel}
( S^{(3)}[\tilde{h}'],  S^{(3)}[\tilde{h}'])_{\beta,m}
\le
\frac{1}{\beta g}\sum_{x \in \Lambda}|h'(x)|^2
\end{equation}
in the same way as in the argument of the proof of (\ref{Eq.infrared}), 
where $(A, B)_{\beta, m}$ is the Duhamel two point function given in (\ref{eq.Duhameldef}). 

Let us consider the limit $\beta\to\infty$. 
Since $\langle S^{(3)}(x)\rangle_{\beta, m}^{(\Lambda)}=0$ by (\ref{eq.ph}), we also have
\begin{equation}
\label{GSexpecS3}
\omega_m^{(\Lambda)}(S^{(3)}(x))=0
\end{equation}
for the ground state $\omega_m^{(\Lambda)}$. We use the following well-known relation: 
\begin{equation}
\label{eq.infDuhamel}
\lim_{\beta \to \infty}\beta ( S^{(3)}[\tilde{h}'],  S^{(3)}[\tilde{h}'])_{\beta, m}
\ge 
\omega_m^{(\Lambda)}(S^{(3)}[\tilde{h}'] \mathcal{H}^{(\Lambda)}(m)^{-1}(1-P_0)S^{(3)}[\tilde{h}']),
\end{equation}
where $P_0$ denotes the projection onto the ground-state sector for $\mathcal{H}^{(\Lambda)}(m)$. 
The derivation is given in Appendix~\ref{Appendix:eq.DuhamelEq}. 
By combining this with the above (\ref{eq.Duhamel}), we have 
\begin{equation}
\label{eq.gsinfrared}
\omega_m^{(\Lambda)}(S^{(3)}[\tilde{h}'] \mathcal{H}^{(\Lambda)}(m)^{-1}(1-P_0)S^{(3)}[\tilde{h}'])
\le
\frac{1}{g}\sum_{x \in \Lambda}|h'(x)|^2.
\end{equation}

To construct a trial state, we use a local operator $A_R$ defined by
\begin{equation}
\label{eq.localOp}
A_R:=\frac{1}{\vert \Omega_R\vert}\sum_{x\in \Omega_R}(-1)^{x^{(1)}+\cdots +x^{(\nu)}}S^{(1)}(x),
\end{equation}
where $\Omega_R :=[-R+1, R]^\nu \subset \mathbb{Z}^\nu$ is a hypercubic lattice with given $R  \in \mathbb{Z}_{>0}$.
Then we consider a trial state as in~\cite{Koma3, Koma4}
\begin{equation}
\label{eq.trial}
\phi_{m, \epsilon, R}^{(\Lambda)}(B)
:=
\frac{\omega_m^{(\Lambda)}(A_R [\mathcal{H}^{(\Lambda)}(m)^{\epsilon/2}] B  [\mathcal{H}^{(\Lambda)}(m)^{\epsilon/2}] A_R)}
{\omega_m^{(\Lambda)}(A_R  \mathcal{H}^{(\Lambda)}(m)^{\epsilon} A_R)}
\end{equation}
for an observale $B$, where $\epsilon>0$ is a small parameter.
If there appears a small excitation energy above the ground state energy, 
then its contribution is eliminated by the factor $\mathcal{H}^{(\Lambda)}(m)^{\epsilon/2}$ in (\ref{eq.trial}).
Namely, it allows us to remove the contributions of the undesired low-lying states 
which are degenerate to the sector of the ground state in the infinite-volume limit.
The energy expectation value we will investigate is 
\begin{equation}
\label{eq.trialEn}
\phi_{m, \epsilon, R}^{(\Lambda)}(\mathcal{H}^{(\Lambda)}(m))
=
\frac{\omega_m^{(\Lambda)}(A_R \mathcal{H}^{(\Lambda)}(m)^{1+\epsilon}  A_R)}
{\omega_m^{(\Lambda)}(A_R  \mathcal{H}^{(\Lambda)}(m)^{\epsilon} A_R)}.
\end{equation}

%%%%%%%%%%%%%%%%%%%%%%%%%%%%%%%%%%%%%%%%%%%%%%%%%%%%%%%
\subsection{Lower bound on the denominator}

Our first task is to estimate the denominator of the right-hand side in (\ref{eq.trial}).
We need the following Kennedy--Lieb--Shastry type inequality~\cite{KLS1}. (For the proof, see \cite[Lem.~5.1]{Koma3}.)
\begin{lemma}
\label{lem.KLS}
Let $\epsilon >0$ be a positive small parameter.
Then, for any operators $A, B$ on $\Lambda$, it follows that
\begin{equation}
	\label{eq.KLS}
	\begin{split}
		\left\vert \omega_{m}^{(\Lambda)}([B, A]) \right\vert^2
		&\le
		\sqrt{C_m(B)}\sqrt{\chi(\epsilon) \omega_m^{(\Lambda)}(\{B, B^\dagger\}) 
			+\omega_m^{(\Lambda)}\left(\left[\left[B^\dagger, \mathcal{H}^{(\Lambda)}(m)\right], B \right] \right)} \\
		&\quad \times \left[\omega_m^{(\Lambda)}\left(A \left[ \mathcal{H}^{(\Lambda)}(m) \right]^\epsilon A^\dagger \right) 
		+\omega_m^{(\Lambda)}\left(A^\dagger \left[ \mathcal{H}^{(\Lambda)}(m) \right]^\epsilon A\right) \right],
	\end{split}
\end{equation}
where
\begin{equation}
	\label{eq.defKLS}
	C_m(B):=
	\omega_m^{(\Lambda)}\left(B(1-P_0)\mathcal{H}^{(\Lambda)}(m)^{-1} B^\dagger \right) 
	+\omega_m^{(\Lambda)}\left(B^\dagger(1-P_0) \mathcal{H}^{(\Lambda)}(m)^{-1} B\right).
\end{equation}
Here $\chi(\epsilon)$ is a positive function obeying $\chi(\epsilon) \to 0$ as $\epsilon \to 0$, 
and $P_0$ is the projection onto the sector of the ground states of $\mathcal{H}^{(\Lambda)}(m)$.
\end{lemma}
We choose $A = A_R$ and $B=S^{(3)}[\tilde h']$ with $\tilde{h}'(x):=h'(x)+h'(x-e_1)$ of (\ref{tildeh'}). 
We also choose 
\begin{equation*}
h'(x):=
\begin{cases}
	1 & (x \in \Omega_{R+1}) \\
	1-\frac{|x|_\infty - (R+1)}{R} & (x \in \Omega_{2R}\backslash \Omega_{R+1})\\
	0 & (\text{otherwise}),
\end{cases}
\end{equation*}
where $|x|_\infty := \max_{1\le j\le \nu}|x^{(j)}|$.
Then, using
\[
S^{(3)}[\tilde h']
=
\sum_{x \in \Omega_{2R}} S^{(3)}(x) h'(x) +
\sum_{x \in \Omega_{2R}} S^{(3)}(x) h'(x-e_1),
\]
we have
\begin{equation}
\label{eq.CR}
\begin{split}
	\left[S^{(3)}[\tilde h'], A_R \right]
	&=
	\frac{2}{\vert \Omega_R\vert}\sum_{x \in \Omega_R}(-1)^{x^{(1)}+ \cdots + x^{(\nu)}}\left[S^{(3)}(x), S^{(1)}(x) \right]\\
	&=
	\frac{4i}{\vert \Omega_R\vert}\sum_{x \in \Omega_R}(-1)^{x^{(1)}+ \cdots + x^{(\nu)}}S^{(2)}(x).
\end{split}
\end{equation}
The bound (\ref{eq.gsinfrared}) now becomes
\begin{equation}
\label{eq.GSIB}
\omega_m^{(\Lambda)}(S^{(3)}[\tilde{h}'] \mathcal{H}^{(\Lambda)}(m)^{-1}(1-P_0)S^{(3)}[\tilde{h}'])
\le
\frac{(4R)^\nu}{g}
\end{equation}
To treat the double commutator in (\ref{eq.KLS}), we show the following lemma.

\begin{lemma}
\label{lem.DC}
There are positive constants $C_1$, $C_2$ such that
\[
\left\|
\left[\left[S^{(3)}[\tilde{h}'],  H^{(\Lambda)}(m)\right],  S^{(3)}[\tilde{h}']\right] 
\right\|
\le
C_1R^{\nu-2}+C_2\vert m\vert R^\nu.
\]
\end{lemma}
\begin{proof}
First, we estimate the hopping term because the Hamiltonian $H^{(\Lambda)}(m)$ has the three terms.
We write $\tilde{h'}(x) = h'(x) + h'(x-e_1)$ and
\[
H_{\mathrm{hop}, j}=i\kappa \sum_{x \in \Lambda}\sum_{\mu=1}^\nu (-1)^{\theta_\mu(x)}
\left[\psi_j^\dagger(x)\psi_j(x+e_\mu) -\mathrm{h.c.} \right]
\]
for short.
Since $S^{(3)}(x) = n_1(x) - n_2(x)$ and $[n_i(x), H_{\mathrm{hop}, j}] = 0$ for $i \neq j$, 
it is enough to deal with 
\begin{equation}
	\label{eq.Dkin}
	\left[\left[S^{(3)}[\tilde h'],  H_{\mathrm{hop}, j}\right],  S^{(3)}[\tilde h']\right] 
	=
	\sum_{x, y \in \Omega_{2R}}
	\left[\left[\tilde{h'}(x)n_j(x),  H_{\mathrm{hop}, j}\right],  \tilde{h'}(y)n_j(y)\right] 
\end{equation}
for $j=1,2$. By the anti-commutation relations, we have 
\begin{equation*}
	\label{eq.ACR}
	\left[\psi_i^\dagger(x)\psi_i(y), n_i(x) + n_i(y) \right]=0
\end{equation*}
for any $i$.
Besides, we note that $[[X, Y], Z] = 0$ if $[X, Z]=0 =[Y, Z]$ and for any $G$
\begin{equation}
	\label{eq.change}
	\tilde{h'}(x)G(x) + \tilde{h'}(y)G(y)
	=
	\tilde{h'}(x)(G(x) + G(y))
	+(\tilde{h'}(y) - \tilde{h'}(x))G(y).
\end{equation}
These relations imply that
\begin{equation}
	\label{eq.DCKin}
	\begin{split}
		&\left[\left[\tilde{h'}(x)n_i(x)+\tilde{h'}(y)n_i(y), \psi_i^\dagger(x) \psi_i(y)
		\right], \tilde{h'}(x)n_i(x) + \tilde{h'}(y)n_i(y)
		\right]\\
		&=
		\left[\left[(\tilde{h'}(y)-\tilde{h'}(x))n_i(y), \psi_i^\dagger(x) \psi_i(y)
		\right], \tilde{h'}(x)n_i(x) + \tilde{h'}(y)n_i(y)
		\right]\\
		&=
		\left[\left[(\tilde{h'}(y)-\tilde{h'}(x))n_i(y), \psi_i^\dagger(x) \psi_i(y)
		\right],  (\tilde{h'}(y)-\tilde{h'}(x))n_i(y)
		\right]\\
		&=
		\vert\tilde{h'}(y)-\tilde{h'}(x)\vert^2
		\left[\left[n_i(y), \psi_i^\dagger(x) \psi_i(y)
		\right], n_i(y)
		\right]
	\end{split}
\end{equation}
Using  $\vert\tilde{h'}(y)-\tilde{h'}(x)\vert \le 2R^{-1}$ for $\vert x-y\vert =1$, the bound (\ref{eq.DCKin}) and (\ref{eq.Dkin}) lead to
\begin{equation}
	\label{eq.DCkine}
	\left\|
	\left[\left[S^{(3)}[\tilde h'],  H_{K}^{(\Lambda)}\right],  S^{(3)}[\tilde h']\right] 
	\right\|
	\le
	\frac{{\rm Const.}\times \nu \vert\kappa\vert(4R)^\nu}{R^2}.
\end{equation}

Next, consider the case for the interaction Hamiltonian $H^{(\Lambda)}_{\mathrm{int}}$. 
Similarly, using (\ref{eq.change}), we have
\begin{equation*}
	\begin{split}
		&\left\vert\left[\left[\tilde{h'}(x)S^{(3)}(x)+\tilde{h'}(y)S^{(3)}(y),  \sum_{a=1}^8 S^{(a)}(x)S^{(a)}(y)\right],  \tilde{h'}(x)S^{(3)}(x)+\tilde{h'}(y)S^{(3)}(y)\right]\right\vert\\
		&=
		\left\vert \tilde{h'}(y) - \tilde{h'}(x) \right\vert^2
		\left\vert \left[\left[S^{(3)}(y),  \sum_{a=1}^8 S^{(a)}(x)S^{(a)}(y)\right], S^{(3)}(y)\right] \right\vert\\
		&=
		\left\vert \tilde{h'}(y) - \tilde{h'}(x) \right\vert^2
		\left\vert 4S^{(1)}(x)S^{(1)}(y)+4S^{(2)}(x)S^{(2)}(y)+\sum_{a=4,5,6,7}S^{(a)}(x)S^{(a)}(y) \right\vert,
	\end{split}
\end{equation*}
where we have used (\ref{eq.B2}) for the first equality, and the commutation relations (\ref{Scommu}).
Therefore, we obtain
\begin{equation}
	\label{eq.DCint}
	\left\| \left[ \left[S^{(3)}[\tilde{h}'],  H^{(\Lambda)}_{\mathrm{int}}\right],  S^{(3)}[\tilde{h}']\right] \right\|
	\le
	\frac{{\rm Const.}\times g\nu (4R)^\nu}{R^2}.
\end{equation}

Finally, for the order parameter, a similar reasoning as above shows that
\[
\left\|
\left[\left[S^{(3)}[\tilde h'],  mO^{(\Lambda)}\right],  S^{(3)}[\tilde h']\right]
\right\|
\le
{\rm Const.}\times m(4R)^\nu.
\]
Combining this with (\ref{eq.DCkine}) and (\ref{eq.DCint}), we have the desired result.
\end{proof}

By combining (\ref{eq.GSIB}), Lemma~\ref{lem.KLS}, and Lemma~\ref{lem.DC}, we obtain the following bound:
\begin{equation}
\label{eq.den}
\begin{split}
	\left\vert \omega_{m}^{(\Lambda)}\left(\left[S^{(3)}[\tilde h'], A_R\right]\right)  \right\vert^2
	&\le
	\sqrt{\frac{(4R)^\nu}{2g}}\sqrt{8(4R)^{2\nu}\chi(\epsilon)  +C_1R^{\nu-2}+C_2\vert m\vert R^\nu} \\
	&\quad \times 2\omega_m^{(\Lambda)}\left(A_R \left[ \mathcal{H}^{(\Lambda)}(m) \right]^\epsilon A_R\right),
\end{split}
\end{equation}
where we have used $\omega_m^{(\Lambda)}(\{B, B^\dagger\}) \le 2\|B\|^2$ and $\|S^{(3)}[\tilde h']\| \le 2(4R)^\nu$. 

Now we want to show that the left-hand side on (\ref{eq.den}) can be written by the staggered magnetization.
Let $T_\mu$ be the lattice shift transformation defined by $T_\mu(\psi_i(x)) = \psi_i(x+ 2e_\mu)$ 
and $T_\mu(\psi^\dagger_i(x)) = \psi^\dagger_i(x+ 2e_\mu)$ for any $\mu = 1, \dots, \nu$.
By the periodic boundary condition on the lattice, we obtain
\begin{align*}
\langle  S^{(2)}(x) \rangle_{\beta, m}^{(\Lambda)}
&=
\frac{1}{Z_{\beta, m}^{(\Lambda)}}\mathrm{Tr}\left[ S^{(2)}(x+2e_\mu) \exp(- \beta T_\mu H^{(\Lambda)}(m)) \right].
\end{align*}
Recalling (\ref{eq.BC}) and (\ref{eq.BCtr}), the unitary transformation $U_{\mathrm{BC}, \mu}(L \to l)$, which does not change $S^{(2)}(x)$, allows us to change the boundary conditions of the Hamiltonian $T_\mu H^{(\Lambda)}(m)$.
Therefore, we have
\begin{equation}
\label{eq.2transinv}
\langle  S^{(2)}(x) \rangle_{\beta, m}^{(\Lambda)}
=
\langle  S^{(2)}(x+2e_\mu) \rangle_{\beta, m}^{(\Lambda)}
\end{equation}
for any $\mu$, and hence $\omega_{m}^{(\Lambda)}\left(S^{(2)}(x)\right)$ is also translation invariant with period $2$.
This translational invariance tells us that
\begin{equation}
\label{eq.finitemagnetization}
\begin{split}
	m_s^{(\Lambda)}(m) &:=
	\frac{1}{\vert\Lambda\vert}\sum_{x \in \Lambda}(-1)^{x^{(1)}+ \cdots + x^{(\nu)}}\omega_{m}^{(\Lambda)}\left(S^{(2)}(x)\right)\\
	&=
	\frac{1}{\vert \Omega_R\vert}\sum_{x \in \Omega_R}(-1)^{x^{(1)}+ \cdots + x^{(\nu)}}\omega_{m}^{(\Lambda)}\left(S^{(2)}(x)\right)
\end{split}
\end{equation}
for any positive integer $R$.
Then (\ref{eq.CR}) and (\ref{eq.den}) lead to
\begin{equation}
\label{eq.deno}
\begin{split}
	\frac{8\left\vert m_s^{(\Lambda)}(m)   \right\vert^2}
	{\sqrt{{(4R)^\nu}/{2g}}\sqrt{8(4R)^{2\nu}\chi(\epsilon)  +C_1R^{\nu-2}+C_2\vert m\vert R^\nu}}
	&\le
	\omega_m^{(\Lambda)}\left(A_R \left[ \mathcal{H}^{(\Lambda)}(m) \right]^\epsilon A_R\right),
\end{split}
\end{equation}
The spontaneous magnetization $m_s$ in the infinie-volume is given by
\[
m_s:=\lim_{m \searrow 0} \lim_{\Lambda \nearrow \mathbb{Z}^\nu}m_s^{(\Lambda)}(m) 
=
\frac{1}{\vert \Omega_R\vert}\sum_{x \in \Omega_R}(-1)^{x^{(1)}+ \cdots + x^{(\nu)}}\omega_{0}\left(S^{(2)}(x)\right),
\]
which does not vanish by the existence of the long-range order~\cite{KT}. 
Here, the infinite-volume state $ \omega_0$ is defined in (\ref{eq.infGS}).

%Using (\ref{eq.deno}), we arrive at
%\begin{equation}
%\label{eq.denom}
%\begin{split}
%\frac{8\left\vert m_s   \right\vert^2}
%{\sqrt{\frac{(4R)^\nu}{2g}}\sqrt{8(4R)^{2\nu}\chi(\epsilon)  +C_1R^{\nu-2}}}
%&\le
% \omega_0\left(A_R \left[ \mathcal{H}^{(\Lambda)}(m) \right]^\epsilon A_R\right),
%\end{split}
%\end{equation}
%This is the desired lower bound.
%Here letting $\epsilon \to 0$, the bound (\ref{eq.denom}) implies that a Nambu--Goldstone type slow decay~\cite{Momoi}, 
%which means that $\omega_0(S^{(1)}(x)S^{(1)}(y))$ cannot decay faster than $|x-y|^{-(\nu-1)}$.

%%%%%%%%%%%%%%%%%%%%%%%%%%%%%%%%%%%%%%%%%%%%%%%%%%%%%%%%%%%
\subsection{Upper bound on the numerator}
\label{sebsect:UB}
Next, we estimate the numerator of the right-hand side in (\ref{eq.trialEn}). 
Let $P(E, +\infty)$ be the spectral projection onto the interval $(E, \infty)$ 
for the Hamiltonian $\mathcal{H}^{(\Lambda)}(m)$ and $E>0$.
Denoting $P_E:=1-P(E, \infty)$, we have
\begin{align*}
&\omega_m^{(\Lambda)}\left(A_R\left[ \mathcal{H}^{(\Lambda)}(m) \right]^{1+\epsilon} A_R\right)\\
&=
\omega_m^{(\Lambda)}\left(A_R P_E\left[ \mathcal{H}^{(\Lambda)}(m) \right]^{1+\epsilon} A_R\right)
+
\omega_m^{(\Lambda)}\left(A_R P(E, \infty)\left[ \mathcal{H}^{(\Lambda)}(m) \right]^{1+\epsilon} A_R\right)\\
&\le
\omega_m^{(\Lambda)}\left(A_R P_E\left[ \mathcal{H}^{(\Lambda)}(m) \right] A_R\right)E^\epsilon
+
\omega_m^{(\Lambda)}\left(A_R P(E, \infty)\left[ \mathcal{H}^{(\Lambda)}(m) \right]^3 A_R\right)E^{\epsilon-2}\\
&\le
\omega_m^{(\Lambda)}\left(A_R \mathcal{H}^{(\Lambda)}(m) A_R\right)E^\epsilon
+
\omega_m^{(\Lambda)}\left(A_R \left[ \mathcal{H}^{(\Lambda)}(m) \right]^3 A_R\right)E^{\epsilon-2}.
\end{align*}
For the first term in the right-hand side, from  the anti-commutation relations and the fact 
that $[S^{(a)}(x), S^{(b)}(y)]=0$ for $x \neq y$, 
there is a positive constant $C_3$ such that
\begin{equation}
\label{eq.numfirst}
\omega_m^{(\Lambda)}\left(A_R  \mathcal{H}^{(\Lambda)}(m) A_R\right)
=
\frac{1}{2}\omega_m^{(\Lambda)}\left(\left[ A_R, \left[ H^{(\Lambda)}(m), A_R\right]\right]\right) 
\le 
\frac{C_3}{R^\nu}.
\end{equation}
To get the bound on the second term, let $\delta(A_R) := i[H^{(\Lambda)}(m), A_R]$.
Similar to (\ref{eq.numfirst}), one finds
\begin{align*}
\omega_m^{(\Lambda)}\left(A_R \left[ \mathcal{H}^{(\Lambda)}(m) \right]^3 A_R\right)
&\le
\omega_m^{(\Lambda)}\left(\left[ A_R, H^{(\Lambda)}(m)\right]  \mathcal{H}^{(\Lambda)}(m)  
\left[H^{(\Lambda)}(m), A_R\right]\right)
\\
&=
\omega_m^{(\Lambda)}\left(\delta(A_R) \mathcal{H}^{(\Lambda)}(m)\delta(A_R) \right)\\
&=
\frac{1}{2}\omega_m^{(\Lambda)}\left(\left[ \delta(A_R), \left[ H^{(\Lambda)}(m), \delta(A_R)\right]\right]\right) \\
&\le
\frac{C_4}{R^\nu}.
\end{align*}
Choosing $E=1$, we obtain the upper bound
\begin{equation}
\label{eq.nume}
\omega_m^{(\Lambda)}\left(A_R\left[ \mathcal{H}^{(\Lambda)}(m) \right]^{1+\epsilon} A_R\right)
\le
\frac{C_5}{R^\nu},
\end{equation}
where $C_5 $ is a positive constant.

%%%%%%%%%%%%%%%%%%%%%%%%%%%%%%%%%%%%%%%%%%%%%%%%%%%%%%%%%%%%%%%%%%%%
\subsection{Existence of a gapless excitation}

Substituting (\ref{eq.deno}) and (\ref{eq.nume}) into (\ref{eq.trialEn}), we have
\begin{equation}
\label{eq.trialene}
\begin{split}
	\phi_{m, \epsilon, R}^{(\Lambda)}(\mathcal{H}^{(\Lambda)}(m))
	&\le
	\frac{C_5\sqrt{(4R)^\nu}\sqrt{8(4R)^{2\nu}\chi(\epsilon) 
			+C_1R^{\nu-2}+C_2\vert m\vert R^\nu}}{8\sqrt{2g}\left\vert m_s^{(\Lambda)} (m)  \right\vert^2 R^\nu}. 
\end{split}
\end{equation}
Choosing $\epsilon$ so that 
\begin{equation}
\label{eq.epsilon}
8(4R)^{2\nu}\chi(\epsilon) 
\le R^{\nu-2},
\end{equation}
the bound (\ref{eq.trialene}) becomes

\begin{equation}
\label{eq.trialener}
\begin{split}
	\phi_{m, \epsilon, R}^{(\Lambda)}(\mathcal{H}^{(\Lambda)}(m))
	&\le
	\frac{{2^{\nu-3}}C_5\sqrt{1+C_1+C_2\vert m\vert R^2}}{\sqrt{2g}\left\vert m_s^{(\Lambda)} (m)  \right\vert^2 R}.
\end{split}
\end{equation}
Formally taking the double limit $m \searrow 0$ and $\Lambda  \nearrow \mathbb{Z}^\nu$, we learn
\begin{equation}
\label{eq.trialenerg}
\begin{split}
	\lim_{m \searrow 0}\lim_{\Lambda  \nearrow \mathbb{Z}^\nu}\phi_{m, \epsilon, R}^{(\Lambda)}(\mathcal{H}^{(\Lambda)}(m))
	&\le
	\frac{{2^{\nu-3}}C_5\sqrt{1+C_1}}{\sqrt{2g}\left\vert m_s  \right\vert^2 R}.
\end{split}
\end{equation}
For a large $R$, this suggests the existence of a gapless excitation above the sector of the ground state 
because the contribution of the spectral components having the ground-state energy, zero, 
is vanishing due to the factor $[\mathcal{H}^{(\Lambda)}(m)]^\epsilon$. 

Our final task is to eliminate the singularity of $\phi_{m, \epsilon, R}^{(\Lambda)}$ which comes from 
the function $s^{\epsilon/2}$ of $s$ at zero.
Following~\cite{Koma3}, we first extend $s^{\epsilon/2}$ for $s \ge 0$ to
\[
\eta(s) :=
\begin{cases}
s^{\epsilon/2} & (s\ge 0) \\
0 & (s <0).
\end{cases}
\]
Then, let $f_1 \in C^\infty(\mathbb{R})$ be a smooth function obeying $0 \le f_1 \le 1$ and
\[
f_1(s)
=
\begin{cases}
1 & (s \le r) \\
0 & (s \ge 2r),
\end{cases}
\]
where $0 < r < \infty$. By definition, we have
\begin{equation}
\label{eq.smea}
\omega_m^{(\Lambda)}(A_R\mathcal{H}^{(\Lambda)}(m) \eta\left(\mathcal{H}^{(\Lambda)}(m)\right)^2 
f_1\left(\mathcal{H}^{(\Lambda)}(m)\right)^2 A_R)
\le
\omega_m^{(\Lambda)}(A_R \mathcal{H}^{(\Lambda)}(m)^{1+\epsilon}A_R).
\end{equation}
Next, we approximate the function $\eta f_1$ with a smooth function.
Let $0\le g \le 1$ be a smooth function satisfying
\[
g(s)=
\begin{cases}
1 & (s \ge 2 \delta)\\
0 & (s \le \delta)
\end{cases}
\]
with $\delta>0$, and $f(s) : = g(s) \eta(s)f_1(s)$. The support of the function $f$ is included in $(\delta,2r)$, 
i.e., ${\rm supp}f\subset (\delta,2r)$. Of course, we assume $\delta<2r$. 
Since $\eta(s)$ is smooth for $s >0$, such $f$ is also smooth and satisfies that 
\begin{equation}
\label{eq.smear}
h(s):=
\eta(s)^2 f_1(s)^2
- f(s)^2
= (1-g(s)^2)\eta(s)^2 f_1(s)^2
\le
(2 \delta)^\epsilon.
\end{equation}
Using $0\le f(s) \le \eta(s) f_1(s)$, (\ref{eq.nume}) and (\ref{eq.smea}), we obtain that
\begin{equation}
\label{eq.smeari}
\begin{split}
	\omega_m^{(\Lambda)}(A_R  f\left(\mathcal{H}^{(\Lambda)}(m)\right)\mathcal{H}^{(\Lambda)}(m) f
	\left(\mathcal{H}^{(\Lambda)}(m)\right) A_R)
	&\le
	\omega_m^{(\Lambda)}(A_R\mathcal{H}^{(\Lambda)}(m)^{1+\epsilon}A_R)\\
	&\le 
	\frac{C_5}{R^\nu}.
\end{split}
\end{equation}
%where we have used the inequalities  and (\ref{eq.smear}).
Next, we estimate the quantity $\omega_m^{(\Lambda)}(A_R \mathcal{H}^{(\Lambda)}(m)^{\epsilon}A_R)$. 
We note that for any $r>0$
\begin{equation}
\label{eq.cutoff1}
\eta(s)^2(1-f_1(s)^2)
\le
\frac{s}{r^{1-\epsilon}}.
\end{equation}
Then we obtain
\begin{align*}
\omega_m^{(\Lambda)}(A_R \mathcal{H}^{(\Lambda)}(m)^{\epsilon}A_R)
&=
\omega_m^{(\Lambda)}(A_R \eta\left(\mathcal{H}^{(\Lambda)}(m)\right)^2 f_1\left(\mathcal{H}^{(\Lambda)}(m)\right)^2A_R) \\
&\quad+
\omega_m^{(\Lambda)}(A_R \eta\left(\mathcal{H}^{(\Lambda)}(m)\right)^2
\left( 1- f_1\left(\mathcal{H}^{(\Lambda)}(m)\right)^2 \right)A_R) \\
&\le
\omega_m^{(\Lambda)}(A_R \eta\left(\mathcal{H}^{(\Lambda)}(m)\right)^2 f_1\left(\mathcal{H}^{(\Lambda)}(m)\right)^2A_R) \\
&\quad+\frac{\omega_m^{(\Lambda)}(A_R\mathcal{H}^{(\Lambda)}(m)A_R)}{r^{1-\epsilon}}\\
&\le
\omega_m^{(\Lambda)}(A_R \eta\left(\mathcal{H}^{(\Lambda)}(m)\right)^2 f_1\left(\mathcal{H}^{(\Lambda)}(m)\right)^2A_R)
+
\frac{C_3}{r^{1-\epsilon}R^\nu}.
\end{align*}
where we have used (\ref{eq.numfirst}).
Using (\ref{eq.smear}), the first term in the right-hand side can be bound as
\begin{align*}
\omega_m^{(\Lambda)}(A_R \eta\left(\mathcal{H}^{(\Lambda)}(m)\right)^2 f_1\left(\mathcal{H}^{(\Lambda)}(m)\right)^2A_R) 
&=
\omega_m^{(\Lambda)}(A_R f\left(\mathcal{H}^{(\Lambda)}(m)\right)^2A_R) \\
&\quad+
\omega_m^{(\Lambda)}(A_R h\left(\mathcal{H}^{(\Lambda)}(m)\right)A_R) \\
&\le
\omega_m^{(\Lambda)}(A_R f\left(\mathcal{H}^{(\Lambda)}(m)\right)^2A_R)
+8\delta^\epsilon,
\end{align*}
where we have used $\|A_R\|^2 \le 4$.
Hence there is a small $\delta_0$ such that
\begin{equation}
\label{eq.smearin}
\begin{split}
	\omega_m^{(\Lambda)}(A_R\mathcal{H}^{(\Lambda)}(m)^{\epsilon}A_R)
	&\le
	\omega_m^{(\Lambda)}(A_R f\left(\mathcal{H}^{(\Lambda)}(m)\right)^2A_R)
	+8\delta^\epsilon+\frac{C_3}{r^{1-\epsilon}R^\nu}\\
	&\le
	\omega_m^{(\Lambda)}(A_R f\left(\mathcal{H}^{(\Lambda)}(m)\right)^2A_R) + \delta_0,
\end{split}
\end{equation}
where we have chosen $\delta$ small enough and $r$ large enough.
Combining this with (\ref{eq.deno}) and (\ref{eq.epsilon}),  we arrive at
\begin{equation}
\label{eq.denomi}
\begin{split}
	\frac{8\left\vert m_s^{(\Lambda)}(m)   \right\vert^2}
	{\sqrt{\frac{(4R)^\nu}{2g}}\sqrt{R^{\nu-2} +C_1R^{\nu-2}+C_2\vert m\vert R^\nu}}
	&\le
	\omega_m^{(\Lambda)}(A_R f\left(\mathcal{H}^{(\Lambda)}(m)\right)^2A_R)
	+\delta_0.
\end{split}
\end{equation}
Finally, we show (\ref{eq.gapless}) by rewriting all the results in terms of $\tau_{f, m}^{(\Lambda)}(A_R)$.
Recalling (\ref{alphatmA}) and (\ref{eq.timelocal}), we write the left-hand side of (\ref{eq.smeari}) as
\begin{equation*}
\label{eq.rew}
\begin{split}
	&\omega_m^{(\Lambda)}(A_R  f\left(\mathcal{H}^{(\Lambda)}(m)\right)
	\mathcal{H}^{(\Lambda)}(m) f\left(\mathcal{H}^{(\Lambda)}(m)\right) A_R)\\
	&=
	\omega_m^{(\Lambda)}\left(\tau_{f, m}^{(\Lambda)}(A_R)^\dagger  
	\left[\mathcal{H}^{(\Lambda)}(m), \tau_{f, m}^{(\Lambda)}(A_R) \right]\right).
\end{split}
\end{equation*}
Then the first term $\omega_m^{(\Lambda)}(A_R f(\mathcal{H}^{(\Lambda)}(m))^2A_R)$ in the right-hand side of (\ref{eq.denomi}) 
is written as
\[
\omega_m^{(\Lambda)}\left(A_R f\left(\mathcal{H}^{(\Lambda)}(m)\right)^2A_R\right)
=
\omega_m^{(\Lambda)}\left(\tau_{f, m}^{(\Lambda)}(A_R)^\dagger   \tau_{f, m}^{(\Lambda)}(A_R) \right).
\]
By combining these observations with (\ref{eq.smeari}) and (\ref{eq.denomi}), we have 
\begin{equation}
\begin{split}
	&\frac{\omega_m^{(\Lambda)}\left(\tau_{f, m}^{(\Lambda)}(A_R)^\dagger  
		\left[\mathcal{H}^{(\Lambda)}(m), \tau_{f, m}^{(\Lambda)}(A_R) \right]\right)}
	{\omega_m^{(\Lambda)}\left(\tau_{f, m}^{(\Lambda)}(A_R)^\dagger   \tau_{f, m}^{(\Lambda)}(A_R) \right)}\\
	\le& \frac{C_5}{R^\nu}\frac{\sqrt{{(4R)^\nu}}\sqrt{R^{\nu-2} +C_1R^{\nu-2}+C_2\vert m\vert R^\nu}}
	{{8\sqrt{2g}\Bigl\vert m_s^{(\Lambda)}(m) \Bigr\vert^2}-\delta_0\times \sqrt{{(4R)^\nu}}\sqrt{R^{\nu-2} +C_1R^{\nu-2}+C_2\vert m\vert R^\nu}}. 
\end{split}
\end{equation}
We choose $\delta_0$ so as to satisfy $\sqrt{4^\nu(1+C_1)}\delta_0\cdot R^{\nu-1}<8\sqrt{2g}|m_{\rm s}|^2$, 
where $m_{\rm s}$ is the spontaneous magnetization, which is non-vanishing because of the existence of the long-range order. 
Then, we have 
\begin{equation}
\lim_{m \searrow 0}\lim_{\Lambda \nearrow \mathbb{Z}^\nu}\frac{\omega_m^{(\Lambda)}\left(\tau_{f, m}^{(\Lambda)}(A_R)^\dagger  
	\left[\mathcal{H}^{(\Lambda)}(m), \tau_{f, m}^{(\Lambda)}(A_R) \right]\right)}
{\omega_m^{(\Lambda)}\left(\tau_{f, m}^{(\Lambda)}(A_R)^\dagger   \tau_{f, m}^{(\Lambda)}(A_R) \right)}\le \frac{\rm Const.}{R}. 
\end{equation}
This left-hand side is equal to 
\[
\lim_{\Lambda \nearrow \mathbb{Z}^\nu}
\frac{\omega_0\left(\tau_{f,0}(A_R)^\dagger  \left[\mathcal{H}^{(\Lambda)}(0), \tau_{f,0}(A_R) \right]\right)}
{\omega_0\left(\tau_{f,0}(A_R)^\dagger  \tau_{f,0}(A_R) \right)} 
\]
in the sense of the weak$^\ast$ limit about the state $\omega_m^{(\Lambda)}$. (See Theorem~4.1 in \cite{Koma2}.) 
Here, $\tau_{f,0}(A)$ of $A$ is the limit of $\tau_{f, m}^{(\Lambda)}(A)$ given by (\ref{eq.limtime}). 
Consequently, we obtain 
\begin{equation}
\label{eq.NGA}
\lim_{\Lambda \nearrow \mathbb{Z}^\nu}
\frac{\omega_0\left(\tau_{f,0}(A_R)^\dagger  \left[\mathcal{H}^{(\Lambda)}(0), \tau_{f,0}(A_R) \right]\right)}
{\omega_0\left(\tau_{f,0}(A_R)^\dagger  \tau_{f,0}(A_R) \right)}
\le
\frac{{\rm Const.}}{R}.
\end{equation}
This right-hand side can be made arbitrarily small by choosing $R$ large enough.
This shows the desired bound (\ref{eq.gapless}). \qed

%%%%%%%%%%%%%%%%%%%%%%%%%%%%%%%%%%%%%%%%%%%%%%%%%%%%%%%%%%%%%%%%%%%%%%%%%%%%%%%%
\section{Number of the Nambu--Goldstone modes}
\label{sect:NGn}
In this section, we verify Theorem~\ref{cor.NGn}, i.e., the number of the Nambu--Goldstone modes, $N_\mathrm{NG}$, is actually six.
The first thing we need to do is to count the number of broken symmetries $N_\mathrm{BS}$.
In general, the state $\omega_0(\cdot)$, which is defined in (\ref{eq.infGS}), has continuous symmetry generated by 
$Q^{(a)}:= \sum_{x\in \Lambda} S^{(a)}(x)$ if and only if
\[
\omega_0\left(\left[Q^{(a)}, A  \right]\right)=0
\]
for any local operator $A \in \mathfrak{A}_\mathrm{loc}$.
As shown in Sect.~\ref{LRO2}, the state $\omega_0(\cdot)$ has long-range order and hence
\[
\frac{1}{\vert\Lambda\vert}\sum_{x \in \Lambda}(-1)^{x^{(1)}+\cdots x^{(\nu)}}\omega_0\left(S^{(2)}(x)\right) >0.
\]
With the aid of the commutation relations (\ref{Scommu}), we find that the symmetries $Q^{(a)}$ with $a=1,3,4,5,6,7$ are broken.
On the other hand, since $[H^{(\Lambda)}(m), Q^{(2)}] = 0$, 
we can assume that the ground states of $H^{(\Lambda)}(m)$, $\{\Psi_i^{(\Lambda)}\}_{i=1}^d$, are also eigenvectors of $Q^{(2)}$.
Hence, for any $A \in \mathfrak{A}_\mathrm{loc}$, it holds that
\[
\omega^{(\Lambda)}_m\left(\left[Q^{(2)}, A  \right] \right)
=
\frac{1}{d}\sum_{i=1}^d
\left\langle \Psi^{(\Lambda)}_i,  \left[Q^{(2)}, A  \right] \Psi^{(\Lambda)}_i\right\rangle
=0.
\]
The same relation also holds for $Q^{(8)}$ since  $[H^{(\Lambda)}(m), Q^{(8)}] = 0$.
Therefore, the state $\omega_0(\cdot)$ has symmetries generated by $Q^{(2)}$ and $Q^{(8)}$. 
This argument implies that the number of broken symmetries, $N_\mathrm{BS}$, is equal to six.

Next, we construct the six Nambu--Goldstone modes. One of the modes has been already constructed in 
the preceding section by using $A_R$ of (\ref{eq.localOp}). 
Instead of $A_R$ of (\ref{eq.localOp}), let us consider
\begin{equation*}
A_R^{(a)}:=\frac{1}{\vert \Omega_R\vert}\sum_{x\in \Omega_R}(-1)^{x^{(1)}+\cdots +x^{(\nu)}}S^{(a)}(x)
\quad (a\neq 1, 2, 8).
\end{equation*}
Then we will check that the observable $\tau_{f,0}(A_R^{(a)})$ also gives the corresponding Nambu--Goldstone mode 
similarly to the case of  $A_R=A_R^{(1)}$.

Using (\ref{eq.rot1}) and Proposition~\ref{prop.rot}, we can find a unitary operator $U$ of the form (\ref{eq.globalrot}) 
such that $U^\dagger A_R^{(3)} U = A_R$ and $[U, H^{(\Lambda)}(m)]=0$.
Actually, we set 
\begin{equation}
\label{eq.Unitary2}
U(\theta):=\prod_{x \in \Lambda}\exp\left[i\theta S^{(2)}(x)\right]
\end{equation}
with a real parameter $\theta$. Then we have 
\begin{equation*}
U\left(\frac{\pi}{4}\right)^\dagger A_R^{(3)} U\left(\frac{\pi}{4}\right)=A_R
\end{equation*}
and 
\begin{equation}
\label{eq.Tchange}
U\left(\frac{\pi}{4}\right)^\dagger \tau_{f, m}^{(\Lambda)}(A_R^{(3)}) U\left(\frac{\pi}{4}\right)
=
\tau_{f, m}^{(\Lambda)}(A_R).
\end{equation}
By combining this with the invariance of the state $\omega_m^{(\Lambda)}(\cdot)$, we obtain 
\begin{align*}
\frac{\omega_m^{(\Lambda)}\left(\tau_{f, m}^{(\Lambda)}(A_R^{(3)})^\dagger  
	\left[\mathcal{H}^{(\Lambda)}(m), \tau_{f, m}^{(\Lambda)}(A_R^{(3)}) \right]\right)}
{\omega_m^{(\Lambda)}\left(\tau_{f, m}^{(\Lambda)}(A_R^{(3)})^\dagger   \tau_{f, m}^{(\Lambda)}(A_R^{(3)}) \right)}
&=
\frac{\omega_m^{(\Lambda)}\left(\tau_{f, m}^{(\Lambda)}(A_R)^\dagger  
	\left[\mathcal{H}^{(\Lambda)}(m), \tau_{f, m}^{(\Lambda)}(A_R) \right]\right)}
{\omega_m^{(\Lambda)}\left(\tau_{f, m}^{(\Lambda)}(A_R)^\dagger  \tau_{f, m}^{(\Lambda)}(A_R) \right)}.
\end{align*}
Taking the double limit $\Lambda \nearrow \mathbb{Z}^\nu$ and $m \searrow 0$, 
we obtain the desired low-energy bound (\ref{eq.gapless}) for the observable $A_R^{(3)}$ by using the bound (\ref{eq.NGA}).

For the rest $A_R^{(a)}$, i.e., $a=4,5,6,7$, we recall the proof in Sect.~\ref{sect.NG}. 
Instead of the expression (\ref{Hinth}) for the interaction Hamiltonian, 
we introduce an interaction Hamiltonian as follows: 
\begin{equation}
\label{eq.Generalint}
\begin{split}
	H_{\mathrm{int}, 4}^{(\Lambda)}(h)
	&:=
	\frac{g}{2}\sum_{\mu=1}^\nu H_{\mathrm{int}, \mu, 4}^{(\Lambda)}(h)
	+\frac{g}{2}\sum_{x \in \Lambda}\sum_{\mu=1}^\nu\sum_{b = 1,3,6,8}[S^{(b)}(x)+S^{(b)}(x+e_\mu)]^2\\
	&\quad -\frac{g}{2}\sum_{x \in \Lambda}\sum_{\mu=1}^\nu\sum_{b =2, 5, 7}[S^{(b)}(x)- S^{(b)}(x+e_\mu)]^2\\
	&\quad -g\nu\sum_{a\neq 2, 5, 7}\sum_{x \in \Lambda}S^{(a)}(x)^2
	+g\nu\sum_{b =2, 5, 7}\sum_{x \in \Lambda}S^{(b)}(x)^2,
\end{split}
\end{equation}
where $h= (h^{(1)}, \dots, h^{(\nu)})$ and
\begin{equation}
\label{eq.hpart}
H_{\mathrm{int}, \mu, 4}^{(\Lambda)}(h)
:=
\frac{g}{2} \sum_{x \in \Lambda}
[S^{(4)}(x) +S^{(4)}(x+e_\mu) + (-1)^{\sum_{j=1}^\nu x^{(j)}} h^{(\mu)}(x)]^2.
\end{equation}
Namely, we replace the role of the operator $S^{(3)}(x)$ by $S^{(4)}(x)$. 
By a similar argument in Section~\ref{sect.RP}, we can obtain the corresponding Gaussian domination bound so that
for any real-valued $h = (h^{(1)}, \dots, h^{(\mu)})$
\begin{align}
\label{GaussianDbund4}
\mathrm{Tr} \left\{\exp\left[- \beta H^{(\Lambda)}_4(m, h)\right] \right\}
&\le
\mathrm{Tr}\left\{\exp\left[- \beta H^{(\Lambda)}(m)\right]\right\},
\end{align}
where we have defined
\[
H^{(\Lambda)}_4(m, h):=  H_K^{(\Lambda)} + H_{\mathrm{int},4}^{(\Lambda)}(h) + mO^{(\Lambda)}.
\]
Recalling
\begin{equation*}
h'(x):=
\begin{cases}
	1 & (x \in \Omega_{R+1}) \\
	1-\frac{|x|_\infty - (R+1)}{R} & (x \in \Omega_{2R}\backslash \Omega_{R+1})\\
	0 & (\text{otherwise}),
\end{cases}
\end{equation*}
we choose $h(x) = (h^{(1)}, 0, 0,\dots)$ and take a function $h^{(1)}$ as
\[
h^{(1)}(x):=(-1)^{\sum_{j=1}^\nu x^{(j)}} h'(x).
\]
This implies that
\[
\sum_{\mu=1}^\nu H_{\mathrm{int}, \mu, 4}^{(\Lambda)}(h)=
\frac{g}{2} \sum_{x \in \Lambda}\sum_{\mu=1}^\nu [S^{(4)}(x) +S^{(4)}(x+e_\mu)]^2 
+ g S^{(4)}[\tilde{h}'] +\frac{g}{2} \sum_{x\in\Lambda} [h'(x)]^2,
\] 
where 
\[
\tilde{h}'(x):=h'(x)+h'(x-e_1), \quad  S^{(4)}[\tilde{h}']:=\sum_{x \in \Lambda}S^{(4)}(x)\tilde{h}'(x).
\]
Then, as in the proof of (\ref{eq.Duhamel}), we also obtain
\begin{equation}
\label{eq.IBS4}
( S^{(4)}[\tilde{h}'],  S^{(4)}[\tilde{h}'])_{\beta,m}
\le
\frac{1}{\beta g}\sum_{x \in \Lambda}|h'(x)|^2.
\end{equation}
In the rest of the proof of the existence of a gapless excitation, 
the only difference from the case of Section~\ref{sect.NG} is Lemma~\ref{lem.DC}.
More precisely, we need the following lemma which is a slight modification of Lemma~\ref{lem.DC}.

\begin{lemma}
\label{lem.DCmodify}
There are positive constants $C_1$, $C_2$ such that
\[
\left\|
\left[\left[S^{(4)}[\tilde{h}'],  H^{(\Lambda)}(m)\right],  S^{(4)}[\tilde{h}']\right] 
\right\|
\le
C_1R^{\nu-2}+C_2\vert m\vert R^\nu.
\]
\end{lemma}

\begin{proof}
The proof is almost the same as that of Lemma~\ref{lem.DC}.
Recalling
\[
H_{\mathrm{hop}, j}:=i\kappa \sum_{x \in \Lambda}\sum_{\mu=1}^\nu (-1)^{\theta_\mu(x)}
\left[\psi_j^\dagger(x)\psi_j(x+e_\mu) -\mathrm{h.c.} \right],
\]
we first estimate the hopping term.
Since $S^{(4)}(x) = \psi_1^\dagger(x)\psi_3(x) +\psi_3^\dagger(x)\psi_1(x)$ and $[\psi_i^\dagger(x)\psi_j(x), H_{\mathrm{hop}, k}] = 0$ 
for $i, j\neq k$, 
it is sufficient to estimate
\begin{equation}
	\begin{split}
		\label{eq.Dkinmod}
		&\left[\left[S^{(4)}[\tilde h'],  H_{\mathrm{hop}, 1} +H_{\mathrm{hop}, 3}\right],  S^{(4)}[\tilde h']\right] \\
		&=
		\sum_{x, y \in \Omega_{2R}}
		\left[\left[\tilde{h'}(x)S^{(4)}(x), H_{\mathrm{hop}, 1} +H_{\mathrm{hop}, 3}\right],  \tilde{h'}(y)S^{(4)}(y)\right] 
	\end{split}
\end{equation}
for $j=1,2$. By the anti-commutation relations, we have 
\begin{align*}
	\left[\psi_1^\dagger(x)\psi_1(y), S^{(4)}(x) + S^{(4)}(y) \right]
	&=
	\left[\psi_1^\dagger(x)\psi_1(y), \psi_3^\dagger(x)\psi_1(x)+\psi_1^\dagger(y)\psi_3(y) \right]\\
	&=
	\psi_1^\dagger(x)\psi_3(y)-\psi_3^\dagger(x)\psi_1(y)
\end{align*}
and
\begin{align*}
	\left[\psi_3^\dagger(x)\psi_3(y), S^{(4)}(x) + S^{(4)}(y) \right]
	&=
	\left[\psi_3^\dagger(x)\psi_3(y), \psi_3^\dagger(x)\psi_1(x)+\psi_1^\dagger(y)\psi_3(y) \right]\\
	&=
	-\psi_1^\dagger(x)\psi_3(y)+\psi_3^\dagger(x)\psi_1(y),
\end{align*}
which shows
\[
\left[\psi_1^\dagger(x)\psi_1(y)+\psi_3^\dagger(x)\psi_3(y), S^{(4)}(x) + S^{(4)}(y) \right]=0.
\]
Let $F(x, y):= \psi_1^\dagger(x)\psi_1(y)+\psi_3^\dagger(x)\psi_3(y)$.
In the same way as (\ref{eq.DCKin}), equation (\ref{eq.change}) implies that
\begin{equation}
	\label{eq.DCKinmod}
	\begin{split}
		&\left[\left[\tilde{h'}(x)S^{(4)}(x)+\tilde{h'}(y)S^{(4)}(y), F(x, y)
		\right], \tilde{h'}(x)S^{(4)}(x) + \tilde{h'}(y)S^{(4)}(y)
		\right]\\
		&=
		\left[\left[(\tilde{h'}(y)-\tilde{h'}(x))S^{(4)}(y), F(x, y)
		\right],  (\tilde{h'}(y)-\tilde{h'}(x))S^{(4)}(y)
		\right]\\
		&=
		\vert\tilde{h'}(y)-\tilde{h'}(x)\vert^2
		\left[\left[S^{(4)}(y), F(x, y)
		\right], S^{(4)}(y)
		\right]
	\end{split}
\end{equation}
Using  $\vert\tilde{h'}(y)-\tilde{h'}(x)\vert \le 2R^{-1}$ for $\vert x-y\vert =1$, (\ref{eq.DCKinmod}) and (\ref{eq.Dkinmod}) yield
\begin{equation}
	\label{eq.DCkinemod}
	\left\|
	\left[\left[S^{(4)}[\tilde h'],  H_{K}^{(\Lambda)}\right],  S^{(4)}[\tilde h']\right] 
	\right\|
	\le
	{\rm Const.}
	\frac{ \nu \vert\kappa\vert(4R)^\nu}{R^2}.
\end{equation}

We next estimate the interaction part. 
Using (\ref{eq.change}) and (\ref{eq.B2}), we have
\begin{equation*}
	\begin{split}
		&\left\vert\left[\left[\tilde{h'}(x)S^{(4)}(x)+\tilde{h'}(y)S^{(4)}(y),  \sum_{a=1}^8 S^{(a)}(x)S^{(a)}(y)\right],  \tilde{h'}(x)S^{(4)}(x)+\tilde{h'}(y)S^{(4)}(y)\right]
		\right\vert\\
		&=
		\left\vert \tilde{h'}(y) - \tilde{h'}(x) \right\vert^2
		\left\vert\left[\left[S^{(4)}(y),  \sum_{a=1}^8 S^{(a)}(x)S^{(a)}(y)\right], S^{(4)}(y)\right]
		\right\vert\\
		&=
		\left\vert \tilde{h'}(y) - \tilde{h'}(x) \right\vert^2
		\left\vert 4S^{(5)}(x)S^{(5)}(y)+3S^{(8)}(x)S^{(8)}(y)+\sum_{a=1,2,3,6,7}S^{(a)}(x)S^{(a)}(y) \right\vert,
	\end{split}
\end{equation*}
where we have used the commutation relations (\ref{Scommu}).
Therefore, we obtain
\begin{equation}
	\label{eq.DCintmod}
	\left\| \left[ \left[S^{(4)}[\tilde{h}'],  H^{(\Lambda)}_{\mathrm{int}}\right],  S^{(4)}[\tilde{h}']\right] \right\|
	\le
	{\rm Const.}
	\frac{g\nu (4R)^\nu}{R^2}.
\end{equation}

Similarly, for the order parameter, we obtain
\[
\left\|
\left[\left[S^{(4)}[\tilde h'],  mO^{(\Lambda)}\right],  S^{(4)}[\tilde h']\right]
\right\|
\le
{\rm Const.}\times m(4R)^\nu.
\]
Combining this with (\ref{eq.DCkinemod}) and (\ref{eq.DCintmod}), we have the desired result.
\end{proof}
Since we have not used any peculiar information of $S^{(3)}$ elsewhere in Section~\ref{sect.NG}, 
we can show that $\tau_{f, 0}(A_R^{(4)})$ gives a gapless excitation by the same argument.
The proof above gives more, namely $\tau_{f, 0}(A_R^{(6)})$ yields a Nambu--Goldstone mode 
in the same way as for $A_R^{(3)}$ by replacing $\pi/4$ in (\ref{eq.Tchange}) with $\pi/2$.

For $\tau_{f, 0}(A_R^{(5)})$ and $\tau_{f, 0}(A_R^{(7)})$, we need the spin-$1$ representation of the interaction 
because $S^{(2)}$, $S^{(5)}$ and $S^{(7)}$ are pure imaginary hermitian.
To obtain the Gaussian domination bound, we introduce a unitary operator, 
\[
V :=-\frac{i}{\sqrt{2}}
\left(\begin{matrix}1 & 0 & -1\\i & 0 & i\\0 & \sqrt{2} & 0\end{matrix}\right),
\]
and consider the transformation $\Psi \to V \Psi$. Since the operator $V$ is unitary, 
the spectral decomposition implies that there exists a hermitian matrix $A$ such that 
the operator $V$ can be written $V = \exp[iA]$. As is well-known, there exists a unitary operator 
${\tilde V}$ on the fermion Fock space such that ${\tilde V}$ induces the same transformation as $V$.  
The explicit form\footnote{In order to prove the equivalence between $V$ and $\tilde V$, consider transformations, 
$V_t:=\exp[itA]$ and ${\tilde V}_t :=\prod_{x \in \Lambda}\exp[it\Psi^\dagger(x) A \Psi(x)]$, with the real parameter $t$.
It is enough to show the equivalence between $V_t\Psi(x)$ and ${\tilde V}_t^\dagger \Psi(x){\tilde V}_t$. 
By differentiating with respect to $t$, one can obtain the two differential equations with the same form. 
They have also the same initial value at $t=0$. Therefore, the uniquness of their solutions implies the desired result at $t=1$.} 
is given by $\tilde V :=\prod_{x \in \Lambda}\exp[i\Psi^\dagger(x) A \Psi(x)]$. Then, we have
\begin{equation}
\begin{split}
	\label{eq.VHtransform}
	\tilde V^\dagger H^{(\Lambda)}(m) \tilde V 
	&=i\kappa \sum_{x \in \Lambda} \sum_{\mu=1}^\nu(-1)^{\theta_\mu(x)}
	[\Psi^\dagger(x)\Psi(x+e_\mu) - \Psi^\dagger(x+e_\mu)\Psi(x)] \\
	&\quad+
	m\tilde{O}^{(\Lambda)}+ 
	g\sum_{x \in \Lambda}\sum_{\mu=1}^\nu\sum_{a=1}^8\tilde S^{(a)}(x)\tilde S^{(a)}(x+e_\mu)
\end{split}
\end{equation}
with 
$$
\tilde{O}^{(\Lambda)}:=\tilde V^\dagger O^{(\Lambda)}\tilde V=\sum_{x\in \Lambda}(-1)^{x^{(1)} + \cdots +x^{(\nu)}}{\tilde  S}^{(2)}(x)
$$
and 
$$
\tilde S^{(a)} := \Psi^\dagger V^\dagger \lambda^{(a)}V \Psi.
$$
The explicit forms of the transformed Gell-Mann matrices are given by  
\begin{align*}
V^\dagger \lambda^{(1)} V &=\left(\begin{matrix}0 & 0 & i\\0 & 0 & 0\\- i & 0 & 0\end{matrix}\right), \quad
V^\dagger \lambda^{(2)} V =\left(\begin{matrix}1 & 0 & 0\\0 & 0 & 0\\0 & 0 & -1\end{matrix}\right), \quad
V^\dagger \lambda^{(3)} V =\left(\begin{matrix}0 & 0 & -1\\0 & 0 & 0\\-1 & 0 & 0\end{matrix}\right)\\
V^\dagger \lambda^{(4)} V &=\frac{1}{\sqrt{2}} \left(\begin{matrix}0 & 1& 0\\ 1 & 0 & - 1\\0 & - 1& 0\end{matrix}\right), \quad
V^\dagger \lambda^{(5)} V =\frac{1}{\sqrt{2}} \left(\begin{matrix}0 & - i & 0\\ i & 0 & - i\\0 & i& 0\end{matrix}\right),\\
V^\dagger \lambda^{(6)} V &= \frac{1}{\sqrt{2}}  \left(\begin{matrix}0 & - i & 0\\ i& 0 & i \\0 & - i & 0\end{matrix}\right),\quad
V^\dagger \lambda^{(7)} V =\frac{1}{\sqrt{2}}  \left(\begin{matrix}0 & - 1 & 0\\- 1& 0 & - 1\\0 & - 1 & 0\end{matrix}\right),\\
V^\dagger \lambda^{(8)} V &=\frac{1}{\sqrt{3}} \left(\begin{matrix}1& 0 & 0\\0 & - 2 & 0\\0 & 0 & 1\end{matrix}\right).
\end{align*}
Thus, the three operators, $V^\dagger \lambda^{(2)} V$, $V^\dagger \lambda^{(5)} V$ and $V^\dagger \lambda^{(7)} V$, 
are corresponding to the spin-1 operators. In particular, $V^\dagger \lambda^{(7)} V$ is real hermitian. 
This is crucial in the following argument. 

Noting that
\begin{align*}
\tilde S^{(1)} &=\frac{i}{\sqrt{2}} \left(\begin{matrix}  \psi_1^\dagger \psi_3-  \psi_3^\dagger \psi_1 \end{matrix}\right),
\quad
\tilde  S^{(2)} =\frac{1}{\sqrt{2}} \left(\begin{matrix}n_1 - n_3\end{matrix}\right)\\
\tilde  S^{(3)} &=-\frac{1}{\sqrt{2}} \left(\begin{matrix}  \psi_3^\dagger\psi_1 +  \psi_1^\dagger \psi_3\end{matrix}\right),\quad
\tilde  S^{(4)} =\frac{1}{\sqrt{2}} \begin{matrix} \left(\psi_2^\dagger \psi_1 +  \left(\psi_1^\dagger - \psi_3^\dagger\right)\psi_2 - \psi_2^\dagger\psi_3 \right)\end{matrix}\\
\tilde  S^{(5)} &=\frac{i}{\sqrt{2}} \begin{matrix}\left( \psi_2^\dagger\psi_1  - \left(\psi_1^\dagger - \psi_3^\dagger\right)\psi_2  - \psi_2^\dagger \psi_3 \right)\end{matrix}\\
\tilde  S^{(6)} &=\frac{i}{\sqrt{2}} \begin{matrix}\left( \psi_2^\dagger \psi_1 -  \left(\psi_1^\dagger + \psi_3^\dagger\right)\psi_2 + \psi_2^\dagger \psi_3 \right)\end{matrix}\\
\tilde  S^{(7)} &=-\frac{1}{\sqrt{2}} \begin{matrix}\left( \psi_2^\dagger \psi_1 + \left(\psi_1^\dagger + \psi_3^\dagger\right)\psi_2 +  \psi_2^\dagger \psi_3\right)\end{matrix},\quad
\tilde  S^{(8)} =\frac{1}{\sqrt{3}}\begin{matrix} \left(n_1- 2 n_2 + n_3\right)\end{matrix}
\end{align*}
and using (\ref{eq.odd2}), we learn
\[
U^\dagger_\mathrm{odd}
\tilde S^{(a)}(x)
U_\mathrm{odd}
=
\begin{cases}
-\tilde S^{(a)}(x) & (x \in \Lambda_\mathrm{odd}, a\neq  1, 5, 6)\\
\tilde S^{(a)}(x) & (\text{otherwise}).
\end{cases}
\]
Similarly to (\ref{eq.Generalint}), we define
\begin{equation*}
\begin{split}
	\label{eq.Generalint7}
	\tilde H_{\mathrm{int}, 7}^{(\Lambda)}(h)
	&:=
	\frac{g}{2}\sum_{\mu=1}^\nu \tilde H_{\mathrm{int}, \mu, 7}^{(\Lambda)}(h)
	+\frac{g}{2}\sum_{x \in \Lambda}\sum_{\mu=1}^\nu\sum_{b = 2, 3, 4,8}[\tilde S^{(b)}(x)+ \tilde S^{(b)}(x+e_\mu)]^2\\
	\quad& -\frac{g}{2}\sum_{x \in \Lambda}\sum_{\mu=1}^\nu\sum_{b =1, 5, 6 }[\tilde S^{(b)}(x)- \tilde S^{(b)}(x+e_\mu)]^2\\
	&\quad-g\nu\sum_{a\neq 1, 5, 6}\sum_{x \in \Lambda}\tilde S^{(a)}(x)^2
	+g\nu\sum_{b =1, 5, 6}\sum_{x \in \Lambda}\tilde S^{(b)}(x)^2,
\end{split}
\end{equation*}
where 
\begin{equation*}
\tilde H_{\mathrm{int}, \mu, 7}^{(\Lambda)}(h)
:=
\frac{g}{2} \sum_{x \in \Lambda}
[\tilde S^{(7)}(x) +\tilde S^{(7)}(x+e_\mu) + (-1)^{\sum_{j=1}^\nu x^{(j)}} h^{(\mu)}(x)]^2.
\end{equation*}
Then, for $\tilde U_1 = U_1 U_\mathrm{odd}$, we find
\begin{align*}
\tilde U^\dagger_1 \tilde H_{\mathrm{int}, \mu, 7}^{(\Lambda)}(h)\tilde U_1
&=
\frac{g}{2}\sum_{\mu=1}^\nu \sum_{x \in \Lambda}
\left[\tilde S^{(7)}(x) -\tilde S^{(7)}(x+e_\mu) + h^{(\mu)}(x)\right]^2 \\
&\quad +\frac{g}{2}\sum_{x \in \Lambda}\sum_{\mu=1}^\nu\sum_{b = 2, 3, 4,8}\left[\tilde S^{(b)}(x)-\tilde S^{(b)}(x+e_\mu)\right]^2
\\
&\quad -\frac{g}{2}\sum_{x \in \Lambda}\sum_{\mu=1}^\nu\sum_{b =1,5,6}\left[\tilde S^{(b)}(x)- \tilde S^{(b)}(x+e_\mu)\right]^2.
\\
&\quad-g\nu\sum_{a\neq 1,5,6}
\sum_{x \in \Lambda}\tilde S^{(a)}(x)^2
+g\nu\sum_{b =1,5,6}\sum_{x \in \Lambda}\tilde S^{(b)}(x)^2
\end{align*}
This is the same form as (\ref{tildeHinth}) since $S^{(a)}$ $(a=1,5,6)$ are pure imaginary 
and $S^{(a)}$ $(a\neq 1,5,6)$ are real hermitian.
Therefore, for any real-valued $h=(h^{(\nu)}, \dots, h^{(\nu)})$ we have
\[
\mathrm{Tr} \left\{\exp\left[- \beta H_7^{(\Lambda)}(m, h)\right] \right\}
\le
\mathrm{Tr}\left\{\exp\left[- \beta   H^{(\Lambda)}(m) \right]\right\},
\]
where we have used $\tilde{V}\exp[X] {\tilde V}^\dagger = \exp[\tilde V  X \tilde V^\dagger ]$ and defined
\[
H_7^{(\Lambda)}(m, h)
:=
H_K^{(\Lambda)}+\tilde V \tilde H_{\mathrm{int}, 7}^{(\Lambda)}(h)\tilde V^\dagger+mO^{(\Lambda)}.
\]
Similarly to the proof of (\ref{eq.IBS4}) and Lemma~\ref{lem.DCmodify}, we find
\begin{equation*}
( S^{(7)}[\tilde{h}'],  S^{(7)}[\tilde{h}'])_{\beta,m}
\le
\frac{1}{\beta g}\sum_{x \in \Lambda}|h'(x)|^2,
\end{equation*}
and there are constants $C_1$ and $C_2$ such that
\[
\left\|
\left[\left[S^{(7)}[\tilde{h}'],  H^{(\Lambda)}(m)\right],  S^{(7)}[\tilde{h}']\right] 
\right\|
\le
C_1R^{\nu-2}+C_2\vert m\vert R^\nu,
\]
where 
\[
\tilde h'(x) := h'(x) + h'(x-e_1), \quad S^{(7)}[\tilde{h}']:=\sum_{x\in\Lambda}S^{(7)}(x) \tilde h'(x).
\]
As explained above, these properties and (\ref{eq.rotS5}) imply that $\tau_{f, 0}(A_R^{(7)})$ 
and $\tau_{f, 0}(A_R^{(5)})$ give gapless excitations.

In the remainder of this section, we will prove that these six excitation states are linearly independent. 
For this purpose, we will use the unbroken symmetries of the symmetry-breaking ground state $\Omega_\omega$. 
More precisely, these are given by the particle-hole symmetry (\ref{eq.ph}) and the rotational invariance by the rotation $U(\theta)$ 
of (\ref{eq.Unitary2}) with the generators $S^{(2)}(x)$. 

As explained at the beginning of Section~\ref{sect.NG}, the excited states that correspond to the gapless modes are given by
\[
\pi_\omega\left(\tau_{f, 0}\left(A_R^{(a)}\right)\right)\Omega_\omega
=:\Phi^{(a)}\Omega_\omega
\]
for the six indices $a$. Firstly, 
by using the particle-hole transformation (\ref{eq.ph}), for any pair of $a\in \{5,7\}$ and $b\in \{1,3,4,6\}$, we have
\[
\left\langle  \tau_{f, m}^{(\Lambda)}\left(A_R^{(a)}\right)^\dagger \tau_{f, m}^{(\Lambda)}\left(A_R^{(b)}\right) \right\rangle_{\beta, m}^{(\Lambda)}
=
-\left\langle  \tau_{f, m}^{(\Lambda)}\left(A_R^{(a)}\right)^\dagger \tau_{f, m}^{(\Lambda)}\left(A_R^{(b)}\right)  \right\rangle_{\beta, m}^{(\Lambda)}
\]
Taking the limit, we obtain
\[
0=\lim_{m \searrow 0} \lim_{\Lambda \nearrow \mathbb{Z}^\nu} \lim_{\beta \to \infty}\left\langle  \tau_{f, m}^{(\Lambda)}\left(A_R^{(a)}\right)^\dagger \tau_{f, m}^{(\Lambda)}\left(A_R^{(b)}\right) \right\rangle_{\beta, m}^{(\Lambda)}
=
\left(\Phi^{(a)}\Omega_\omega, \Phi^{(b)}\Omega_\omega \right)
\]
for $a\in \{5,7\}$ and $b\in \{1,3,4,6\}$.
Hence these pairs are linearly independent.

Next, by (\ref{eq.rot1}) and (\ref{eq.rotS4}), we find
\[
U(\pi)^\dagger S^{(1)}U(\pi) =S^{(1)},
\quad
U(\pi)^\dagger S^{(4)}U(\pi) =-S^{(4)}
\]
where $U(\theta)$ is given in (\ref{eq.Unitary2}).
This shows
\[
\left\langle  \tau_{f, m}^{(\Lambda)}\left(A_R^{(1)}\right)^\dagger \tau_{f, m}^{(\Lambda)}\left(A_R^{(4)}\right) \right\rangle_{\beta, m}^{(\Lambda)}
=
-\left\langle  \tau_{f, m}^{(\Lambda)}\left(A_R^{(1)}\right)^\dagger \tau_{f, m}^{(\Lambda)}\left(A_R^{(4)}\right)  \right\rangle_{\beta, m}^{(\Lambda)}
\]
and hence
\[
\left(\Phi^{(1)}\Omega_\omega, \Phi^{(4)}\Omega_\omega \right)=0.
\]
Similarly, we have
\[
U\left(\frac{\pi}{2}\right)^\dagger S^{(1)}U\left(\frac{\pi}{2}\right) =-S^{(1)},
\quad
U\left(\frac{\pi}{2}\right)^\dagger S^{(6)}U\left(\frac{\pi}{2}\right)=S^{(4)},
\]
which implies
\[
\left\langle  \tau_{f, m}^{(\Lambda)}\left(A_R^{(1)}\right)^\dagger \tau_{f, m}^{(\Lambda)}\left(A_R^{(6)}\right) \right\rangle_{\beta, m}^{(\Lambda)}
=
-\left\langle  \tau_{f, m}^{(\Lambda)}\left(A_R^{(1)}\right)^\dagger \tau_{f, m}^{(\Lambda)}\left(A_R^{(4)}\right)  \right\rangle_{\beta, m}^{(\Lambda)}=0.
\]
Hence 
\[
\left(\Phi^{(1)}\Omega_\omega, \Phi^{(6)}\Omega_\omega \right)=0.
\]
Then, taking $\theta = \pm \pi/4$, we have
\begin{align*}
0&=
\left\langle  \tau_{f, m}^{(\Lambda)}\left(A_R^{(1)}\right)^\dagger \tau_{f, m}^{(\Lambda)}\left(A_R^{(6)}\right) \right\rangle_{\beta, m}^{(\Lambda)}\\
&=
\frac{1}{\sqrt{2}}
\left\langle  \tau_{f, m}^{(\Lambda)}\left(A_R^{(3)}\right)^\dagger \left[  \tau_{f, m}^{(\Lambda)}\left(A_R^{(4)}\right) +\tau_{f, m}^{(\Lambda)}\left(A_R^{(6)}\right)\right] \right\rangle_{\beta, m}^{(\Lambda)}
\\
&=
-\frac{1}{\sqrt{2}}\left\langle  \tau_{f, m}^{(\Lambda)}\left(A_R^{(3)}\right)^\dagger \left[  \tau_{f, m}^{(\Lambda)}\left(A_R^{(4)}\right) -\tau_{f, m}^{(\Lambda)}\left(A_R^{(6)}\right)\right] \right\rangle_{\beta, m}^{(\Lambda)}.
\end{align*}
This yields
\[
\left(\Phi^{(3)}\Omega_\omega, \Phi^{(4)}\Omega_\omega+\Phi^{(6)}\Omega_\omega \right)
=
-\left(\Phi^{(3)}\Omega_\omega, \Phi^{(4)}\Omega_\omega-\Phi^{(6)}\Omega_\omega \right),
\]
which shows
\[
\left(\Phi^{(3)}\Omega_\omega, \Phi^{(4)}\Omega_\omega\right)=0=(\Phi^{(3)}\Omega_\omega, \Phi^{(6)}\Omega_\omega ).
\]
Collectively, we find that $\Phi^{(a)}\Omega_\omega$ and $\Phi^{(b)}\Omega_\omega$ are linearly independent for any $a\in \{1,3\}$ and $b\in \{4,6\}$.

Furthermore, we will prove that $\Phi^{(1)}\Omega_\omega$ and $\Phi^{(3)}\Omega_\omega$ are linearly independent.
For this purpose, we introduce ladder operators by
\[
A_\pm := \frac{1}{\sqrt 2}\left(A_R^{(3)} \pm iA_R^{(1)} \right)
\]
and define
\[
\Phi_\pm
:=
\pi_\omega\left(\tau_{f, 0}\left(A_\pm\right)\right).
\]
Now we claim that
\begin{equation}
\label{eq.nonvanishupdown}
\left\| \Phi_\pm \Omega_\omega\right\| >0,
\end{equation}
which yields the desired result.
Indeed, by the anti-commutation relation, we learn
\begin{equation}
\label{eq.updownApm}
\left[Q^{(2)}, A_+\right]= A_+, \quad \left[Q^{(2)}, A_-\right]= -A_-.
\end{equation}
Since $[H^{(\Lambda)}(m), Q^{(2)}]=0$, we can choose the ground-state vectors 
$\{\Psi_j^{(\Lambda)}\}_{j=1}^d$ of the Hamiltonian $H^{(\Lambda)}(m)$ so that they are also the eigenvactors of $Q^{(2)}$. Namely, they satisfy $Q^{(2)}\Psi_j^{(\Lambda)} = q_j \Psi_j^{(\Lambda)}$ with the eigenvalue $q_j$.
We deduce from (\ref{eq.updownApm}) that for any~$j$
\[
Q^{(2)} A_+\Psi_j^{(\Lambda)} = \left(q_j+1 \right)A_+\Psi_j^{(\Lambda)},\quad
Q^{(2)} A_-\Psi_j^{(\Lambda)} = \left(q_j-1 \right)A_-\Psi_j^{(\Lambda)},
\]
which implies that $A_\pm\Psi_j^{(\Lambda)}$ are the eigenvectors of $Q^{(2)}$
when they are non-vanishing. We also have
$$
Q^{(2)}f(\mathcal{H}^{(\Lambda)}(m))A_\pm\Psi_j^{(\Lambda)} = (q_j\pm 1)f(\mathcal{H}^{(\Lambda)}(m))A_\pm\Psi_j^{(\Lambda)}. 
$$ 
This implies 
$$
\langle f(\mathcal{H}^{(\Lambda)}(m))A_+ \Psi_j^{(\Lambda)}, f(\mathcal{H}^{(\Lambda)}(m))A_-\Psi_j^{(\Lambda)} \rangle = 0.
$$ 
Combining this with (\ref{alphatmA}) and (\ref{eq.timelocal}), we obtain 
\begin{align*}
\left\langle  \tau_{f, m}^{(\Lambda)}\left(A_+\right)^\dagger \tau_{f, m}^{(\Lambda)}\left(A_-\right) \right\rangle_{\beta, m}^{(\Lambda)}
&=
\frac{1}{d}\sum_{j=1}^d\left\langle A_+ \Psi_j^{(\Lambda)},  f(\mathcal{H}^{(\Lambda)}(m))^2 A_- \Psi_j^{(\Lambda)} \right\rangle\\
&=0.
\end{align*}
Thus, the two vectors $\Phi_+ \Omega_\omega$ and $\Phi_- \Omega_\omega$ are orthogonal with each other.
Then, together with $\sqrt 2\Phi_\pm \Omega_\omega = (\Phi^{(3)} \pm i \Phi^{(1)})\Omega_\omega$, the positivity (\ref{eq.nonvanishupdown}) immediately implies that $\Phi^{(1)}\Omega_\omega$ and $\Phi^{(3)}\Omega_\omega$ are linearly independent.

We now turn to the proof of (\ref{eq.nonvanishupdown}).
To repeat the argument in Sec.~\ref{sect.NG}, let $\sqrt 2 Q_\pm(x) := S^{(3)}(x) \pm iS^{(1)}(x)$ and write
\[
A_\pm = \frac{1}{\vert \Omega_R\vert}\sum_{x \in \Omega_R}(-1)^{x^{(1)}+\cdots +x^{(\nu)}}Q_\pm(x)
\]
for short. We will use an upper bound for the staggered magnetization.
Since $[Q_\pm(x), S^{(3)}(x)]=\pm 2S^{(2)}(x)$, we have
\begin{align*}
\left[S^{(3)}[\tilde h'], A_\pm \right]
&=
\frac{2}{\vert \Omega_R\vert} \sum_{x \in \Omega_R}(-1)^{x^{(1)}+\cdots +x^{(\nu)}} \left[S^{(3)}(x), Q_\pm(x)\right]\\
&=\mp
\frac{4}{\vert \Omega_R\vert} \sum_{x \in \Omega_R}(-1)^{x^{(1)}+\cdots +x^{(\nu)}}S^{(2)}(x).
\end{align*}
Therefore, the staggered magnetization $m_s^{(\Lambda)}(m)$ is written as
\begin{equation}
\label{eq.staggeredcommu}
\begin{split}
	\mp 4 m_s^{(\Lambda)}(m)
	&=
	\omega_m^{(\Lambda)}\left(\left[S^{(3)}[\tilde h'], A_\pm \right]\right)\\
	&=\omega_m^{(\Lambda)}\left(S^{(3)}[\tilde h']A_\pm \right)
	-
	\omega_m^{(\Lambda)}\left(A_\pm S^{(3)}[\tilde h']\right).
\end{split}
\end{equation}

Next, by Lemma~\ref{lem.A}, we note that
\begin{gather*}
U^{(3)}\left(\frac{\pi}{2}\right)^\dagger S^{(1)}(x)U^{(3)}\left(\frac{\pi}{2}\right)
=-S^{(1)}(x),
\, U^{(3)}\left(\frac{\pi}{2}\right)^\dagger S^{(2)}(x)U^{(3)}\left(\frac{\pi}{2}\right)
=-S^{(2)}(x) 
\end{gather*}
where $U^{(a)}(\theta)$ is given in (\ref{eq.globalrot}).
Using this and $[H^{(\Lambda)}(0), U^{(a)}(\theta)]=0$, we obtain that
\begin{equation}
\label{AS3h'exp}
\begin{split}
	\left\langle A_\pm S^{(3)}[\tilde h']\right\rangle_{\beta, m}^{(\Lambda)}
	&=
	\frac{1}{Z_{\beta, m}^{(\Lambda)}}\mathrm{Tr} \left(U^{(3)}\left(\frac{\pi}{2}\right)^\dagger A_\pm S^{(3)}[\tilde h']e^{-\beta H^{(\Lambda)}(m)} U^{(3)}\left(\frac{\pi}{2}\right)\right)\\
	&=\frac{1}{Z_{\beta, m}^{(\Lambda)}}
	\mathrm{Tr} \left( A_\mp S^{(3)}[\tilde h']e^{-\beta (H^{(\Lambda)}(0)-mO^{(\Lambda)})} \right).
\end{split}
\end{equation}

By the definition of the phase factors $\theta_j(x)$ on the hopping term, the free part of the Hamiltonian $H^{(\Lambda)}_K$ is translation invariant under the lattice shift in the $x^{(\nu)}$-direction 
except for the boundary condition. Write $\tilde T$ for the unite lattice shift which is defined by $\tilde T \psi_i(x)=\psi_i(x+e_\nu)$. 
Since the shift operator $\tilde T$ is realized by a unitary transformation on the fermion Fock space, 
one has $\tilde T \psi_i^\dagger(x)=\psi_i^\dagger(x+e_\nu)$.
% and $\tilde T h'(x)=h'(x+e_1)$.
Then, from the periodic boundary condition of the lattice, the equation  (\ref{AS3h'exp}) is written 
\[
\left\langle A_\pm S^{(3)}[\tilde h']\right\rangle_{\beta, m}^{(\Lambda)}
=
\frac{1}{Z_{\beta, m}^{(\Lambda)}}
\mathrm{Tr} \left(\tilde T\left(A_\mp S^{(3)}[\tilde h']\right)e^{-\beta (\tilde T H^{(\Lambda)}(0)+mO^{(\Lambda)})} \right),
\]
where we have used
\begin{align*}
\sum_{x \in \Lambda}(-1)^{\sum_{\mu=1}^\nu x^{(\mu)}}S^{(2)}(x + e_\nu)
&=
\sum_{x \in \Lambda_\mathrm{odd} }S^{(2)}(x)
-\sum_{x \in \Lambda \backslash \Lambda_\mathrm{odd}}S^{(2)}(x)\\
&=-O^{(\Lambda)}.
\end{align*}
Besides, using (\ref{eq.BC}) and (\ref{eq.BCtr}), the boundary condition of $\tilde T H^{(\Lambda)}_K$ can be changed by the unitary transformation $U_{\mathrm{BC}, \nu}(L \to L+1)$.
Hence we have
\[
\left\langle A_\pm S^{(3)}[\tilde h']\right\rangle_{\beta, m}^{(\Lambda)}=
\left\langle \tilde T\left(A_\mp S^{(3)}[\tilde h']\right)\right\rangle_{\beta, m}^{(\Lambda)}.
\]
In order to rewrite the operator $\tilde T(A_\mp S^{(3)}[\tilde h'])$ in the right-hand side, let 
\[
\tilde \Omega_R:= \left\{x \in \Lambda \colon -R+2 \le x^{(\nu)} \le R+1, \, -R+1 \le x^{(\mu)} \le R, \, \mu =1, 2,\dots, \nu-1 \right\}.
\]
Then, we obtain that 
\begin{align*}
&\sum_{x \in \Omega_R}(-1)^{x^{(1)}+\cdots +x^{(\nu)}}
\left\langle Q_\mp(x+e_1)  \tilde T(S^{(3)}[\tilde h'])\right\rangle_{\beta, m}^{(\Lambda)}\\
&=
\sum_{x \in \tilde \Omega_R \cap \Lambda_\mathrm{odd}}
\left\langle Q_\mp(x)  \tilde T(S^{(3)}[\tilde h'])\right\rangle_{\beta, m}^{(\Lambda)}
-   \sum_{x \in \tilde \Omega_R \cap (\Lambda \backslash \Lambda_\mathrm{odd})}
\left\langle Q_\mp(x)  \tilde T(S^{(3)}[\tilde h'])\right\rangle_{\beta, m}^{(\Lambda)}\\
&=-\sum_{x \in \tilde \Omega_R}(-1)^{x^{(1)}+\cdots +x^{(\nu)}} \left\langle 
Q_\mp(x)  \tilde T(S^{(3)}[\tilde h'])\right\rangle_{\beta, m}^{(\Lambda)}.
\end{align*}
Hence, the above equation (\ref{eq.staggeredcommu}) becomes
\begin{equation*}
\begin{split}
	\mp 4 m_s^{(\Lambda)}(m)
	&=
	\omega_m^{(\Lambda)}\left(S^{(3)}[\tilde h']A_\pm \right)
	+\omega_m^{(\Lambda)}\left(\tilde A_\mp  \tilde T(S^{(3)}[\tilde h'])\right),
\end{split}
\end{equation*}
where
\[
\tilde A_\mp := \frac{1}{\vert \Omega_R\vert}\sum_{x \in \tilde \Omega_R}(-1)^{x^{(1)}+\cdots +x^{(\nu)}}
Q_\mp(x).
\]
We write 
\[
D_\mp:=\tilde A_\mp - A_\mp .
\]
%Then, one has $\| D_\mp \| \le {\rm Const.}R^{-1}$. 
Note that 
\begin{eqnarray*}
\tilde T(S^{(3)}[\tilde h'])&=&\sum_{x\in\Lambda}S^{(3)}(x+e_\nu)\tilde{h}'(x)\\
&=&\sum_{x\in\Lambda}S^{(3)}(x)\tilde{h}'(x-e_\nu)\\
&=&\sum_{x\in\Lambda}S^{(3)}(x){\tilde h}'_{\rm shift}(x)=S^{(3)}[{\tilde h}'_{\rm shift}],
\end{eqnarray*}
where we have written ${\tilde h}'_{\rm shift}(x):=\tilde{h}'(x-e_\nu)$. 
From these observations, we have
\begin{eqnarray*}
4\left\vert m_s^{(\Lambda)}(m)\right\vert
&\le&
\left\vert\omega_m^{(\Lambda)}\left(S^{(3)}[\tilde h']A_\pm \right)\right\vert
+\left\vert\omega_m^{(\Lambda)}\left({\tilde A}_\mp S^{(3)}[\tilde h'_{\rm shift}]\right)\right\vert\\
&\le& \left\vert\omega_m^{(\Lambda)}\left(S^{(3)}[\tilde h']A_\pm \right)\right\vert
+\left\vert\omega_m^{(\Lambda)}\left(A_\mp S^{(3)}[\tilde h'_{\rm shift}]\right)\right\vert\\
&\quad&+\left\vert\omega_m^{(\Lambda)}\left(D_\mp S^{(3)}[\tilde h'_{\rm shift}]\right)\right\vert.
\end{eqnarray*}
Clearly, we have 
\begin{equation}
\begin{split}
	\label{eq.Staggeredbound}
	\frac{16}{3}\left\vert m_s^{(\Lambda)}(m)\right\vert^2
	&\le \left\vert\omega_m^{(\Lambda)}\left(S^{(3)}[\tilde h']A_\pm \right)\right\vert^2
	+\left\vert\omega_m^{(\Lambda)}\left(A_\mp S^{(3)}[\tilde h'_{\rm shift}]\right)\right\vert^2\\
	&\quad+\left\vert\omega_m^{(\Lambda)}\left(D_\mp S^{(3)}[\tilde h'_{\rm shift}]\right)\right\vert^2.
\end{split}
\end{equation}

In order to estimate the right-hand side of (\ref{eq.Staggeredbound}),
we introduce the spectral decomposition $1 = P[0, 2\delta) +P[2\delta, +\infty)$, where $PI$ denotes the spectral projection onto the interval $I$ for Hamiltonian $\mathcal{H}^{(\Lambda)}(m)$, and $\delta$ is defined in the support of $f$: $\mathrm{supp} \, f \subset (\delta, 2r)$.
From now on, we also assume $r>1$. We write $P_0$ for the projection onto the sector of the ground states. 

Consider the first term in the right-hand side of (\ref{eq.Staggeredbound}),
Using $Q_\pm(x) = Q_\mp^\dagger$ and the Cauchy--Schwarz inequality, we have
\[
\left\vert\omega_m^{(\Lambda)}\left(S^{(3)}[\tilde h']P_0  A_\pm \right)\right\vert^2
\le
\omega_m^{(\Lambda)}\left(A_\mp  A_\pm \right)
\omega_m^{(\Lambda)}\left(S^{(3)}[\tilde h']  P_0S^{(3)}[\tilde h'] \right).
\]
Here the right-hand side must be zero since the expectation value of $S^{(3)}[\tilde h']$  in the ground states vanishes (see the last line of Appendix~\ref{Appendix:eq.DuhamelEq}).
Therefore, by $P[0, 2\delta) = P_0+P(0, 2\delta)$, we deduce from  (\ref{eq.GSIB})  that
\begin{equation}
\label{Eq.SApm}
\begin{split}
	&\left\vert\omega_m^{(\Lambda)}\left(S^{(3)}[\tilde h']P[0, 2\delta)  A_\pm \right)\right\vert^2\\
	&\le
	\omega_m^{(\Lambda)}\left(A_\mp \mathcal{H}^{(\Lambda)}(m)P(0, 2\delta)  A_\pm \right)
	\omega_m^{(\Lambda)}\left(S^{(3)}[\tilde h'] \frac{ P(0, 2\delta)  }{\mathcal{H}^{(\Lambda)}(m)} S^{(3)}[\tilde h'] \right)
	\\
	&\le 2\delta\,
	\omega_m^{(\Lambda)}\left(A_\mp P(0, 2\delta)  A_\pm \right)
	\omega_m^{(\Lambda)}\left(S^{(3)}[\tilde h'] \frac{1-P_0 }{\mathcal{H}^{(\Lambda)}(m)} S^{(3)}[\tilde h'] \right)\\
	&\le \mathcal{K}_1 \delta \cdot R^{\nu}.
\end{split}
\end{equation}
where $\mathcal{K}_1$ is the positive constant. 

Similarly, using $f(s) = s^{\epsilon/2}$ for $2\delta \le s \le r$, we have
\begin{equation}
\label{Eq.SApm2}
\begin{split}
	&\left\vert\omega_m^{(\Lambda)}\left(S^{(3)}[\tilde h']P[ 2\delta, +\infty)  A_\pm \right)\right\vert^2\\
	&\le
	\omega_m^{(\Lambda)}\left(A_\mp P[ 2\delta, +\infty)   A_\pm \right)
	\omega_m^{(\Lambda)}\left(S^{(3)}[\tilde h']P[ 2\delta, +\infty)   S^{(3)}[\tilde h'] \right)\\
	&\le
	\left[\delta^{-\epsilon}\omega_m^{(\Lambda)}\left(A_\mp f\left(\mathcal{H}^{(\Lambda)}(m) \right)^2  A_\pm \right)
	+\omega_m^{(\Lambda)}\left(A_\mp P[r, +\infty)   A_\pm \right) \right]\\
	&\quad \times
	\omega_m^{(\Lambda)}\left(S^{(3)}[\tilde h'](1- P_0)   S^{(3)}[\tilde h'] \right).
\end{split}
\end{equation}
In addition, by (\ref{eq.numfirst}), we obtain
\begin{align*}
\omega_m^{(\Lambda)}\left(A_\mp P[r, +\infty)   A_\pm \right)
&\le
r^{-1}\omega_m^{(\Lambda)}\left(A_\mp \mathcal{H}^{(\Lambda)}(m)   A_\pm \right) \\
&\le
3\left[\omega_m^{(\Lambda)}\left(A_R^{(1)} \mathcal{H}^{(\Lambda)}(m)   A_R^{(1)} \right) 
+\omega_m^{(\Lambda)}\left(A_R^{(3)} \mathcal{H}^{(\Lambda)}(m)   A_R^{(3)} \right)\right]\\
&\le \frac{\mathcal{K}_2}{ R^\nu},
\end{align*}
where $\mathcal{K}_2$ is the positive constant and we have used the Cauchy--Schwarz inequality for $\omega_m^{(\Lambda)}(A_R^{(1)}  \mathcal{H}^{(\Lambda)}(m)A_R^{(3)} )$ in the second inequality.
Moreover, the Cauchy--Schwarz inequality and (\ref{eq.GSIB}) imply that
\begin{align*}
\omega_m^{(\Lambda)}\left(S^{(3)}[\tilde h'](1- P_0)   S^{(3)}[\tilde h'] \right)^2
&\le
\omega_m^{(\Lambda)}\left(S^{(3)}[\tilde h']\mathcal{H}^{(\Lambda)}(m)    S^{(3)}[\tilde h'] \right)\\
&\quad \times\omega_m^{(\Lambda)}\left(S^{(3)}[\tilde h']\mathcal{H}^{(\Lambda)}(m) ^{-1} (1- P_0)   S^{(3)}[\tilde h'] \right)\\
&\le \frac{(4R)^\nu}{g}\omega_m^{(\Lambda)}\left(S^{(3)}[\tilde h']\mathcal{H}^{(\Lambda)}(m)    S^{(3)}[\tilde h'] \right).
\end{align*}
Furthermore, we obtain from Lemma~\ref{lem.DC} that
\begin{align*}
2\omega_m^{(\Lambda)}\left(S^{(3)}[\tilde h']\mathcal{H}^{(\Lambda)}(m)   S^{(3)}[\tilde h'] \right)
&=\omega_m^{(\Lambda)}\left( \left[S^{(3)}[\tilde h'], \left[ H^{(\Lambda)}(m),   S^{(3)}[\tilde h'] \right]\right]\right)
\\
&\le \mathcal{K}_3\left(R^{\nu-2}+\vert m\vert R^\nu\right)
\end{align*}
with the positive constant $\mathcal{K}_3$. 
Hence (\ref{Eq.SApm2}) becomes
\begin{equation}
\label{Eq.SApm3P}
\begin{split}
	\left\vert\omega_m^{(\Lambda)}\left(S^{(3)}[\tilde h']P[ 2\delta, +\infty)  A_\pm \right)\right\vert^2
	&\le
	\mathcal{K}_4\left[\delta^{-\epsilon}\omega_m^{(\Lambda)}\left(A_\mp f\left(\mathcal{H}^{(\Lambda)}(m) \right)^2  A_\pm \right)
	+\frac{\mathcal{K}_2}{ R^\nu}\right]\\
	&\quad \times R^\nu
	\sqrt{R^{-2}+\vert m\vert}.
\end{split}
\end{equation}
with the positive constant $\mathcal{K}_4$. 
Combining this with (\ref{Eq.SApm}), we obtain 
\begin{equation*}
\label{1term}
\begin{split}
	&\left\vert\omega_m^{(\Lambda)}\left(S^{(3)}[\tilde h']A_\pm \right)\right\vert^2\\
	&\le 2\mathcal{K}_1\delta\cdot R^\nu
	+2 \mathcal{K}_4\left[\delta^{-\epsilon}\omega_m^{(\Lambda)}\left(A_\mp f\left(\mathcal{H}^{(\Lambda)}(m) \right)^2  A_\pm \right)
	+\frac{\mathcal{K}_2}{ R^\nu}\right] R^\nu \sqrt{R^{-2}+\vert m\vert}\\
\end{split}
\end{equation*}
for the first term in the right-hand side of (\ref{eq.Staggeredbound}).

In the same way, one can treat the second term $\omega_m^{(\Lambda)}( {A}_\mp  S^{(3)}[\tilde h'_{\rm shift}])$ 
in the right-hand side of (\ref{eq.Staggeredbound}). 
Namely, the above arguments give the same bound for the quantity $\omega_m^{(\Lambda)}( {A}_\mp  S^{(3)}[\tilde h'_{\rm shift}])$.

Finally, we evaluate the third term in the right-hand side of (\ref{eq.Staggeredbound}).  
We will use the spectral projection again, and will not consider the contribution from the sector of the ground states 
because of the same reason as above. 
To begin with, we write  
\[
D^{(a)}_R
:=
\frac{1}{\vert \Omega_R\vert}\sum_{x \in \Omega_{R} \backslash \tilde \Omega_{R}}(-1)^{x^{(1)}+\cdots +x^{(\nu)}}
S^{(a)}(x),
\]
and note that
\begin{align*}
2\omega_m^{(\Lambda)}\left(D^{(a)}_R  \mathcal{H}^{(\Lambda)}(m) D^{(a)}_R\right)
&=
\omega_m^{(\Lambda)}\left(\left[ D^{(a)}_R, \left[ H^{(\Lambda)}(m), D^{(a)}_R\right]\right]\right) \\
&\le
{\rm Const.}\frac{R^{\nu-1}}{R^{2\nu}} 
= 
\frac{{\rm Const.}}{R^{\nu+1}}.
\end{align*}
Moreover, by the Cauchy--Schwarz inequality,
\begin{align*}
\omega_m^{(\Lambda)}\left(D_\mp \mathcal{H}^{(\Lambda)}(m)  D_\pm \right)
&\le
3	\left[\omega_m^{(\Lambda)}\left(D_R^{(1)} \mathcal{H}^{(\Lambda)}(m)   D_R^{(1)} \right) 
+\omega_m^{(\Lambda)}\left(D_R^{(3)} \mathcal{H}^{(\Lambda)}(m)   D_R^{(3)} \right)\right]\\
&\le 		\frac{{\rm Const.}}{R^{\nu+1}}.
\end{align*}
This yields 
\begin{equation}
\label{omegaDP01S3hbound}
\begin{split}
	&\left\vert\omega_m^{(\Lambda)}\left(D_\mp P(0, 1) S^{(3)}[{\tilde h}'_{\rm shift}]  \right)\right\vert^2\\
	&\le
	\omega_m^{(\Lambda)}\left(D_\mp \mathcal{H}^{(\Lambda)}(m)P(0, 1)  D_\pm \right)
	\omega_m^{(\Lambda)}\left(S^{(3)}[{\tilde h}'_{\rm shift}] \frac{ P(0, 1)  }{\mathcal{H}^{(\Lambda)}(m)} S^{(3)}[{\tilde h}'_{\rm shift}] \right)
	\\
	&\le 
	\omega_m^{(\Lambda)}\left(D_\mp \mathcal{H}^{(\Lambda}(m)  D_\pm \right)
	\omega_m^{(\Lambda)}\left(S^{(3)}[{\tilde h}'_{\rm shift}] \frac{1-P_0 }{\mathcal{H}^{(\Lambda)}(m)} S^{(3)}[{\tilde h}'_{\rm shift}] \right)\\
	&\le
	\frac{\mathcal{K}_5}{R},
\end{split}
\end{equation}
where $\mathcal{K}_5$ is the positive constant, and we have used
\begin{equation}
\label{omegaS3hHinvhS3h}
\begin{split}
	\omega_m^{(\Lambda)}(S^{(3)}[{\tilde h}'_{\rm shift}] \mathcal{H}^{(\Lambda)}(m)^{-1}(1-P_0)S^{(3)}[{\tilde h}'_{\rm shift}])
	&\le
	\frac{1}{g}\sum_{x \in \Lambda}\vert{\tilde h}'_{\rm shift}(x))\vert^2\\
	&\le {\rm Const.}\frac{R^\nu}{g}.
\end{split}
\end{equation}
Similarly,
\begin{equation}
\label{Eq.SApm2P}
\begin{split}
	&\left\vert\omega_m^{(\Lambda)}\left(D_\mp P[1, +\infty))S^{(3)}[{\tilde h}'_{\rm shift}]\right)\right\vert^2\\
	&\le
	\omega_m^{(\Lambda)}\left(D_\mp P[1, +\infty)   D_\pm \right)
	\omega_m^{(\Lambda)}\left(S^{(3)}[{\tilde h}'_{\rm shift}](1- P_0)   S^{(3)}[{\tilde h}'_{\rm shift}] \right).
\end{split}
\end{equation}
In order to estimate the right-hand side, we use the following inequalities: 
\begin{equation}
\label{omegaDPDbound}
\begin{split}
	\omega_m^{(\Lambda)}\left(D_\mp P[1, +\infty)   D_\pm \right)
	&\le
	\omega_m^{(\Lambda)}\left(D_\mp \mathcal{H}^{(\Lambda)}(m)   D_\pm \right) \\
	&\le \mathrm{Const.}\,\frac{1}{ R^{\nu+1}},
\end{split}
\end{equation}
and 
\begin{align*}
&\omega_m^{(\Lambda)}\left(S^{(3)}[{\tilde h}'_{\rm shift}](1- P_0)   S^{(3)}[{\tilde h}'_{\rm shift}] \right)^2\\
&\le
\omega_m^{(\Lambda)}\left(S^{(3)}[{\tilde h}'_{\rm shift}] \mathcal{H}^{(\Lambda)}(m)    S^{(3)}[{\tilde h}'_{\rm shift}] \right) \\
&\quad \times
\omega_m^{(\Lambda)}\left(S^{(3)}[{\tilde h}'_{\rm shift}] \mathcal{H}^{(\Lambda)}(m)^{-1}  (1- P_0)  S^{(3)}[{\tilde h}'_{\rm shift}] \right) \\
&\le
{\rm Const.}R^\nu	\omega_m^{(\Lambda)}\left(S^{(3)}[{\tilde h}'_{\rm shift}] \mathcal{H}^{(\Lambda)}(m)    S^{(3)}[{\tilde h}'_{\rm shift}] \right),
\end{align*}
where we have used the above bound (\ref{omegaS3hHinvhS3h}) for getting the second inequality. 
Further, in the same way as in Lemma 9.2, the right-hand side can be estimated as 
\begin{align*}
&2\omega_m^{(\Lambda)}\left(S^{(3)}[{\tilde h}'_{\rm shift}] \mathcal{H}^{(\Lambda)}(m) S^{(3)}[{\tilde h}'_{\rm shift}]\right)\\
&=
\omega_m^{(\Lambda)}\left(\left[ S^{(3)}[{\tilde h}'_{\rm shift}], \left[ H^{(\Lambda)}(m), S^{(3)}[{\tilde h}'_{\rm shift}]\right]\right]\right) \\
&\le
{\rm Const.}(R^{\nu-2} + \vert m\vert R^{\nu}).
\end{align*}
Hence
\begin{align*}
\omega_m^{(\Lambda)}\left(S^{(3)}[{\tilde h}'_{\rm shift}](1- P_0)   S^{(3)}[{\tilde h}'_{\rm shift}] \right)
&\le
{\rm Const.}R^\nu\sqrt{R^{-2}+\vert m\vert}.
\end{align*}
Substituting this and (\ref{omegaDPDbound}) into the right-hand side of (\ref{Eq.SApm2P}), we obtain   	
\[
\left\vert\omega_m^{(\Lambda)}\left(D_\pm P[1,+\infty)S^{(3)}[{\tilde h}'_{\rm shift}]   \right)\right\vert^2
\le
{\rm Const.}\frac{\sqrt{R^{-2}+\vert m\vert}}{R}.
\]
Further, by combining this with (\ref{omegaDP01S3hbound}), one has 
\begin{equation}
\label{3term}
\left\vert\omega_m^{(\Lambda)}\left(D_\pm S^{(3)}[{\tilde h}'_{\rm shift}]   \right)\right\vert^2
\le
\frac{2\mathcal{K}_5}{R}+	\frac{2{\mathcal{K}_6}\sqrt{R^{-2}+\vert m\vert}}{R},
\end{equation}
where $\mathcal{K}_6$ is the positive constant. 

Consequently, by substituting (\ref{1term}),  (\ref{3term}) and the bound for 
$\omega_m^{(\Lambda)}( {A}_\mp  S^{(3)}[\tilde h'_{\rm shift}])$, whose upper bound is the same as that for (\ref{1term}) as mentioned above, into the right-hand side of (\ref{eq.Staggeredbound}), we arrive at
\begin{equation*}
\begin{split}
	\frac{8}{3}\left\vert m_s^{(\Lambda)}(m) \right\vert^2
	&\le 
	2\mathcal{K}_4\left[\delta^{-\epsilon}\omega_m^{(\Lambda)}\left(A_\mp f\left(\mathcal{H}^{(\Lambda)}(m) \right)^2  A_\pm \right)
	+\frac{\mathcal{K}_2}{ R^\nu}\right]
	R^\nu\sqrt{R^{-2}+\vert m\vert} \\
	&\quad+2\mathcal{K}_1\delta \cdot R^{\nu}+\frac{\mathcal{K}_5+\mathcal{K}_6\sqrt{R^{-2}+\vert m\vert}}{R} 
\end{split}
\end{equation*}
where $\mathcal{K}_i$ are the positive constants. 
Taking the double limit $m\searrow 0$ and $\Lambda\nearrow \mathbb{Z}^\nu$, we have
\begin{equation}
\label{eq.Apmpositive}
\frac{8}{3}\left\vert m_s\right\vert^2
\le
2\mathcal{K}_4\times\omega_0\left(\tau_{f, 0} (A_\mp ) \tau_{f, 0}( A_\pm) \right)\frac{R^{\nu-1}}{\delta^\epsilon }
+2\mathcal{K}_2\delta \cdot R^{\nu}
+\frac{\mathcal{K}_7}{R}
\end{equation}
with the positive constant $\mathcal{K}_7$. 
We choose the parameter $\delta$ to satisfy $\delta < R^{-\nu-1}$. Then, one has 
\begin{equation}
\frac{8}{3}|m_s|^2-\frac{\mathcal{K}_8}{R}\le 2\mathcal{K}_4\times\omega_0\left(\tau_{f, 0} (A_\mp ) \tau_{f, 0}( A_\pm) \right)\frac{R^{\nu-1}}{\delta^\epsilon }
\end{equation}
with the positive constant $\mathcal{K}_8$. This inequality implies that for a sufficinetly large $R$, 
the value $\omega_0\left(\tau_{f, 0} (A_\mp ) \tau_{f, 0}( A_\pm) \right)$ is strictly positive.  
This is nothing but the desired result (\ref{eq.nonvanishupdown}).

The similar argument proves that the two sets $\{\Phi^{(4)}\Omega_\omega, \Phi^{(6)}\Omega_\omega\}$ and  $\{\Phi^{(5)}\Omega_\omega, \Phi^{(7)}\Omega_\omega\}$ 
are linearly independent respectively.

To summarize, we conclude that the set $\{\Phi^{(a)} \Omega_\omega\}_{a \neq 2,8}$ is linearly independent with each other. 
Consequently, we have the six Nambu--Goldstone modes, i.e., $N_\mathrm{NG}=N_\mathrm{BS}=6$.

Following through all the steps in this section with the appropriate substitutions, 
we can establish $N_\mathrm{NG}=N_\mathrm{BS}=2$ for the SU(2) NJL model. 
This completes the proof of Theorem~\ref{cor.NGn}.
\qed
%%%%%%%%%%%%%%%%%%%%%%%%%%%%%%%%%%%%%%%%%%%%%%%%%%%%%%%%%%%%%%%%%%%%%%%%%%%%%
\appendix

\section{Properties of SU(3) Gell-Mann matrices}
\label{sect.SU(3)}
We write $\lambda^{(a)}$, $a=1,2,\ldots,8$, for the SU(3) Gell-Mann matrices.   
The explicit expressions are given by 
\begin{equation*}
\lambda^{(1)}=
\begin{pmatrix}
	0 & 1 & 0 \\
	1 & 0 & 0 \\
	0 & 0 & 0 
\end{pmatrix},\quad 
\lambda^{(2)}=
\begin{pmatrix}
	0 & -i & 0 \\
	i & 0 & 0 \\
	0 & 0 & 0 
\end{pmatrix},\quad 
\lambda^{(3)}=
\begin{pmatrix}
	1 & 0 & 0 \\
	0 & -1 & 0 \\
	0 & 0 & 0
\end{pmatrix},
\end{equation*}
\begin{equation*}
\lambda^{(4)}=
\begin{pmatrix}
	0 & 0 & 1 \\
	0 & 0 & 0 \\
	1 & 0 & 0 
\end{pmatrix},\quad 
\lambda^{(5)}=
\begin{pmatrix}
	0 & 0 & -i \\
	0 & 0 & 0 \\
	i & 0 & 0 
\end{pmatrix},\quad 
\lambda^{(6)}=
\begin{pmatrix}
	0 & 0 & 0 \\
	0 & 0 & 1 \\
	0 & 1 & 0
\end{pmatrix},
\end{equation*}
\begin{equation*}
\lambda^{(7)}=
\begin{pmatrix}
	0 & 0 & 0 \\
	0 & 0 & -i \\
	0 & i & 0 
\end{pmatrix},\quad 
\lambda^{(8)}=\frac{1}{\sqrt{3}}
\begin{pmatrix}
	1 & 0 & 0 \\
	0 & 1 & 0 \\
	0 & 0 & -2 
\end{pmatrix}.
\end{equation*}
These satisfy the commutation relations, 
\begin{equation}
\label{lambdacommu}
[\lambda^{(a)},\lambda^{(b)}]=i\sum_{c=1}^8 f_{abc}\lambda^{(c)},
\end{equation}
where $f_{abc}$ are the structure constants which are totally antisymmetric. 
As is well known, for given two indices $a$ and $b$, 
the third index $c$ of the non-vanishing $f_{abc}$ are uniquely determined, and hence the above right-hand side is often written 
$if_{abc}$ without the sum about $c$.  
Their non-zero elements are $f_{123}=2$, $f_{147}=f_{246}=f_{257}=f_{345}=-f_{156}=-f_{367}=1$, and $f_{458}=f_{678}=\sqrt 3$. 

Let us check the SU(3) invariance of the Hamiltonian $H^{(\Lambda)}(0)$, i.e., $m=0$.  
We write 
\begin{equation}
S^{(a)}(x):=\Psi^\dagger(x)\lambda^{(a)}\Psi(x).
\end{equation}
Consider a transformation, 
\begin{equation}
\Psi(x)\rightarrow U(s)\Psi(x),
\end{equation}
where $U(s)=\exp[is \lambda^{(b)}]$ is a unitary transformation of SU(3) rotation with a small parameter $s\in\mathbb R$. 
Note that  
\begin{equation}
\begin{split}
	&\Psi^\dagger(x)U^\dagger(s)\lambda^{(a)}U(s)\Psi(x)\\
	&=\Psi^\dagger(x)\lambda^{(a)}\Psi(x)
	-is\Psi^\dagger(x)[\lambda^{(b)}\lambda^{(a)}-\lambda^{(a)}\lambda^{(b)}]\Psi(x)+\cdots\\
	&=\Psi^\dagger(x)\lambda^{(a)}\Psi(x)-s\sum_{c=1}^8 f_{abc}\Psi^\dagger(x)\lambda^{(c)}\Psi(x)+\cdots\\
	&=S^{(a)}(x)-s\sum_{c=1}^8 f_{abc}S^{(c)}(x)+\cdots,
\end{split}
\end{equation}
where we have used the relation (\ref{lambdacommu}). This implies that the operator $S^{(a)}(x)$ is transformed as 
\begin{equation}
S^{(a)}(x)\rightarrow S^{(a)}(x)-s\sum_{c=1}^8 f_{abc}S^{(c)}(x)+\cdots
\end{equation}
under the unitary transformation $U(s)$. Therefore, we have 
\begin{equation}
\begin{split}
	\sum_{a=1}^8 &S^{(a)}(x)S^{(a)}(y)\\
	&\rightarrow \sum_{a=1}^8 S^{(a)}(x)S^{(a)}(y)
	-s\sum_{a,c=1}^8 f_{abc}[S^{(c)}(x)S^{(a)}(y)+S^{(a)}(x)S^{(c)}(y)]+\cdots.
\end{split}
\end{equation}
By using the antisymmetricity $f_{cba}=-f_{abc}$, one can show that the second sum in the right-hand side is vanishing. 
This implies that the quantity $\sum_a S^{(a)}(x)S^{(a)}(y)$ is invariant under the SU(3) transformation $U(s)$. 
Since the interaction term in the Hamiltonian $H^{(\Lambda)}(m)$ of (\ref{HamSU(3)}) has the form 
$\sum_a S^{(a)}(x)S^{(a)}(x+e_\mu)$, the Hamiltonian $H^{(\Lambda)}(0)$ with $m=0$ is invariant under the SU(3) transformation. 
(See Appendix~\ref{Ap.U(1)} below.)

%%%%%%%%%%%%%%%%%%%%%%%%%%%%%%%%%%%%%%%%%%%%%%%%%%%

\section{Algebra of the operators $S^{(a)}(x)$} 
\label{app.alg}
In this appendix, we show that the following commutation relations are valid: 
\begin{equation}
\label{Scommu}
[S^{(a)}(x),S^{(b)}(x)]=i\sum_{c=1}^8 f_{abc}S^{(c)}(x)
\end{equation}
for $a,b=1,2,\ldots,8$ and $x\in\Lambda$. These yield 
\begin{equation}
\label{eq.B2}
\sum_{a=1}^8 [S^{(a)}(x)S^{(a)}(y),S^{(b)}(x)+S^{(b)}(y)]=0
\end{equation}
for any $b$ and $x\ne y$. Actually, 
\begin{equation}
\begin{split}
	&\sum_{a=1}^8 [S^{(a)}(x)S^{(a)}(y),S^{(b)}(x)+S^{(b)}(y)]\\
	&=\sum_{a=1}^8 [S^{(a)}(x),S^{(b)}(x)]S^{(a)}(y)+\sum_{a=1}^8 S^{(a)}(x)[S^{(a)}(y),S^{(b)}(y)]\\
	&=\sum_{a,c=1}^8 if_{abc}\{S^{(c)}(x)S^{(a)}(y)+S^{(a)}(x)S^{(c)}(y)\}=0,
\end{split}
\end{equation}
where we have also used $f_{cba}=-f_{abc}$. 

Let us show the relations (\ref{Scommu}). From the expression of $S^{(a)}(x)$, one has 
\begin{equation}
\begin{split}
	[S^{(a)}(x),S^{(b)}(x)]&=
	\sum_{i,j,k,\ell}[\psi_i^\dagger(x)\lambda_{i,j}^{(a)}\psi_j(x),\psi_k^\dagger(x)\lambda_{k,\ell}^{(b)}\psi_\ell(x)]\\
	&=\sum_{i,j,k,\ell}\lambda_{i,j}^{(a)}\lambda_{k,\ell}^{(b)}[\psi_i^\dagger(x)\psi_j(x),\psi_k^\dagger(x)\psi_\ell(x)].
\end{split}
\end{equation}
The commutator in the summand in the right-hand side can be calculated as follows: 
\begin{equation}
\label{eq.SaSb}
\begin{split}
	& \psi_i^\dagger(x)\psi_j(x)\psi_k^\dagger(x)\psi_\ell(x)-\psi_k^\dagger(x)\psi_\ell(x)\psi_i^\dagger(x)\psi_j(x)\\
	&=\psi_i^\dagger(x)\psi_j(x)\psi_k^\dagger(x)\psi_\ell(x)-\delta_{i,\ell}\psi_k^\dagger(x)\psi_j(x)
	+\psi_k^\dagger(x)\psi_i^\dagger(x)\psi_\ell(x)\psi_j(x)\\
	&=\psi_i^\dagger(x)\psi_j(x)\psi_k^\dagger(x)\psi_\ell(x)-\delta_{i,\ell}\psi_k^\dagger(x)\psi_j(x)
	+\psi_i^\dagger(x)\psi_k^\dagger(x)\psi_j(x)\psi_\ell(x)\\
	&=-\delta_{i,\ell}\psi_k^\dagger(x)\psi_j(x)+\delta_{j,k}\psi_i^\dagger(x)\psi_\ell(x). 
\end{split}
\end{equation}
Substituting this into the above right-hand side, we have 
\begin{equation}
\begin{split}
	[S^{(a)}(x),S^{(b)}(x)]&=\sum_{i,j,\ell}\lambda_{i,j}^{(a)}\lambda_{j,\ell}^{(b)}\psi_i^\dagger(x)\psi_\ell(x)
	-\sum_{i,j,k}\lambda_{i,j}^{(a)}\lambda_{k,i}^{(b)}\psi_k^\dagger(x)\psi_j(x)\\
	&=\Psi^\dagger(x)\lambda^{(a)}\lambda^{(b)}\Psi(x)-\Psi^\dagger(x)\lambda^{(b)}\lambda^{(a)}\Psi(x)\\
	&=i\sum_c f_{abc}\Psi^\dagger(x)\lambda^{(c)}\Psi(x)\\
	&=i\sum_c f_{abc}S^{(c)}(x),
\end{split} 
\end{equation}
where we have used the commutation relations (\ref{lambdacommu}). This is the desired result (\ref{Scommu}). 

%%%%%%%%%%%%%%%%%%%%%%%%%%%%%%%%%%%%%%%%%%%%%%%%%%%%%%%%%%%
\section{Rotation of the order parameters}
\label{Ap.U(1)}

In this Appendix, we show that 
\begin{equation}
\label{eq.order}
\langle S^{(a)}(x) \rangle_{\beta, m=0}^{(\Lambda)} = 0\quad  \ \mbox{for \ } a=1,2,\dots,8.
\end{equation}
This is the consequence of the SU(3) rotational symmetry of the thermal equilibrium state without 
the symmetry breaking field. 

First, we prove the following lemma.
\begin{lemma}
\label{lem.A}
Let $U_x^{(a)}(\theta)= e^{i\theta S^{(a)}(x)}$ with
a real parameter $\theta$. 
Then the following relations hold: 
\begin{align}
	& U_x^{(7)}(\theta)^\dagger S^{(1)}(x)U_x^{(7)}(\theta)
	=
	S^{(1)}(x)\cos \theta + S^{(4)}(x) \sin \theta \label{eq.rot3}, \\
	& U_x^{(2)}\left(\frac{\theta}{2}\right)S^{(1)}(x)U_x^{(2)}\left(\frac{\theta}{2}\right)^\dagger 
	=
	S^{(1)}(x)\cos \theta +S^{(3)}(x) \sin \theta  \label{eq.rot1}, \\
	& U_x^{(2)}\left(\frac{\theta}{2}\right)S^{(3)}(x)U_x^{(2)}\left(\frac{\theta}{2}\right)^\dagger 
	=
	S^{(3)}(x)\cos \theta -S^{(1)}(x) \sin \theta  \label{eq.rotS3}, \\
	& U_x^{(3)}\left(\frac{\theta}{2}\right)^\dagger S^{(1)}(x)U_x^{(3)}\left(\frac{\theta}{2}\right)
	=
	S^{(1)}(x)\cos \theta +S^{(2)}(x) \sin \theta  \label{eq.rot2}, \\
	& U_x^{(2)}(\theta) S^{(4)}(x)U_x^{(2)}(\theta)^\dagger
	=
	S^{(4)}(x)\cos \theta + S^{(6)}(x) \sin \theta  \label{eq.rotS4}\\
	& U_x^{(2)}(\theta) S^{(6)}(x)U_x^{(2)}(\theta)^\dagger
	=
	S^{(6)}(x)\cos \theta - S^{(4)}(x) \sin \theta  \label{eq.rotS6},\\
	& U_x^{(2)}(\theta) S^{(5)}(x)U_x^{(2)}(\theta)^\dagger
	=
	S^{(5)}(x)\cos \theta + S^{(7)}(x) \sin \theta  \label{eq.rotS7}\\
	& U_x^{(2)}(\theta) S^{(7)}(x)U_x^{(2)}(\theta)^\dagger
	=
	S^{(7)}(x)\cos \theta - S^{(5)}(x) \sin \theta  \label{eq.rotS5},\\
	& U_x^{(7)}(\theta)^\dagger S^{(2)}(x)U_x^{(7)}(\theta)
	=S^{(2)}(x)\cos \theta +S^{(5)}(x) \sin \theta  \label{eq.rot6}, \\
	& U_x^{(7)}(\theta)^\dagger S^{(3)}(x)U_x^{(7)}(\theta)\nonumber\\
	&=
	S^{(3)}(x)\cos^2 \theta +\frac{S^{(3)}(x) +\sqrt 3 S^{(8)}(x)}{2}\sin^2 \theta -S^{(6)}(x) \sin \theta \cos \theta,
	\label{eq.rot7}
	%\\
	%& U_x^{(5)}(\theta)^\dagger S^{(3)}(x)U_x^{(5)}(\theta)\nonumber\\
	%&=
	%S^{(3)}(x)\cos^2 \theta +\frac{S^{(3)}(x) -\sqrt 3 S^{(8)}(x)}{2}\sin^2 \theta +S^{(4)}(x) \sin \theta \cos \theta .
	%\label{eq.rot8}
\end{align}
\end{lemma}
\begin{proof}[Proof of Lemma~\ref{lem.A}]
Consider first the case of (\ref{eq.rot3}). 
We write 
$$
L(\theta)= U_x^{(7)}(\theta)^\dagger S^{(1)}(x)U_x^{(7)}(\theta)
$$ 
for the right-hand side of (\ref{eq.rot3}). Then, one has 
\[
\frac{d}{d\theta}L(\theta, x)
=-iS^{(7)}(x) L(\theta) + i L(\theta)S^{(7)}(x)
=i[L_x(\theta), S^{(7)}(x)].
\]
For the left-hand side of (\ref{eq.rot3}), we write 
$$
R(\theta):= S^{(1)}(x)\cos \theta + S^{(4)}(x) \sin \theta.  
$$
By using the commutation relation $[S^{(a)}(x), S^{(b)}(x)] = i\sum_cf_{abc}S^{(c)}(x)$, one obtains 
\[
\frac{d}{d\theta}R(\theta) =i[R(\theta), S^{(7)}(x)].
\]
Since $L(0)=S^{(1)}(x)=R(0)$, we obtain $L(\theta) =R(\theta)$ for any $\theta \in \mathbb{R}$ 
by the uniqueness of the solution for the initial value problem.
This shows (\ref{eq.rot3}).
The proofs of (\ref{eq.rot1})--(\ref{eq.rot7}) are the same.
\end{proof}

Let $U^{(a)}(\theta)$ be the global rotation
\begin{equation}
\label{eq.globalrot}
U^{(a)}(\theta):=\prod_{x \in \Lambda}U^{(a)}_x(\theta).
\end{equation}
Our Hamiltonian $H^{(\Lambda)}(0)$ is invariant under this global rotation.
Indeed, the invariance of the hopping part follows from the anti-commutation relations.
For the interaction part, we note that $[S^{(a)}(x), S^{(b)}(y)] = 0$ and, by (\ref{eq.B2}),
\[
\sum_{x\in\Lambda}[H^{(\Lambda)}_\mathrm{int},  S^{(b)}(x)]
=\frac{g}{2}\sum_{\substack{x \in \Lambda:\\ |x-y|=1}}\sum_{a=1}^8[S^{(a)}(x)S^{(a)}(y), S^{(b)}(x)+S^{(b)}(y)] =0.
\]
This implies $[U^{(b)}(\theta), H_\mathrm{int}^{(\Lambda)}]=0$ for all $b$, and hence $[U^{(b)}(\theta), H^{(\Lambda)}(0)]=0$.
Using $U^{(7)}(\pi)^\dagger S^{(1)}U^{(7)}(\pi) = -S^{(1)}$, we have
\[
\langle S^{(1)}(x) \rangle_{\beta, m=0}^{(\Lambda)} =\langle U^{(7)}(\pi)^\dagger S^{(1)}(x)U^{(7)}(\pi) \rangle_{\beta, m=0}^{(\Lambda)}  
=
-\langle S^{(1)}(x) \rangle_{\beta, m=0}^{(\Lambda)}
=0.
\]
Changing the role of the indices in the equations (\ref{eq.rot3})--(\ref{eq.rot7}), we conclude (\ref{eq.order}).\qed

The following is also shown in the above proof.
\begin{proposition}
\label{prop.rot}
Let $U^{(a)}(\theta)$ be the global rotation defined by (\ref{eq.globalrot}).
Then, for any $\theta \in \mathbb{R}$ and $a=1, \dots, 8$, it holds that
\[
[U^{(a)}(\theta), H^{(\Lambda)}(0)]=0.
\]
\end{proposition}

%%%%%%%%%%%%%%%%%%%%%%%%%%%%%%%%%%%%%%%%%%%%%%%%%%%%%%%%
\section{Symmetries}
\label{AppendixSymmetries}

In this Appendix, we prove two relations, (\ref{eq.indinv}) and (\ref{eq.S8}), below. 

\subsection{Spatial symmetry}
\label{Appendix:SpatialSymmetry}

For any $\mu$, it holds that
\begin{equation}
\label{eq.indinv}
\sum_{x \in \Lambda}\left\langle S^{(a)}(x)S^{(a)}(x+e_\mu) \right\rangle_{\beta, m}^{(\Lambda)}
=
\sum_{x \in \Lambda}\left\langle S^{(a)}(x)S^{(a)}(x+e_1) \right\rangle_{\beta, m}^{(\Lambda)}.
\end{equation}
To see this, let $P_\mu$ be the transformation of the permutation
\[
\left(x^{(1)}, x^{(2)}, \dots, x^{(\mu)} \right) \to (x^{(\mu)}, x^{(1)}, x^{(2)}, \dots, x^{(\mu-1)}).
\]
Then one has
\[
\sum_{x \in \Lambda}\left\langle S^{(a)}(x)S^{(a)}(x+e_1) \right\rangle_{\beta, m}^{(\Lambda)}
=
\frac{1}{Z_{\beta, m}^{(\Lambda)}}
\sum_{x \in \Lambda}\mathrm{Tr} \, \left[S^{(a)}(x)S^{(a)}(x+e_\mu) e^{-\beta P_\mu H^{(\Lambda)}(m)}\right].
\]
Since $U_\mathrm{HA}(j \to 1)$ given in (\ref{eq.gaugeHA}) does not change $S^{(a)}(x)$ and
\[
U_\mathrm{HA}^{-1}(j \to 1) P_\mu H^{(\Lambda)}(m) U_\mathrm{HA}(j \to 1) = H^{(\Lambda)}(m)
\]
by (\ref{eq.invHA}), we have (\ref{eq.indinv}). 

%%%%%%%%%%%%%%%%%%%%%%%%%%%%%%%%%%%%%%%%%%%%%%%%%%%%%
\subsection{Rotational symmetry} 

We show that
\begin{equation}
\label{eq.S8}
\left\langle S^{(3)}(x) S^{(3)}(x+e_1)  \right\rangle_{\beta,0} = \left\langle S^{(8)}(x) S^{(8)}(x+e_1)  \right\rangle_{\beta,0}.
\end{equation}
In order to prove this equality, we intoroduce the transformation, 
\begin{equation}
\Psi(x)\rightarrow \mathcal{U}^{(b)}(\theta)\Psi(x),
\end{equation}
for the fermion operator $\Psi(x)$ with 
\begin{equation}
\mathcal{U}^{(b)}(\theta)=e^{i\theta\lambda^{(b)}},
\end{equation}
where $\theta$ is a real parameter. The corresponding transformation is given by 
\[
U^{(b)}(\theta):=\prod_{x \in \Lambda}e^{i\theta S^{(b)}(x)}
\]
for the operators $S^{(a)}$. Namely, we write 
\begin{equation}
S_{\theta,b}^{(a)}(x):=U^{(b)}(\theta)^\dagger S^{(a)}(x)U^{(b)}(\theta).
\end{equation}
We also write
\begin{equation}
\mathcal{S}_{\theta,b}^{(a)}(x):=\Psi^\dagger(x)\mathcal{U}^{(b)}(\theta)^\dagger \lambda^{(a)}\mathcal{U}^{(b)}(\theta)\Psi(x).
\end{equation}
One can easily check that these two transformations are equivalent to each other. 

We choose $b=1,4,6$ for $\lambda^{(b)}$. For example, for $b=1$, one has 
\begin{equation}
\mathcal{U}^{(1)}(\theta)=\begin{pmatrix}
	\cos\theta & i\sin\theta & 0 \\
	i\sin\theta & \cos\theta & 0 \\
	0 & 0 & 1
\end{pmatrix}.
\end{equation}
In particular, for $\theta=\pi/2$, we have 
\begin{equation}
\mathcal{U}^{(1)}(\pi/2)=\begin{pmatrix}
	0& i & 0 \\
	i & 0 & 0 \\
	0 & 0 & 1
\end{pmatrix}.
\end{equation}
Namely, this interchanges the two components, $\psi_1(x)$ and $\psi_2(x)$, of $\Psi(x)$. 
Similarly, the rest of two, $b=4,6$, also interchange the two components of $\Psi(x)$. 
Therefore, for $x\ne y$, we have
\begin{equation}
\label{eq.ij}
\left\langle n_i(x)n_j(y) \right\rangle_{\beta,0}
=
\left\langle n_k(x)n_l(y) \right\rangle_{\beta,0} \quad \mbox{for all \ }i \neq j, k\neq l.
\end{equation}
and 
\begin{equation}
%\label{eq.2233}
\left\langle n_i(x)n_i(y) \right\rangle_{\beta,0}
=
\left\langle n_j(x)n_j(y) \right\rangle_{\beta,0}\ \mbox{for all \ }i\ne j. 
\end{equation}
Here, $n_j=\psi_j^\dagger(x)\psi_j(x)$. These imply 
\begin{equation}
\left\langle S^{(8)}(x) S^{(8)}(y) \right\rangle_{\beta,0}
=\left\langle S^{(3)}(x) S^{(3)}(y) \right\rangle_{\beta,0}.
\end{equation}
for $x\ne y$.

%%%%%%%%%%%%%%%%%%%%%%%%%%%%%%%%%%%%%%
\section{Derivation of (\ref{eq.infDuhamel})}
\label{Appendix:eq.DuhamelEq}

In order to show the bound (\ref{eq.infDuhamel}), we write
\begin{equation}
\label{eq.DuhamelEq}
\begin{split}
	\beta &(S^{(3)}[\tilde h'] , (S^{(3)}[\tilde h'] )_{\beta, m}\\
	&=
	\frac{\beta }{Z_{\beta, m}^{(\Lambda)}  }\int_0^1 ds \; e^{-s \beta E_0^{(\Lambda)}}
	\sum_{j=1}^d\left\langle \Psi_j^{(\Lambda)}, S^{(3)}[\tilde h']  
	e^{-(1-s)\beta H^{(\Lambda)}(m) } S^{(3)}[\tilde h']  \Psi_j^{(\Lambda)}\right\rangle \\
	&\quad+
	\frac{\beta}{Z_{\beta, m}^{(\Lambda)}  }\int_0^1 ds 
	\sum_{j: E_j^{(\Lambda)} \neq E_0^{(\Lambda)}}e^{-s \beta E_j^{(\Lambda)}}\left\langle \phi_{j}, S^{(3)}[\tilde h'] 
	e^{-(1-s)\beta H^{(\Lambda)}(m) } S^{(3)}[\tilde h']  \phi_{j}\right\rangle,
\end{split}
\end{equation}
where $\Psi_j^{(\Lambda)}$ are the ground states for $H^{(\Lambda)}(m)$ with the energy $E_0^{(\Lambda)}$ 
and $\phi_j$ are the excited states which satisfy $H^{(\Lambda)}(m) \phi_j = E_j^{(\Lambda)} \phi_j$.
The second term in the right-hand side vanishes as $\beta \to \infty$ as follows.
We first note that, as $\beta \to \infty$,
\begin{equation}
\label{eq.infPart}
e^{  \beta E_0^{(\Lambda)}}Z_{\beta, m}^{(\Lambda)}  
=
d + \sum_{j: E_j^{(\Lambda)} \neq E_0^{(\Lambda)}}e^{- \beta( E_j^{(\Lambda)} - E_0^{(\Lambda)})} \to d,
\end{equation}
where we have used $E_j^{(\Lambda)} - E_0^{(\Lambda)} > 0$.
Writing the second term in (\ref{eq.DuhamelEq}) as
\begin{align*}
\frac{\beta e^{- \beta E_0^{(\Lambda)}}}{Z_{\beta, m}^{(\Lambda)}  }\int_0^1 ds &
\sum_{j: E_j^{(\Lambda)} \neq E_0^{(\Lambda)}} e^{-s \beta( E_j^{(\Lambda)}-E_0^{(\Lambda)}) } \\
\times &\left\langle \phi_{j}, S^{(3)}[\tilde h']  e^{-(1-s)\beta (H^{(\Lambda)}(m) -E_0^{(\Lambda)})} S^{(3)}[\tilde h']  
\phi_{j}\right\rangle
\end{align*}
and using $H^{(\Lambda)}(m) -E_0^{(\Lambda)}\ge 0 $, one has 
\begin{align*}
\beta \sum_{j: E_j^{(\Lambda)} \neq E_0^{(\Lambda)}}&e^{-s \beta( E_j^{(\Lambda)} - E_0^{(\Lambda)})}\left\langle 
\phi_{j}, S^{(3)}[\tilde h']  e^{-(1-s)\beta (H^{(\Lambda)}(m) -E_0^{(\Lambda)})} S^{(3)}[\tilde h']  \phi_{j}\right\rangle\\
&\le
\|S^{(3)}[\tilde h']  \|^2 
\sum_{j: E_j^{(\Lambda)} \neq E_0^{(\Lambda)}}\beta e^{-s \beta (E_j^{(\Lambda)} - E_0^{(\Lambda)})} \to 0 \quad 
\ \mbox{as \ }\beta \to \infty.
\end{align*}
Therefore, the second term in the right-hand side of (\ref{eq.DuhamelEq}) goes to zero by the dominated convergence theorem 
as $\beta\to\infty$.

Furthermore, we decompose the first term in the right-hand side of (\ref{eq.DuhamelEq}) into two parts as follows: 
\begin{equation}
\label{eq.DuhFirst}
\begin{split}
	\frac{\beta e^{- \beta E_0^{(\Lambda)}}}{Z_{\beta, m}^{(\Lambda)}  }&\int_0^1 ds 
	\sum_{j=1}^d\left\langle \Psi_j^{(\Lambda)}, S^{(3)}[\tilde h']  e^{-(1-s)\beta (H^{(\Lambda)}(m) -E_0^{(\Lambda)})} 
	(1-P_0)S^{(3)}[\tilde h']  \Psi_j^{(\Lambda)}\right\rangle \\
	&+
	\frac{\beta e^{- \beta E_0^{(\Lambda)}}}{Z_{\beta, m}^{(\Lambda)}  }
	\sum_{j=1}^d\left\langle \Psi_j^{(\Lambda)}, S^{(3)}[\tilde h']   P_0 S^{(3)}[\tilde h']  \Psi_j^{(\Lambda)}\right\rangle.
\end{split}
\end{equation}
Using the operator identity
\[
\int_0^1 ds\;e^{-As}=\frac{1-e^{-A}}{A}
\]
for any hermitian matrix $A>0$, the first part can be written 
\begin{align*}
&\lim_{\beta \to \infty} 
\frac{e^{- \beta E^{(\Lambda)}(m)}}{Z_{\beta, m}^{(\Lambda)}  }
\sum_{j=1}^d\left\langle \Psi_j^{(\Lambda)}, S^{(3)}[\tilde h']   \mathcal{H}^{(\Lambda)}(m)^{-1} (1-P_0)S^{(3)}[\tilde h']  
\Psi_j^{(\Lambda)}\right\rangle\\
&=\frac{1}{d}\sum_{j=1}^d\left\langle \Psi_j^{(\Lambda)}, S^{(3)}[\tilde h']  
\mathcal{H}^{(\Lambda)}(m)^{-1}  (1-P_0)S^{(3)}[\tilde h']  \Psi_j^{(\Lambda)}\right\rangle
\end{align*}
in the limit $\beta \to \infty$, where we have used (\ref{eq.infPart}). Clearly, the second part is non-negative. 
Since it is enough to obtain the lower bound of (\ref{eq.infDuhamel}), we have 
\begin{align*}
& \lim_{\beta \to \infty} \beta (S^{(3)}[\tilde h'], S^{(3)}[\tilde h'])_\beta\\
&\ge \frac{1}{d}\sum_{j=1}^d\left\langle \Psi_j^{(\Lambda)}, S^{(3)}[\tilde h']  
\mathcal{H}^{(\Lambda)}(m)^{-1}  (1-P_0)S^{(3)}[\tilde h']  \Psi_j^{(\Lambda)}\right\rangle. 
\end{align*}
This proves (\ref{eq.infDuhamel}). 

In passing, we remark the following: 
In the same way, we notice that the second part of (\ref{eq.DuhFirst}) behaves like 
\begin{equation}
\frac{\beta}{d}  
\sum_{j=1}^d\left\langle \Psi_j^{(\Lambda)}, S^{(3)}[\tilde h']   P_0 S^{(3)}[\tilde h']  \Psi_j^{(\Lambda)}\right\rangle
\end{equation}
for a large $\beta$. As mentioned above, the summand is non-negative. 
If it is non-vanishing, then it contradicts with 
the upper bound (\ref{eq.Duhamel}) in the limit $\beta\to\infty$. Thus, it must be vanishing. 
Namely, the expectation value of $S^{(3)}[\tilde h']$ in the sector of the ground states is inevitably 
vanishing. This is stronger than the fact $\omega_m^{(\Lambda)}(S^{(3)}[\tilde h'])=0$ of (\ref{GSexpecS3}).

\subsection*{Acknowledgements} 
Y.~G. is grateful to Tetsuo Hatsuda and Masaru Hongo for informing her about the paper by Vafa and Witten~\cite{VW}, and for helpful discussions.
Partial financial support from JSPS Kakenhi Grant Number 23K12989 (Y.~G.), and the RIKEN iTHEMS Mathematical Physics Working Group (Y.~G.) is gratefully acknowledged.

%%%%%%%%%%%%%%%%%%%%%%%%%%%%%%%%%%%%%%%%%%%%%%%%%%%%%%%%%%%%%%%%%


\begin{thebibliography}{99}
\bibitem{Aoki1} S.~Aoki,
Solution to the U(1) problem on a lattice,
{\it Phys. Rev. Lett.\/} {\bf 57}, 3136 (1986).

\bibitem{Aoki2} S.~Aoki,
New phase structure for lattice QCD with Wilson fermions,
{\it Phys. Rev. D\/} {\bf 30}, 2653 (1984).

\bibitem{ABG} S. Aoki, S. Boettcher and A. Gocksch,
Spontaneous breaking of flavor symmetry and parity in the Nambu--Jona-Lasinio model with Wilson fermions, 
{\it Phys.Lett.} B {\bf 331}  157--164 (1994).

\bibitem{AG} S. Aoki and A. Gocksch, 
Spontaneous breaking of flavor symmetry and parity in lattice QCD with Wilson fermions, 
{\it Phys. Rev.} D {\bf 45}, 3845--3853 (1992). 

\bibitem{AL} I.~Affleck and E.~H.~Lieb, 
A Proof of Part of Haldane's Conjecture on Spin Chains,
{\it Lett. Math. Phys. \/} {\bf 12}, 57--69 (1986).

\bibitem{BS}
C.~Borgs and E.~Seiler,
Lattice Yang-Mills theory at nonzero temperature and the confinement problem,
{\it Commun.~Math. Phys.\/} {\bf 91}, 329--380 (1983).

\bibitem{BRI} O. Bratteli and D. W. Robinson,
Operator Algebras and Quantum Statistical Mechanics 1, Springer-Verlag, New York, Berlin Heidelberg, 2nd ed. (1987).

\bibitem{BR} O. Bratteli and D. W. Robinson,
Operator Algebras and Quantum Statistical Mechanics 2, Springer-Verlag, New York, Berlin Heidelberg, 2nd ed. (1997).

\bibitem{DLS} F. J. Dyson, E. H. Lieb and B. Simon, 
Phase Transitions in Quantum Spin Systems with Isotropic and Nonisotropic Interactions, 
{\it J. Stat. Phys.\/} {\bf 13}, 335--383 (1978).

\bibitem{FILS} J. Fr\"ohlich, R. B. Israel, E. H. Lieb and B. Simon, 
Phase Transitions and Reflection Positivity. II. Lattice Systems with Short-Range and Coulomb Interactions, 
{\it J. Stat. Phys.\/} {\bf 22}, 297--347 (1980).  

\bibitem{FSS} J. Fr\"ohlich, B. Simon, and T. Spencer, 
Infrared bounds, phase transitions and continuous symmetry breaking,  
{\it Commun. Math. Phys.\/} {\bf 50}  79--95 (1976).

\bibitem{Goldstone} J. Goldstone, 
Field Theories with $\langle\!\langle$Superconductor$\rangle\!\rangle$ Solutions, 
{\it Nuovo Cimento\/} {\bf 19} 154--164 (1961) . 

\bibitem{GK}Y. Goto and T. Koma,
Spontaneous Mass Generation and Chiral Symmetry Breaking in a Lattice Nambu--Jona-Lasinio Model,
{\it Commun. Math. Phys.\/}  {\bf 404(3)} 1463--1493(2023). 

\bibitem{GSW} J. Goldstone, A. Salam, and S. Weinberg, 
Broken Symmetries, 
{\it Phys. Rev.\/} {\bf 127} (1962) 965--970. 

\bibitem{Hatsuda} T. Hatsuda and T. Kunihiro, 
QCD Phenomenology based on a Chiral Effective Lagrangian, 
{\it Phys. Rept.\/} {\bf 247}, 221--367 (1994), arXiv:hep-ph/9401310.  

\bibitem{Hidaka}Y. Hidaka,
Counting rule for Nambu--Goldstone modes, 
{\it Phys.~Rev.~Lett.\/} {\bf 110}, 091601 (2013).


\bibitem{JP} A. Jaffe and F. L. Pedrocchi, 
Reflection Positivity for Majoranas,
{\it Ann. Henri Poincar\'e\/}, {\bf 16}(1), 189--203 (2015). 

\bibitem{KLS1}
T. Kennedy, E. H. Lieb, and B. S. Shastry,
Existence of N\'eel order in some spin-1/2 Heisenberg antiferromagnets,
{\it J. Stat. Phys. 53\/}, {\bf 1019} (1988).

\bibitem{KLS2} T. Kennedy, E. H. Lieb, and B. S. Shastry,
The XY Model Has Long-Range Order for All Spins and All Dimensions Greater than One,
{\it Phys. Rev. Lett.\/} {\bf 61}, 2582 (1988).

%\bibitem{KennedyTasaki} T. Kennedy, H. Tasaki, 
%Hidden Symmetry Breaking and the Haldane Phase in $S=1$ Quantum Spin Chains, 
%{\it Commun. Math. Phys.\/} {\bf 147}, 431--484 (1992).

\bibitem{KS} J. Kogut and L. Susskind, 
Hamiltonian formulation of Wilson's lattice gauge theories,
{\it Phys. Rev. D\/}, {\bf 11} 395--408 (1975). 

\bibitem{KT}
T.~Koma and H.~Tasaki, 
Symmetry breaking in Heisenberg antiferromagnets,
{\it Commun. Math. Phys.\/} {\bf 158}, 191--214 (1993).

\bibitem{Koma1}T.~Koma,
Spectral Gaps of Quantum Hall Systems with Interactions,
{\it J. Stat. Phys.\/} {\bf 99}, 313--381 (2000). 

\bibitem{Koma2} T.~Koma,
Dispersion Relations of Nambu--Goldstone Modes, Preprint, arXiv:2105.04970 

\bibitem{Koma3} T.~Koma,
Maximum Spontaneous Magnetization and Nambu--Goldstone Mode, Preprint, arXiv:1712.09018.

\bibitem{Koma4} T. Koma, 
Nambu--Goldstone Modes for Superconducting Lattice Fermions, Preprint,\hfill\break arXiv:2201.13135. 

\bibitem{Koma5} T. Koma, 
$\pi$ Flux Phase and Superconductivity for Lattice Fermions Coupled to Classical Gauge Fields, Preprint, arXiv:2205.00835. 

\bibitem{LM}
E.~H.~Lieb, D.~Mattis,
Ordering energy levels in interacting spin chains,
{\it J.~Math.~Phys. \/} {\bf 3}, 749--751 (1962).

\bibitem{LSM}
E.~H.~Lieb, T. Schultz, and D. Mattis,
Two soluble models of an antiferromagnetic chain, 
{\it Ann. Phys.\/} {\bf 16},
407--466 (1961).

\bibitem{Luscher}
M.~L\"uscher,
Construction of a selfadjoint, strictly positive transfer matrix for Euclidean lattice gauge theories,
{\it Commun.~Math.~Phys.\/} {\bf 54}, 283--292 (1977).

\bibitem{Marshall}
W.~Marshall,
Antiferromagnetism,
{\it Proc.~R.~Soc.~A \/} {\bf 232}, 48 (1955).

\bibitem{Momoi} T. Momoi,
An Upper Bound for the Spin-Wave Spectrum of the Heisenberg Antiferromagnet, 
{\it J. Phys. Soc. Jpn.\/} {\bf 63}, 2507--2510 (1994).

\bibitem{Momoi2} T. Momoi,
Quantum Fluctuations in Quantum Lattice Systems with Continuous Symmetry , 
{\it J. Stat. Phys.\/} {\bf 85}, 193--210 (1996).

\bibitem{Nakamura}  S.~Nakamura,
Remarks on discrete Dirac operators and their continuum limits,
{\it J. Spectr. Theory.\/} {\bf 14} 255--269 (2024). arXiv:2306.14180.

\bibitem{Nambu} Y. Nambu, 
Axial Vector Current Conservation in Weak Interactions,
{\it Phys. Rev. Lett.\/} {\bf 4} 380--382 (1960).

\bibitem{Nambu2} Y. Nambu, 
Quasiparticles and Gauge Invariance in the Theory of Superconductivity,
{\it Phys. Rev.\/} {\bf 117}  648--663 (1960).

\bibitem{NJL} Y. Nambu, and G. Jona-Lasinio, 
Dynamical Model of Elementary Particles Based on 
an Analogy with Superconductivity. I, 
{\it Phys. Rev.\/} {\bf 122}  345--358 (1961).

\bibitem{NJL2} Y. Nambu, and G. Jona-Lasinio, 
Dynamical Model of Elementary Particles Based on an Analogy with Superconductivity. II, 
{\it Phys. Rev.\/} {\bf 124}, 246--254 (1961). 

\bibitem{Neves1986LongRO}
E.~J.~Neves and J.~Perez,
Long range order in the ground state of two-dimensional antiferromagnets,
{\it Phys. Lett. A\/} {\bf 114}, 331--333 (1986).

\bibitem{OS}
K.~Osterwalder and E.~Seiler,
Gauge field theories on a lattice,
{\it Ann.~Phys.\/} {\bf 110}, 440--471	(1978).

\bibitem{Rothe} H. J. Rothe,
Lattice Gauge Theories: An Introduction (Fourth Edition),
World Scientific Lecture Notes in Physics, {\bf Vol.~82} World Scientific, (2012).


\bibitem{Ruelle}D.~Ruelle,
Statistical Mechanics: Rigorous Results, W.A. Benjamin Publ., New York (1969).

\bibitem{SS1}
M.~Salmhofer and E.~Seiler,
Proof of chiral symmetry breaking in strongly coupled lattice gauge theory,
{\it Commun. Math. Phys.\/} {\bf 139}, 395--431 (1991).


\bibitem{SS2}
M.~Salmhofer and E.~Seiler,
Proof of chiral symmetry breaking in lattice gauge theory,
{\it Lett. Math. Phys.\/} {\bf 21}, 13--21 (1991).

\bibitem{Seiler}
E. Seiler, 
Gauge Theories as a Problem of Constructive Field Theory and Statistical Physics,
Lecture Notes in Physics, {\bf Vol.~159}, Springer Verlag, (1982).

\bibitem{Susskind} L. Susskind, 
Lattice fermions, 
{\it Phys. Rev. D\/} {\bf 16}, 3031--3039 (1977). 

\bibitem{Tasaki} H.~Tasaki,
Physics and Mathematics of Quantum Many-Body Systems, Graduate Texts in Physics, Springer, (2020).

\bibitem{WM} H. Watanabe and H. Murayama 
Unified Description of Nambu--Goldstone Bosons without Lorentz Invariance,
{\it Phys. Rev. Lett.\/} {\bf 108}, 251602 (2012)

\bibitem{VW}
C.~Vafa and E.~Witten,
Restrictions on symmetry breaking in vector-like gauge theories,
{\it Nucl. Phys. B\/} {\bf 234},173--188 (1984).

\end{thebibliography}
\end{document}